\documentclass[aps,pra,letterpaper,notitlepage,longbibliography,reprint,superscriptaddress,onecolumn]{revtex4-1}
\usepackage[left=18mm,right=18mm,top=22mm,bottom=27mm]{geometry}

\input{./sections/preamble}

\allowdisplaybreaks

\begin{document}
\title{Boundary topological entanglement entropy in two and three dimensions}

\author{Jacob C. Bridgeman}

\email{jcbridgeman1@gmail.com}
\homepage{jcbridgeman.bitbucket.io}
\orcid{0000-0002-5638-6681}

\affiliation{Perimeter Institute for Theoretical Physics, Waterloo, Ontario, Canada}

\author{Benjamin J. Brown}
\affiliation{Centre for Engineered Quantum Systems, School of Physics, The University of Sydney, Sydney, NSW 2006, Australia}

\author{Samuel J. Elman}
\affiliation{Centre for Engineered Quantum Systems, School of Physics, The University of Sydney, Sydney, NSW 2006, Australia}
\affiliation{Department of Physics, Imperial College London, London, SW7 2AZ, United Kingdom}
\affiliation{School of Physics and Astronomy, University of Leeds, Leeds LS2 9JT, United Kingdom}

\commoncomment{Authors listed alphabetically}

\begin{abstract}
	The topological entanglement entropy is used to measure long-range quantum correlations in the ground space of topological phases.
	Here we obtain closed form expressions for the topological entropy of (2+1)- and (3+1)-dimensional loop gas models, both in the bulk and at their boundaries, in terms of the data of their input fusion categories and algebra objects. Central to the formulation of our results are generalized $\S$-matrices.
	We conjecture a general property of these $\S$-matrices, with proofs provided in many special cases. This includes constructive proofs for categories up to rank 5.
\end{abstract}

\maketitle

% !TeX spellcheck = en_US
% This line sets the project root file.
% !TEX root = ../WW_EntanglementEntropy.tex
%

\section{Introduction}\label{sec:introduction}

The classification of topological phases is fundamental to the study of modern condensed matter physics~\cite{haldane1984periodic,Wen1989vacuum, wen1990ground, wen2004quantum}. Moreover, they have properties that may be valuable for the robust storage and manipulation of quantum information~\cite{kitaev1997fault, Brown2016Quantum}. Their characteristics include a stable gap at zero temperature and quasiparticle excitations with non-trivial braid statistics~\cite{wilczek1982magnetic,wilczek1982quantum}.

An important class of topological phases are represented by topological loop-gas models~\cite{levin2005string,Walker2012TQFT}. These models can be defined in terms of an input unitary fusion category, and their ground states by superpositions of string diagrams labeled by objects from the category. The categorical framework provides a collection of local relations that ensure topological invariance of the ground states. In (2+1)-dimensions, these models are called Levin-Wen models~\cite{levin2005string}. Levin-Wen models have point-like excitations, commonly called anyons, with non-trivial fusion rules and braid statistics.
In (3+1)-dimensions, the input category must be equipped with a premodular braiding, leading to a Walker-Wang model~\cite{Walker2012TQFT}. Generically, Walker-Wang models support point-like and loop-like excitations. In contrast to Levin-Wen models, the excitations in the bulk of a Walker-Wang model may be trivial, specifically if the input category is modular.

Loop-gas models can be defined on manifolds with boundaries by modifying the local relations governing the strings in the vicinity of the boundary. One way to define a boundary to a topological loop-gas is to allow some strings to terminate on the boundary. This is captured in the current work using particular objects called algebras~\cite{1706.03329,1706.00650}. Despite their trivial bulk excitations, Walker-Wang models may have highly non-trivial boundary excitations.

Intimately connected to the topological properties of these phases is the long-range entanglement present in the ground state of the Hamiltonians describing these phases~\cite{chen2010local, wang2017twisted}. The long-range quantum correlations found in the ground states of topological phases can be measured using the topological entanglement entropy~\cite{hamma2005bipartite, kitaev2006topological, levin2006detecting}. We typically expect that the entanglement entropy shared between two subsystems of the ground state of a gapped many body system to respect an area law~\cite{eisert2010area}. However, supposing a sensible choice of bipartition, the entanglement entropy of the ground state of topological phases has a constant universal correction~\cite{hamma2005bipartite}. In (2+1)-dimensions, it is known that this correction relates to the total quantum dimension of the quasiparticle excitations supported by the phase~\cite{kitaev2006topological, levin2006detecting}. We can also evaluate the quantum dimensions of individual excitations~\cite{dong2008topological} and defects~\cite{brown2013topological, Bonderson2017Anyonic} of a phase using topological entanglement entropy. Other work has shown we can use the topological entanglement entropy to calculate the fusion rules~\cite{Shi2020} and braid statistics~\cite{zhang2012quasi} of (2+1)-dimensional phases.

Generalizations of topological entanglement entropy diagnostics have been found~\cite{castelnovo2008topological, grover2011entanglement} for (3+1)-dimensional phases with bulk topological order. These diagnostics were first demonstrated using the (3+1)-dimensional toric code model~\cite{hamma2005string} as an example. This phase gives rise to one species of bosonic excitation that braids non-trivially with a loop-like excitation in the bulk of the system.
In contrast, particular classes of Walker-Wang models~\cite{Walker2012TQFT} have been shown to behave differently using the same diagnostics. Modular examples of these models demonstrate zero bulk topological entanglement entropy~\cite{von2013three, bullivant2016entropic}, even though, at their boundary, they realize quasiparticle excitations with non-trivial braid statistics~\cite{von2013three}.

In \onlinecite{kim2015ground}, two new diagnostics were found to interrogate the long-range entanglement at the boundary of a (3+1)-dimensional-dimensional topological phase. The behavior of the diagnostics was determined by making quite generic considerations of the support of creation operators for topological excitations, without assuming any knowledge of the underlying particle theory of the phase. It was shown that the diagnostics will show a null outcome only if all the particles that can be created at the boundary have trivial braid statistics. Conversely, boundary topological order must necessarily show positive topological entanglement entropy if quasi-particles that demonstrate non-trivial braid statistics can be created. In that work, the diagnostics were tested at the different boundaries of the (3+1)-dimensional-dimensional toric code where null outcomes were obtained at boundaries where the appropriate types of particles condense. However, a limitation of the diagnostics presented in that paper is that the meaning of a positive outcome is not well understood.

From the input fusion category perspective, the topological entanglement entropies can be understood as arising from constraints on the `string flux' passing through a surface. In (3+1)-dimensions, there are also additional corrections due to braiding. Allowing strings to terminate in the vicinity of a physical boundary alters the flux (and braiding) constraints in the vicinity, thereby altering the topological entropy.

In this work, we obtain closed form expressions for bulk and boundary topological entanglement entropy diagnostics for topological loop-gas models. We obtain our results by evaluating the entanglement entropy of various regions of ground states of Levin-Wen and Walker-Wang models. This requires careful analysis of various string diagrams, such as generalized $\S$-matrices which encode the braiding properties of the input category. Additionally, we examine how the inclusion of boundaries, via algebra objects, alter these diagrams, and so the topological entropy. In all cases, we find that the entropy can be expressed in terms of the quantum dimension of the input category and the quantum dimension of the algebra object. In the bulk of (3+1)-dimensional models, we conjecture, and prove in many cases, that the entropy is the logarithm of the total quantum dimension of the particle content of the theory, extending the results of \onlinecite{bullivant2016entropic}.

\subsubsection*{Overview}

Following a brief summary of our results, the remainder of this paper is structured as follows.
In \cref{sec:preliminaries}, we provide some mathematical definitions and minor results that are required for the remainder of the paper. In \cref{sec:models}, we briefly review the models of interest, and discuss the class of boundaries we consider.
In \cref{sec:entropydiagnostics}, we explain the origin and meaning of topological entanglement entropy, and define the diagnostics used to detect boundary topological entanglement entropy.
In \cref{sec:bulkentropy}, we compute the entropy of bulk regions for Levin-Wen models. These computations are required for the Walker-Wang models, and provide a good warm-up. We then discuss the additional considerations for Walker-Wang models, and extend the computations to these.
In \cref{sec:boundaryentropy}, we compute the boundary entropy diagnostics for Levin-Wen models with boundary, followed by some classes of Walker-Wang models with boundary.
We summarize in \cref{sec:remarks}.

We include two appendices. In \cref{app:FC_pfs}, we provide proofs of some lemmas concerning (generalized) $\S$-matrices.
In \cref{app:SN_results}, we provide proofs of some results concerning loop-gas models, and their entropies.

\subsection{Summary of results}

\begin{table}\renewcommand{\arraystretch}{1.2}\setlength{\tabcolsep}{15pt}
	\centering
	\begin{threeparttable}
		\begin{tabular}{!{\vrule width 1pt}>{\columncolor[gray]{.9}[\tabcolsep]}c!{\vrule width 1pt}c!{\vrule width 1pt}c!{\vrule width 1pt}c!{\vrule width 1pt}}
			\toprule[1pt]
			\rowcolor[gray]{.9}[\tabcolsep]
			Model                         & Bulk strings                      & Boundary algebra object                                        & Topological entropy                                \\
			\toprule[1pt]
			                              &                                   &                                                                & $\LWTEE=\log\D^2_{\drinfeld{\C}}$                  \\
			\greycline{3-4}
			\multirow{-2}{*}{Levin-Wen}   & \multirow{-2}{*}{$\C$ fusion}     & $A$                                                            & $\LWbndTEE=\log\D^2$                               \\
			\toprule[1pt]
			                              &                                   &                                                                & $\WWTEE=\log\D^2_{\mug{\C}}$\tnote{a}              \\
			\greycline{3-4}
			                              &                                   &                                                                & $\WWbndTEEPoint=?$\tnote{b}                        \\
			                              &                                   & \multirow{-2}{*}{$A$}                                          & $\WWbndTEELoop=\log d_A^2-\log\D^2+\WWbndTEEPoint$ \\
			\greycline{3-4}
			                              &                                   &                                                                & $\WWbndTEEPoint=\log\D^2$                          \\
			\multirow{-5}{*}{Walker-Wang} & \multirow{-5}{*}{$\C$ premodular} & \multirow{-2}{*}{$A=1$}                                        & $\WWbndTEELoop=0$                                  \\
			\toprule[1pt]
			                              &                                   &                                                                & $\WWTEE=\log\D^2_{\mug{\C}}$                       \\
			\greycline{3-4}
			                              &                                   &                                                                & $\WWbndTEEPoint=\log\D^2-\log d_A$                 \\
			\multirow{-3}{*}{Walker-Wang} & \multirow{-3}{*}{$\C$ symmetric}  & \multirow{-2}{*}{$A$}                                          & $\WWbndTEELoop=\log d_A$                           \\
			\toprule[1pt]
			                              &                                   &                                                                & $\WWTEE=\log\D^2_{\mug{\C}}$                       \\
			\greycline{3-4}
			                              &                                   &                                                                & $\WWbndTEEPoint=\log\D^2-2\log d_A$                \\

			                              &                                   & \multirow{-2}{*}{$A$ such that $A\cap\mug{\C}=\{1\}$\tnote{c}} & $\WWbndTEELoop=0$                                  \\
			\greycline{3-4}
			                              &                                   &                                                                & $\WWbndTEEPoint=\log\D^2-\log d_A$                 \\
			\multirow{-5}{*}{Walker-Wang} & \multirow{-5}{*}{$\C$ pointed}    & \multirow{-2}{*}{$A$ such that $A\cap\mug{\C}=A$}              & $\WWbndTEELoop=\log d_A$                           \\
			\toprule[1pt]
		\end{tabular}
		\footnotesize
		\begin{tablenotes}
			\item[a] Conjectured, proven in many cases as indicated in \cref{thm:WWentropyexamples}.
			\item[b] We do not have a general form at present.
			\item[c] Includes the case $\C$ modular.
		\end{tablenotes}
	\end{threeparttable}
	\caption{
		Summary of results, technical terms defined in \cref{sec:preliminaries}. The bulk strings are labeled by a unitary fusion category $\C$, possibly with extra structure. The number $\D$ denotes the total quantum dimension of $\C$, $A$ is an algebra object (with extra structure, see \cref{sec:preliminaries}) of dimension $d_A$, $\drinfeld{\C}$ and $\mug{\C}$ are the Drinfeld and M\"uger centers of $\C$ respectively. Topological entanglement entropies for Levin-Wen models are denoted $\LWTEE$ and $\LWbndTEE$ for the bulk and near the boundary respectively. These quantities are defined in \cref{eqn:LWent}, \cref{eqn:LWbndent} and \cref{fig:LWregions}. The corresponding quantities in (3+1)-dimensions are denoted $\WWTEE$ for the bulk entropy (\cref{eqn:WWent} and \cref{fig:WWregionsblk}), and $\WWbndTEEPoint$ ($\WWbndTEELoop$) for the boundary entropy detecting point-like (loop-like) excitations as defined in \cref{eqn:WWptdef} and \cref{fig:WWregionsbnd_pt} (\cref{eqn:WWloopdef} and \cref{fig:WWregionsbnd_lp}).
	}
	\label{tab:resultsummary}
\end{table}

In \cref{tab:resultsummary}, we summarize our main results. The models we discuss will be introduced in the following sections, followed by the proofs of these results. We note that many of these results were previously known, for example the bulk Levin-Wen appears in \onlinecites{kitaev2006topological,levin2006detecting}. When the Levin-Wen model is defined by the fusion category $\vvec{G}$, the boundary Levin-Wen results appear in \onlinecite{1801.01519}. The bulk results for symmetric and modular Walker-Wang models appear in \onlinecite{bullivant2016entropic}. We extend this to include all pointed inputs (all quantum dimensions equal to 1), as well as all input categories up to rank 5. This allows us to conjecture a general result. To the best of our knowledge, there are no results concerning boundary entropies of Walker-Wang models beyond the (3+1)-dimensional toric code~\cite{kim2015ground}.
% !TeX spellcheck = en_US
% This line sets the project root file.
% !TEX root = ../WW_EntanglementEntropy.tex
%

\section{Preliminaries}\label{sec:preliminaries}

Each of the physical models of interest in this work is defined by a collection of algebraic data. In the case of the (2+1)-dimensional models, this data can be conveniently packaged into a \emph{fusion category}. In one higher dimension, additional data is required, so the package is a \emph{premodular category}. Boundaries of these models can be specified using particular objects, called \emph{algebra objects} in the input category.

In this section, we provide some (standard) definitions of these constructions, and introduce the notation we will use in the remainder of the manuscript. Many of the definitions that appear in this section are adapted from \onlinecite{0111204,0804.3587}. Readers familiar with braided fusion categories and their algebra objects can skip to \cref{sec:models}.

Throughout this work, we find it helpful to define a generalized Kronecker delta function. For this purpose, we use the \define{indicator function} $\indicator{X}=1\iff X=\mathrm{true}$, and zero otherwise.

\begin{definition}[Unitary fusion category]\label{def:UFC}
	We sketch the definition of a unitary (skeletal) fusion category. For a more complete definition, we refer to \onlinecite{MR3242743}, or for the physically minded reader \onlinecites{kitaev2006anyons,Bondersonthesis,Beer2018}. For our purposes, a unitary fusion category $\C$ consists of:
	\begin{itemize}
		\item A finite set of simple objects $\{1,a_1,a_2,\ldots a_k\}$, where $1$ is the distinguished object called the unit, and $k+1=\rk{\cat{C}}$ is called the \define{rank} of the category.
		\item For each triple of simple objects $a,b,c\in \C$, a finite dimensional vector space $\C(a\otimes b,c)$, called a fusion space. The dimension of $\C(a\otimes b,c)$ defines the integer $N_{ab}^c$. For any simple objects $a$ and $b$, fusion with the unit obeys $N_{1,a}^{b}=N_{a,1}^{b}=\indicator{a=b}$.
		\item Associator isomorphisms $(a\otimes b)\otimes c\cong a\otimes(b\otimes c)$
		\item To every object $a\in \C$, a dual object $\dual{a}$ so that $\C(a\otimes\dual{a},1)\cong\mathbb{C}\cong\C(\dual{a}\otimes a,1)$.
	\end{itemize}

	We refer to the integers $N_{ab}^c$ as fusion rules. If all $N_{a,b}^c\in\{0,1\}$, we call the category \define{multiplicity free}.
	The unique positive solution to the set of equations
	\begin{align}
		d_ad_b & =\sum_c N_{ab}^c d_c
	\end{align}
	defines the \define{quantum dimensions} $d_a$ of the simple objects. The \define{total quantum dimension} of $\C$ is defined as
	\begin{align}
		\mathcal{D}^2:=\sum_c d_c^2.
	\end{align}
	If all $d_a=1$ (equivalently $\D^2=\rk{\C}$), the category is called \define{pointed}. We remark that in the physics literature, this property is commonly called Abelian.

	It is common to use a diagrammatic calculus of string diagrams to discuss fusion categories. Strings are labeled by objects from the category.
	If, in particular, a string is labeled by a simple object, it cannot change its labeling to another simple object except at a vertex, since there are no morphisms between distinct simple objects. Frequently, we neglect to draw strings labeled by the unit object for simplicity.

	In string diagrams, the quantum dimensions are assigned to loops
	\begin{align}
		\begin{array}{c}
			\includeTikz{bubbleA}{
				\begin{tikzpicture}[scale=.5];
					\draw (0,0) circle (.75);
					\node[anchor=east] at (-.75,0) {$a$};
				\end{tikzpicture}
			}
		\end{array}
		 & =d_a,\label{eqn:bubble}
	\end{align}
	and we refer to the insertion or removal of a loop as a \define{loop move}.

	Once we choose a basis for $\C(a\otimes b,c)$, we denote a basis vector $\mu\in\C(a\otimes b,c)$ using a trivalent vertex\footnote{We refrain from drawing arrows on the diagrams, instead using the convention that all lines are oriented upwards.}
	\begin{align}
		\begin{array}{c}
			\includeTikz{trivalent1}{
				\begin{tikzpicture}[scale=.5]
					\pgfmathsetmacro{\s}{sqrt(2)}
					\draw (0,0)--(1,1) node[below,pos=0,inner sep=.1]{\strut$a$} (1,1)--(2,0)node[below,pos=1,inner sep=.1]{\strut$b$} (1,1)--(1,1+\s)node[above,pos=1,inner sep=.1]{\strut$c$};
					\node[left] at (1,1) {\strut$\mu$};
				\end{tikzpicture}
			}
		\end{array}.
	\end{align}
	With these bases fixed, the associators can be expressed as unitary matrices
	\begin{align}
		\begin{array}{c}
			\includeTikz{FLHS}{
				\begin{tikzpicture}[scale=.5,yscale=-1];
					\pgfmathsetmacro{\s}{sqrt(2)}
					\draw (0,0)--(0,1)--(-1,2)--(-2,3) (-1,2)--(0,3) (0,1)--(2,3);
					\node[anchor=south,inner sep=.1] at (0,0) {\strut$d$};
					\node[anchor=north,inner sep=.1] at (-2,3) {\strut$a$};
					\node[anchor=north,inner sep=.1] at (0,3) {\strut$b$};
					\node[anchor=north,inner sep=.1] at (2,3) {\strut$c$};
					\node[anchor=south east,inner sep=.1] at (-.5,1.5) {\strut$e$};
					\node[anchor=west] at (0,1) {\strut$\alpha$};
					\node[anchor=west] at (-1,2) {\strut$\beta$};
				\end{tikzpicture}
			}
		\end{array}
		 & =
		\sum_{(\mu,f,\nu)}\bigg[F^d_{abc}\bigg]_{(\alpha,e,\beta), (\mu,f,\nu)}
		\begin{array}{c}
			\includeTikz{FRHS}{
				\begin{tikzpicture}[scale=.5,xscale=-1,yscale=-1];
					\pgfmathsetmacro{\s}{sqrt(2)}
					\draw (0,0)--(0,1)--(-1,2)--(-2,3) (-1,2)--(0,3) (0,1)--(2,3);
					\node[anchor=south,inner sep=.1] at (0,0) {\strut$d$};
					\node[anchor=north,inner sep=.1] at (-2,3) {\strut$c$};
					\node[anchor=north,inner sep=.1] at (0,3) {\strut$b$};
					\node[anchor=north,inner sep=.1] at (2,3) {\strut$a$};
					\node[anchor=south west,inner sep=.1] at (-.5,1.5) {\strut$f$};
					\node[anchor=east] at (0,1) {\strut$\mu$};
					\node[anchor=east] at (-1,2) {\strut$\nu$};
				\end{tikzpicture}
			}
		\end{array},
	\end{align}
	with $F^d_{abc}$ a unitary matrix for each valid choice of labels. This re-association is commonly referred to as an \define{$F$-move}. These matrices must obey the pentagon equations
	\begin{align}
		\sum_\delta \bigg[F^e_{fcd}\bigg]_{(\alpha,g,\beta), (\delta,x,\rho)}\bigg[F^e_{abx}\bigg]_{(\delta,f,\gamma), (\mu,y,\nu)} & =
		\sum_{\substack{(\sigma,z,\tau)                                                                                                 \\\epsilon}}
		\bigg[F^g_{abc}\bigg]_{(\beta,f,\gamma), (\sigma,z,\tau)}\bigg[F^e_{azd}\bigg]_{(\alpha,g,\sigma), (\mu,y,\epsilon)}
		\bigg[F^y_{bcd}\bigg]_{(\epsilon,z,\tau), (\nu,x,\rho)}.
	\end{align}
	Additionally, the unit obeys the triangle equation~\cite{MR3242743}, however we always, without loss of generality, choose the unit to be strict.

	The category $\C$ is equipped with a dagger structure, so we also have dual spaces to each fusion space. These are represented using upside-down vertices. We normalize the basis for $\C(a\otimes b,c)$ and $\C(c,a\otimes b)$ so that
	\begin{align}
		\begin{array}{c}
			\includeTikz{TVV_norm_LHS}{
				\begin{tikzpicture}[scale=.5]
					\pgfmathsetmacro{\s}{sqrt(2)}
					\begin{scope}[yscale=-1]
						\draw (1,1)--(0,0) node[pos=1,left,inner sep=.5] {\strut$a$} node[pos=0,left]{\strut$\mu$} (1,1)--(1,1+\s)node[pos=.9,right,inner sep=.5]{\strut$c$} (1,1)--(2,0) node[pos=1,right,inner sep=.5] {\strut$b$};
					\end{scope}
					\draw (1,1)--(0,0)  node[pos=0,left]{\strut$\nu$} (1,1)--(1,1+\s)node[pos=.9,right,inner sep=.5]{\strut$d$} (1,1)--(2,0);
				\end{tikzpicture}
			}
		\end{array}
		 & =\indicator{c=d}\indicator{\mu=\nu}
		\sqrt{\frac{d_ad_b}{d_c}}
		\begin{array}{c}
			\includeTikz{TVV_norm_RHS}{
				\begin{tikzpicture}[scale=.5]
					\pgfmathsetmacro{\s}{sqrt(2)}
					\draw (0,-1-\s)--(0,1+\s) node[right,pos=.5,inner sep=.5]{\strut$c$};
				\end{tikzpicture}
			}
		\end{array},
		\\
		\begin{array}{c}
			\includeTikz{identitynormLHS}{
				\begin{tikzpicture}[scale=.5]
					\pgfmathsetmacro{\s}{sqrt(2)}
					\draw (0,-1-\s)--(0,1+\s) node[left,pos=0,inner sep=.5]{\strut$a$};\node[left,inner sep=.5] at (0,1+\s){\phantom{\strut$a$}};
					\draw (1,-1-\s)--(1,1+\s) node[right,pos=0,inner sep=.5]{\strut$b$};\node[right,inner sep=.5] at (1,1+\s){\phantom{\strut$b$}};
				\end{tikzpicture}
			}
		\end{array}
		 & =\sum_{c,\mu}
		\sqrt{\frac{d_c}{d_ad_b}}
		\begin{array}{c}
			\includeTikz{identitynormRHS}{
				\begin{tikzpicture}[scale=.5]
					\pgfmathsetmacro{\s}{sqrt(2)};
					\draw (0,-1-\s)--(.5,-1)--(.5,1)--(0,1+\s) ;
					\draw (1,-1-\s)--(.5,-1)--(.5,1)--(1,1+\s);
					\node[left,inner sep=.5] at (0,-1-\s){\strut$a$};\node[left,inner sep=.5] at (0,1+\s){\strut$a$};
					\node[right,inner sep=.5] at (1,-1-\s){\strut$b$};\node[right,inner sep=.5] at (1,1+\s){\strut$b$};
					\node[right,inner sep=.5] at (.5,0){\strut$c$};
					\node[left] at (.5,-1){\strut$\mu$};\node[left] at (.5,1){\strut$\mu$};
				\end{tikzpicture}
			}
		\end{array}.
	\end{align}
	Unitary fusion categories can be equipped with a unique unitary \define{spherical structure}~\cite{Etingof2005}. This allows all diagrams to behave as though they were drawn on the surface of a sphere, for example
	\begin{align}
		\begin{array}{c}
			\includeTikz{sphericalLHS}{
				\begin{tikzpicture}[scale=1]
					\draw (0,.2) to[out=90,in=0] (-.25,.35)to[out=180,in=90] (-.5,0) to[out=270,in=180](-.25,-.35) to[out=0,in=270] (0,-.2);
					\filldraw[fill=white] (0,0) circle (.2);\node at (0,0) {$x$};
				\end{tikzpicture}
			}
		\end{array}
		 & =
		\begin{array}{c}
			\includeTikz{sphericalRHS}{
				\begin{tikzpicture}[scale=1]
					\draw (0,.2) to[out=90,in=180] (.25,.35)to[out=0,in=90] (.5,0) to[out=270,in=0](.25,-.35) to[out=180,in=270] (0,-.2);
					\filldraw[fill=white] (0,0) circle (.2);\node at (0,0) {$x$};
				\end{tikzpicture}
			}
		\end{array},\label{eqn:spherical}
	\end{align}
	where $x$ is any subdiagram.

	We refer to \onlinecites{kitaev2006anyons,Bondersonthesis} for a more detailed overview of these diagrams.
\end{definition}

\begin{definition}[Unitary braided fusion category]\label{def:BFC}
	Given a unitary fusion category $\C$, a braiding is a map $a\otimes b\cong b\otimes a$, which is compatible with the associator of $\C$. Graphically, the braiding is encoded in the unitary $R$ matrices
	\begin{align}
		\begin{array}{c}
			\includeTikz{braid_LHS}{
				\begin{tikzpicture}[scale=.4]
					\pgfmathsetmacro{\s}{sqrt(2)}
					\draw[draw=white,double=black,ultra thick] (2,-2) to [out=135,in=225] (0,0)--(1,1)--(2,0)  to [out=315,in=45] (0,-2);
					\draw[draw=white,double=black,ultra thick](2,0)  to [out=315,in=45] (0,-2);
					\draw [thick](1,1)--(1,1+\s) node[pos=1,above,inner sep=.1] {\strut$c$};
					\node[below,inner sep=.1] at (0,-2) {\strut$a$};\node[below,inner sep=.1] at (2,-2) {\strut$b$};
					\node[left] at (1,1) {\strut$\mu$};
				\end{tikzpicture}
			}
		\end{array}
		 & =\sum_\nu\bigg[R_{ab}^c\bigg]_{\mu,\nu}
		\begin{array}{c}
			\includeTikz{braid_RHS}{
				\begin{tikzpicture}[scale=.4]
					\pgfmathsetmacro{\s}{sqrt(2)}
					\draw	[thick](1,1)--(1,1+\s)node[pos=1,above,inner sep=.1] {\strut$c$};;
					\draw	[thick] (1,1)--(0,-2);
					\draw	[thick] (1,1)--(2,-2);
					\node[below,inner sep=.1] at (0,-2) {\strut$a$};\node[below,inner sep=.1] at (2,-2) {\strut$b$};
					\node[left] at (1,1) {\strut$\nu$};
				\end{tikzpicture}
			}
		\end{array}.
	\end{align}
	Compatibility with the (given) $F$ matrices is encoded in the hexagon equations
	\begin{align}
		\sum_{\gamma,\delta}
		\bigg[R_{ac}^e\bigg]_{\beta,\gamma}
		\bigg[F^d_{acb}\bigg]_{(\alpha,e,\gamma), (\sigma,f,\delta)}
		\bigg[R_{bc}^f\bigg]_{\delta,\tau}
		 & =
		\sum_{\substack{(\epsilon,g,\rho)                     \\\mu}}
		\bigg[F^d_{cab}\bigg]_{(\alpha,e,\beta), (\epsilon,g,\rho)}
		\bigg[R_{gc}^d\bigg]_{\epsilon,\mu}
		\bigg[F^d_{abc}\bigg]_{(\mu,g,\rho), (\sigma,f,\tau)} \\
		\sum_{\gamma,\delta}
		\bigg[R_{ca}^e\bigg]_{\gamma,\beta}^*
		\bigg[F^d_{acb}\bigg]_{(\alpha,e,\gamma), (\sigma,f,\delta)}
		\bigg[R_{cb}^f\bigg]_{\tau,\delta}^*
		 & =
		\sum_{\substack{(\epsilon,g,\rho)                     \\\mu}}
		\bigg[F^d_{cab}\bigg]_{(\alpha,e,\beta), (\epsilon,g,\rho)}
		\bigg[R_{cg}^d\bigg]_{\mu,\epsilon}^*
		\bigg[F^d_{abc}\bigg]_{(\mu,g,\rho), (\sigma,f,\tau)}.
	\end{align}
\end{definition}

\begin{definition}[Premodular category]\label{def:Premodular}
	A unitary braided fusion category can be equipped with a unique unitary ribbon structure~\cite{Galindo2014}. This is defined by a set of \define{twists}
	\begin{align}
		\begin{array}{c}
			\includeTikz{thetaop1}{
				\begin{tikzpicture}[scale=.5]
					\pgfmathsetmacro{\s}{sqrt(2)};
					\draw[draw=white,double=black,ultra thick](0,0)to[out=180,in=270](-.5,1);
					\draw[draw=white,double=black,ultra thick](-.5,-1)to[out=90,in=180](0,.5)to[out=0,in=0](0,0);
					\node[anchor=north,inner sep=.1] at (-.5,-1) {\strut$a$};
				\end{tikzpicture}
			}
		\end{array}
		 & =\theta_a
		\begin{array}{c}
			\includeTikz{thetaop2}{
				\begin{tikzpicture}[scale=.5]
					\pgfmathsetmacro{\s}{sqrt(2)};
					\draw[draw=white,double=black,ultra thick](-.5,-1)--(-.5,1);
					\node[anchor=north,inner sep=.1] at (-.5,-1) {\strut$a$};
				\end{tikzpicture}
			}
		\end{array},
	\end{align}
	obeying the ribbon equations
	\begin{align}
		\sum_\nu \left[R^c_{ba}\right]_{\mu,\nu}\left[R^c_{ab}\right]_{\nu,\rho} & =\frac{\theta_c}{\theta_a\theta_b}\indicator{\mu=\rho}.
	\end{align}
	When $\C$ is equipped with this structure, we call $\C$ premodular.

	In a premodular category, we can define the $\S$-matrix
	\begin{align}
		\S_{a,b}:=\frac{1}{\D}\tr B_{a,\dual{b}}
		 & =\frac{1}{\D}
		\begin{array}{c}
			\includeTikz{Smatrixdef}{
				\begin{tikzpicture}[scale=.75]
					\centerarc[draw=white,double=black,ultra thick](-.75,0)(0:180:1);
					\draw[draw=white,double=black,ultra thick] (.75,0) circle (1);
					\centerarc[draw=white,double=black,ultra thick](-.75,0)(180:360:1);% Syntax: [draw options] (center) (initial angle:final angle:radius)
					\node[anchor=west,inner sep=.5] at(1.75,0) {\strut$b$};
					\node[anchor=west,inner sep=.5] at(.25,0) {\strut$a$};
					\node[anchor=east,inner sep=.5] at(-.25,0) {\strut$\dual{b}$};
					\node[anchor=east,inner sep=.5] at(-1.75,0) {\strut$\dual{a}$};
				\end{tikzpicture}
			}
		\end{array}        \label{eqn:Smatrix}                                  \\
		 & =\frac{1}{\D}\sum_x N_{a,\dual{b}}^x \frac{\theta_x}{\theta_a\theta_{\dual{b}}}d_x,
	\end{align}
	where
	\begin{align}
		B_{a,b}:=
		\begin{array}{c}
			\includeTikz{RoperatorLHS}{
				\begin{tikzpicture}[scale=.75]
					\draw[draw=white,double=black,ultra thick] (.5,0) to [out=90,in=270] (-.5,1);
					\draw[draw=white,double=black,ultra thick] (-.5,0) to [out=90,in=270] (.5,1) to[out=90,in=270] (-.5,2);
					\draw[draw=white,double=black,ultra thick] (-.5,1) to [out=90,in=270] (.5,2);
					\node[anchor=north,inner sep=.1] at (-.5,0) {\strut$a$};\node[anchor=north,inner sep=.1] at (.5,0) {\strut$b$};
					\node[anchor=south,inner sep=.1] at (-.5,2) {\phantom{\strut$a$}};\node[anchor=south,inner sep=.1] at (.5,2) {\phantom{\strut$b$}};
				\end{tikzpicture}
			}
		\end{array}
		 & =
		\sum_{x,\mu}\sqrt{\frac{d_x}{d_ad_b}}\frac{\theta_x}{\theta_a\theta_b}
		\begin{array}{c}
			\includeTikz{RoperatorRHS}{
				\begin{tikzpicture}[scale=.75]
					\draw (-.5,0)--(0,.5) (.5,0)--(0,.5)--(0,1.5)--(-.5,2) (0,1.5)--(.5,2);
					\node[anchor=north,inner sep=.1] at (-.5,0) {\strut$a$};\node[anchor=north,inner sep=.1] at (.5,0) {\strut$b$};
					\node[anchor=south,inner sep=.1] at (-.5,2) {\strut$a$};\node[anchor=south,inner sep=.1] at (.5,2) {\strut$b$};
					\node[anchor=west,inner sep=.1] at (0,1) {$x$};
					\node[left] at (0,.5) {\strut$\mu$};
					\node[left] at (0,1.5) {\strut$\mu$};
				\end{tikzpicture}
			}
		\end{array}.
	\end{align}
\end{definition}

For the results in this manuscript, it is important to understand which strings can be `uncrossed'. This is captured by the M\"uger center.

\begin{definition}[M\"uger center~\cite{0804.3587}]\label{def:mugC}
	Let $\C$ be a unitary premodular category. The \define{symmetric} or \define{M\"uger} center of $\C$ is the full subcategory of $\C$ with objects
	\begin{align}
		\mug{\C}:=\{X\in\C|B_{X,Y}=\id_{X\otimes Y}\forall Y\in \C\}.
	\end{align}
	A premodular category is \define{symmetric} if $\mug{\C}=\C$, and \define{modular} if $\mug{\C}=\vvec{}$. We will refer to premodular categories which are neither symmetric nor modular as \define{properly premodular}.

	The $\S$-matrix acts as a witness for these properties. Symmetric categories have (matrix) rank 1 $\S$-matrices, while modular categories have invertible (unitary) $\S$-matrices.
\end{definition}

It will be convenient to define a slight generalization of the $\S$-matrix.
\begin{definition}[Connected $\S$-matrix]\label{def:commectedS}
	Recall that the trivalent vertices define a vector space, with Greek labels indicating basis vectors. We can therefore define the operator $\S_c$ by its action on the fusion space~\cite{kitaev2006anyons}
	\begin{align}
		\S_c
		\begin{array}{c}
			\includeTikz{Scop1}{
				\begin{tikzpicture}[scale=.5]
					\pgfmathsetmacro{\s}{sqrt(2)};
					\draw (0,-\s)--(0,0)--(-1,1) (0,0)--(1,1);
					\node[anchor=north,inner sep=.1] at (0,-\s) {\strut$c$};
					\node[anchor=south,inner sep=.1] at (-1,1) {\strut$b$};\node[anchor=south,inner sep=.1] at (1,1) {\strut$\dual{b}$};
					\node[anchor=north west,inner sep=.3] at(0,0) {\strut$\beta$};
				\end{tikzpicture}
			}
		\end{array}
		 & =\frac{\sqrt{d_c}}{\D}\sum_x d_x
		\begin{array}{c}
			\includeTikz{Scop2}{
				\begin{tikzpicture}[scale=.5]
					\pgfmathsetmacro{\s}{sqrt(2)};
					\draw[draw=white,double=black,ultra thick] (-.25,0) to [out=90,in=270] (-1,2);
					\draw[draw=white,double=black,ultra thick] (0,0) circle (1);
					\draw[draw=white,double=black,ultra thick] (-2,2) to [out=270,in=180] (-1,-2) to [out=0,in=270] (-.25,0);
					\draw (0,-1)--(0,-3);
					\node[anchor= north west,inner sep=.5] at(0,-1) {\strut$\beta$};
					\node[anchor=west,inner sep=.5] at(1,0) {\strut$\dual{b}$};
					\node[anchor=west,inner sep=1] at(-1,2) {\strut$x$};
					\node[anchor=east,inner sep=1] at(-2,2) {\strut$\dual{x}$};
					\node[anchor=north,inner sep=.5] at(0,-3) {\strut$c$};
				\end{tikzpicture}
			}
		\end{array}.\label{eqn:Scaction}
	\end{align}
	The matrix elements of this operator are
	\begin{align}
		\left[\S_c\right]_{(a,\alpha),(b,\beta)} & =\frac{1}{\D}
		\begin{array}{c}
			\includeTikz{ConnectedSmatrix}{
				\begin{tikzpicture}[scale=.65]
					\centerarc[draw=white,double=black,ultra thick](-.75,0)(0:180:1);
					\draw[draw=white,double=black,ultra thick] (.75,0) circle (1);
					\centerarc[draw=white,double=black,ultra thick](-.75,0)(180:360:1);
					\node[anchor=west,inner sep=.5] at(1.75,0) {\strut$b$};
					\node[anchor=west,inner sep=.5] at(.25,0) {\strut$a$};
					\node[anchor=east,inner sep=.5] at(-.25,0) {\strut$\dual{b}$};
					\node[anchor=east,inner sep=.5] at(-1.75,0) {\strut$\dual{a}$};
					\draw[thick] (.75,-1)--(.75,-1.2);\draw[thick] (-.75,1)--(-.75,1.2);
					\draw[thick] (.75,-1.2)to[out=270,in=270] (2.5,0)to[out=90, in=90] (-.75,1.2);
					\node[anchor=south,inner sep=.5] at(.75,-1) {\strut$\beta$};
					\node[anchor=north,inner sep=.75] at(-.75,1) {\strut$\alpha$};
					\node[anchor=east,inner sep=.5] at(-.75,1.25) {\strut$c$};
					\node[anchor=west,inner sep=.5] at(2.5,0) {\strut$\dual{c}$};
				\end{tikzpicture}
			}
		\end{array}.\label{eqn:connectedS}
	\end{align}
	The usual $\S$-matrix (defined in \cref{eqn:Smatrix}) occurs as a special case, namely $\S_{a,b}=\left[\S_1\right]_{a,b}$.
	The connected $\S$-matrix appears in Theorem~3.1.17 of \onlinecite{bakalov2001lectures}, and is closely related to the punctured $\S$-matrix of \onlinecite{Bonderson_2019}.
\end{definition}

The (connected) $\S$-matrix has a number of properties that we require.
\begin{lemma}\label{lem:productS}
	Let $\C$ be a unitary premodular category. The matrix elements of $\S$ obey
	\begin{align}
		\frac{\D}{d_c}\S_{a,c}\S_{b,c} & =\S_{a\otimes b,c}=\sum_x N_{a,b}^x \S_{x,c}.
	\end{align}
	\begin{proof}
		Provided in Appendix~\hyperref[lem:productS_pf]{\ref*{app:FC_pfs}}.
	\end{proof}
\end{lemma}

\begin{lemma}\label{lem:sumS}
	Let $\C$ be a unitary premodular category, then
	\begin{align}
		\sum_b d_b \S_{a,b} & =\indicator{a\in\mug{\C}} d_a \D,
	\end{align}
	where $\mug{\C}$ is the M\"uger center.
	\begin{proof}
		Provided in Appendix~\hyperref[lem:sumS_pf]{\ref*{app:FC_pfs}}.
	\end{proof}
\end{lemma}

\begin{corollary}[Premodular trap]\label{cor:premodulartrap}
	For any premodular theory, we have
	\begin{align}
		\begin{array}{c}
			\includeTikz{trapLHS}{
				\begin{tikzpicture}
					\draw[draw=white,double=black,ultra thick] (.25,-.75)--(.25,0);
					\draw[draw=white,double=black,ultra thick] (0,0) circle (.5);
					\draw[draw=white,double=black,ultra thick] (.25,0)--(.25,.75);
					\node[anchor=west] at (.5,0) {\strut$a$};
					\node[anchor=north,inner sep=.1] at (.25,-.75) {\strut$b$};
				\end{tikzpicture}
			}
		\end{array}
		 & =\frac{\S_{a,\dual{b}}}{\S_{1,\dual{b}}}
		\begin{array}{c}
			\includeTikz{trapRHS}{
				\begin{tikzpicture}
					\draw[draw=white,double=black,ultra thick] (.25,-.75)--(.25,0);
					\draw[draw=white,double=black,ultra thick] (.25,0)--(.25,.75);
					\node[anchor=north,inner sep=.1] at (.25,-.75) {\strut$b$};
				\end{tikzpicture}
			}
		\end{array}=\frac{\D \S_{a,\dual{b}}}{d_b}.
	\end{align}
	Using \cref{lem:productS,lem:sumS}, this gives
	\begin{align}
		\frac{1}{\D^2}\sum_a d_a
		\begin{array}{c}
			\includeTikz{trapLHS2}{
				\begin{tikzpicture}
					\draw[draw=white,double=black,ultra thick] (-.5,-.75)to[out=45,in=270](-.25,0) (.5,-.75)to[out=90+45,in=270](.25,0);
					\draw[draw=white,double=black,ultra thick] (0,0) circle (.5);
					\draw[draw=white,double=black,ultra thick] (-.25,0)to[out=90,in=270+45](-.5,.75) (.25,0)to[out=90,in=180+45](.5,.75);
					\node[anchor=west] at (.5,0) {\strut$a$};
					\node[anchor=north,inner sep=.1] at (-.5,-.75) {\strut$x$};\node[anchor=north,inner sep=.1] at (.5,-.75) {\strut$y$};
					\node[anchor=south,inner sep=.1] at (-.5,.75) {\phantom{\strut$x$}};\node[anchor=south,inner sep=.1] at (.5,.75) {\phantom{\strut$y$}};
				\end{tikzpicture}
			}
		\end{array}
		 & =\sum_{z\in\mug{\C},\mu} \sqrt{\frac{d_z}{d_xd_y}}
		\begin{array}{c}
			\includeTikz{trapRHS2}{
				\begin{tikzpicture}
					\draw (-.5,-.75)--(0,-.5)--(0,.5)--(-.5,.75) (.5,-.75)--(0,-.5)--(0,.5)--(.5,.75);
					\node[anchor=north,inner sep=.1] at (-.5,-.75) {\strut$x$};\node[anchor=north,inner sep=.1] at (.5,-.75) {\strut$y$};
					\node[anchor=south,inner sep=.1] at (-.5,.75) {\strut$x$};\node[anchor=south,inner sep=.1] at (.5,.75) {\strut$y$};
					\node[anchor=west,inner sep=.1] at(0,0) {\strut$z$};
					\node[anchor=south east,inner sep=.3] at(0,-.5) {\strut$\mu$};\node[anchor=north east,inner sep=.3] at(0,.5) {\strut$\mu$};
				\end{tikzpicture}
			}
		\end{array}.
	\end{align}
\end{corollary}

\begin{lemma}\label{lem:TrScSc}
	Let $\C$ be a unitary premodular category, then
	\begin{align}
		\sum_{c\in\C}\Tr\S_c^\dagger\S_c & =\D^2,
	\end{align}
	where $\S_c$ is the connected $\S$-matrix and $\D$ is the total dimension of $\C$.
	\begin{proof}
		Provided in Appendix~\hyperref[lem:TrScSc_pf]{\ref*{app:FC_pfs}}.
	\end{proof}
\end{lemma}

\begin{definition}[Algebra object]\label{def:Alg}
	An algebra $(A,m,\eta)$ in a fusion category is an object $A=1\oplus a_1\oplus a_2\oplus \cdots$, along with morphisms $m:A\times A\to A$ and $\eta:1\to A$. For simplicity, we restrict $\C$ and $A$ to be multiplicity free, that is each simple object $a_i$ occurs at most once. We represent the multiplication morphism $m$ as
	\begin{align}
		m=
		\begin{array}{c}
			\includeTikz{Algm_1}{
				\begin{tikzpicture}[scale=.5]
					\pgfmathsetmacro{\s}{sqrt(2)}
					\draw (0,0)--(1,1) node[below,pos=0,inner sep=.1]{\strut$A$} (1,1)--(2,0)node[below,pos=1,inner sep=.1]{\strut$A$}  (1,1)--(1,1+\s)node[above,pos=1,inner sep=.1]{\strut$A$};
					\fill [shift={(1,1)}] (-.1,-.1) rectangle (.1,.1);
				\end{tikzpicture}
			}
		\end{array}
		 & =\sum_{a,b,c\in A}m_{ab}^c
		\begin{array}{c}
			\includeTikz{Algm_2}{
				\begin{tikzpicture}[scale=.5]
					\pgfmathsetmacro{\s}{sqrt(2)}
					\draw (0,0)--(1,1) node[below,pos=0,inner sep=.1]{\strut$a$} (1,1)--(2,0)node[below,pos=1,inner sep=.1]{\strut$b$} (1,1)--(1,1+\s)node[above,pos=1,inner sep=.1]{\strut$c$};
				\end{tikzpicture}
			}
		\end{array}.
	\end{align}
	To simplify the notation, we define $m_{xy}^z=0$ whenever any of the labels do not occur in the decomposition of $A$. This allows us to always sum over simple objects in $\C$. Additionally, we suppress the $A$ label. Any unlabeled lines carry an implicit $A$.
	The multiplication of the algebra should be associative (in $\C$)
	\begin{align}
		\begin{array}{c}
			\includeTikz{AAssociativeLHS}{
				\begin{tikzpicture}[scale=.5]
					\pgfmathsetmacro{\s}{sqrt(2)}
					\draw (0,0)--(1,1) (1,1)--(2,0) (1,1)--(2,2)--(2,2+\s) (2,2)--(4,0);
					\fill [shift={(1,1)}] (-.1,-.1) rectangle (.1,.1);
					\fill [shift={(2,2)}] (-.1,-.1) rectangle (.1,.1);
				\end{tikzpicture}
			}
		\end{array}
		 & =
		\begin{array}{c}
			\includeTikz{AAssociativeRHS}{
				\begin{tikzpicture}[scale=.5,xscale=-1]
					\pgfmathsetmacro{\s}{sqrt(2)}
					\draw (0,0)--(1,1) (1,1)--(2,0) (1,1)--(2,2)--(2,2+\s) (2,2)--(4,0);
					\fill [shift={(1,1)}] (-.1,-.1) rectangle (.1,.1);
					\fill [shift={(2,2)}] (-.1,-.1) rectangle (.1,.1);
				\end{tikzpicture}
			}
		\end{array},
	\end{align}
	or in components
	\begin{align}
		\sum_x m_{ab}^xm_{xc}^d
		\begin{array}{c}
			\includeTikz{AAssociativeLHS_2}{
				\begin{tikzpicture}[scale=.5]
					\pgfmathsetmacro{\s}{sqrt(2)}
					\draw (0,0)--(1,1)node[below,pos=0,inner sep=.1]{\strut$a$} (1,1)--(2,0)node[below,pos=1,inner sep=.1]{\strut$b$} (1,1)--(2,2)--(2,2+\s)node[above,pos=1,inner sep=.1]{\strut$d$} (2,2)--(4,0)node[below,pos=1,inner sep=.1]{\strut$c$};
					\node[above left,inner sep=.1] at(1.5,1.5) {\strut$x$};
				\end{tikzpicture}
			}
		\end{array}
		 & =
		\sum_y m_{ay}^dm_{bc}^y
		\begin{array}{c}
			\includeTikz{AAssociativeRHS_2}{
				\begin{tikzpicture}[scale=.5,xscale=-1]
					\pgfmathsetmacro{\s}{sqrt(2)}
					\draw (0,0)--(1,1)node[below,pos=0,inner sep=.1]{\strut$c$} (1,1)--(2,0)node[below,pos=1,inner sep=.1]{\strut$b$} (1,1)--(2,2)--(2,2+\s)node[above,pos=1,inner sep=.1]{\strut$d$} (2,2)--(4,0)node[below,pos=1,inner sep=.1]{\strut$a$};
					\node[above right,inner sep=.1] at(1.5,1.5) {\strut$y$};
				\end{tikzpicture}
			}
		\end{array}
		=
		\sum_{x,z} m_{ab}^xm_{xc}^d \bigg[F_{abc}^{d}\bigg]_{xz}
		\begin{array}{c}
			\includeTikz{AAssociativeLHS_3}{
				\begin{tikzpicture}[scale=.5,xscale=-1]
					\pgfmathsetmacro{\s}{sqrt(2)}
					\draw (0,0)--(1,1)node[below,pos=0,inner sep=.1]{\strut$c$} (1,1)--(2,0)node[below,pos=1,inner sep=.1]{\strut$b$} (1,1)--(2,2)--(2,2+\s)node[above,pos=1,inner sep=.1]{\strut$d$} (2,2)--(4,0)node[below,pos=1,inner sep=.1]{\strut$a$};
					\node[above right,inner sep=.1] at(1.5,1.5) {\strut$z$};
				\end{tikzpicture}
			}
		\end{array},
	\end{align}
	where the final equality is obtained by using the $F$-move on the left hand side. The components of $m$ therefore obey
	\begin{align}
		m_{ay}^{d}m_{bc}^y & =\sum_x m_{ab}^xm_{xc}^d \bigg[F_{abc}^{d}\bigg]_{xy}.\label{eqn:algCompatibility}
	\end{align}
	Multiplication by the unit obeys
	\begin{align}
		\begin{array}{c}
			\includeTikz{AUnitLHS}{
				\begin{tikzpicture}[scale=.5]
					\pgfmathsetmacro{\s}{sqrt(2)}
					\draw (0,0)--(1,1) (1,1)--(3,-1) (1,1)--(1,1+\s);
					\fill [shift={(1,1)}] (-.1,-.1) rectangle (.1,.1);
					\fill [shift={(0,0)}] (0,0) circle (.1);
				\end{tikzpicture}
			}
		\end{array}
		 & =
		\begin{array}{c}
			\includeTikz{AUnitMid}{
				\begin{tikzpicture}[scale=.5]
					\pgfmathsetmacro{\s}{sqrt(2)}
					\draw (0,-1)--(0,1+\s);
				\end{tikzpicture}
			}
		\end{array}
		=
		\begin{array}{c}
			\includeTikz{AUnitRHS}{
				\begin{tikzpicture}[scale=.5,xscale=-1]
					\pgfmathsetmacro{\s}{sqrt(2)}
					\draw (0,0)--(1,1) (1,1)--(3,-1) (1,1)--(1,1+\s);
					\fill [shift={(1,1)}] (-.1,-.1) rectangle (.1,.1);
					\fill [shift={(0,0)}] (0,0) circle (.1);
				\end{tikzpicture}
			}
		\end{array},
	\end{align}
	or in components
	\begin{align}
		\eta m_{1x}^{x} & =1=\eta m_{x1}^x.
	\end{align}
	A \define{coalgebra} is the same, with everything flipped upside down.

	Two algebras $(A,m,\eta)$ and $(A,n,\theta)$ are isomorphic if they can be related by
	\begin{align}
		m_{ab}^c & =n_{ab}^c\frac{\beta_c}{\beta_a\beta_b}, \\
		\eta     & =\frac{1}{\beta_1}\theta,
	\end{align}
	where $\beta_a$ are nonzero complex numbers. We can (and will) always use this to normalize $\eta=1$. When it does not cause confusion, we will indicate an algebra by its object, for example $A=1$, the `unit algebra'.
\end{definition}

\begin{definition}[Frobenius algebra]\label{def:FrobAlg}
	A Frobenius algebra is a quintuple $(A,m,\eta,\mu,\epsilon)$, where $(A,m,\eta)$ is an algebra, $(A,\mu,\epsilon)$ is a coalgebra. The algebra and coalgebra maps obey
	\begin{align}
		\begin{array}{c}
			\includeTikz{Frob_1}{
				\begin{tikzpicture}[scale=.5]
					\pgfmathsetmacro{\s}{sqrt(2)}
					\clip (-2,-1-\s/2) rectangle (2,1+\s/2);
					\begin{scope}[shift={(0,-\s/2)}]
						\draw (0,0)--(-1,-1)  (0,0)--(1,-1) (0,0)--(0,\s);
						\fill [] (-.1,-.1) rectangle (.1,.1);
					\end{scope}
					\begin{scope}[shift={(0,\s/2)},yscale=-1]
						\draw (0,0)--(-1,-1)  (0,0)--(1,-1) (0,0)--(0,\s);
						\fill [] (-.1,-.1) rectangle (.1,.1);
					\end{scope}
				\end{tikzpicture}
			}
		\end{array}
		 & =
		\begin{array}{c}
			\includeTikz{Frob_2}{
				\begin{tikzpicture}[scale=.5]
					\pgfmathsetmacro{\s}{sqrt(2)}
					\clip (-2,-1-\s/2) rectangle (2,1+\s/2);
					\begin{scope}[shift={(-.5,.5)}]
						\draw (0,0)--(-1,-1)--(-1,-2.5)  (0,0)--(1,-1) (0,0)--(0,\s);
						\fill [] (-.1,-.1) rectangle (.1,.1);
					\end{scope}
					\begin{scope}[shift={(.5,-.5)},yscale=-1]
						\draw (0,0)--(-1,-1)  (0,0)--(1,-1)--(1,-2.5) (0,0)--(0,\s);
						\fill [] (-.1,-.1) rectangle (.1,.1);
					\end{scope}
				\end{tikzpicture}
			}
		\end{array}
		=
		\begin{array}{c}
			\includeTikz{Frob_3}{
				\begin{tikzpicture}[scale=.5,xscale=-1]
					\pgfmathsetmacro{\s}{sqrt(2)}
					\clip (-2,-1-\s/2) rectangle (2,1+\s/2);
					\begin{scope}[shift={(-.5,.5)}]
						\draw (0,0)--(-1,-1)--(-1,-2.5)  (0,0)--(1,-1) (0,0)--(0,\s);
						\fill [] (-.1,-.1) rectangle (.1,.1);
					\end{scope}
					\begin{scope}[shift={(.5,-.5)},yscale=-1]
						\draw (0,0)--(-1,-1)  (0,0)--(1,-1)--(1,-2.5) (0,0)--(0,\s);
						\fill [] (-.1,-.1) rectangle (.1,.1);
					\end{scope}
				\end{tikzpicture}
			}
		\end{array}.
	\end{align}
	A Frobenius algebra is called \define{unitary}, or a Q-system, when the coalgebra structure is the adjoint of the algebra structure, or in components $\mu^{ab}_c=\left(m_{ab}^c\right)^*$ and $\epsilon=\eta^*$. Unitarity is assumed in all mentions of Frobenius algebras in this manuscript.

\end{definition}

A Frobenius algebra is said to be \define{strongly separable} if
\begin{align}
	\begin{array}{c}
		\includeTikz{A_strsep_1}{
			\begin{tikzpicture}[scale=.5]
				\draw (0,-1)--(0,1);
				\fill [shift={(0,-1)}] (0,0) circle (.1);
				\fill [shift={(0,1)}] (0,0) circle (.1);
			\end{tikzpicture}
		}
	\end{array} & =\alpha_A
	\\
	\begin{array}{c}
		\includeTikz{AspecialLHS}{
			\begin{tikzpicture}[scale=.5]
				\pgfmathsetmacro{\s}{sqrt(2)}
				\begin{scope}[yscale=-1]
					\draw (1,1)--(0,0) (1,1)--(1,1+\s) (1,1)--(2,0);
					\fill [shift={(1,1)}] (-.1,-.1) rectangle (.1,.1);
				\end{scope}
				\draw (1,1)--(0,0) (1,1)--(1,1+\s) (1,1)--(2,0);
				\fill [shift={(1,1)}] (-.1,-.1) rectangle (.1,.1);
			\end{tikzpicture}
		}
	\end{array} & =	\beta_A
	\begin{array}{c}
		\includeTikz{AspecialRHS}{
			\begin{tikzpicture}[scale=.5]
				\pgfmathsetmacro{\s}{sqrt(2)}
				\draw (0,-1-\s)--(0,1+\s);
			\end{tikzpicture}
		}
	\end{array},\label{eqn:algnorm}
\end{align}
where $\alpha_A\beta_A\neq 0$. We normalize $\alpha_A=1$ and $\beta_A=d_A$, where $d_A=\sum_{a\in A}d_a$.

If the underlying category $\C$ is braided, we say the algebra $A$ is \define{commutative} if
\begin{align}
	\begin{array}{c}
		\includeTikz{AcommLHS}{
			\begin{tikzpicture}[scale=.35]
				\pgfmathsetmacro{\s}{sqrt(2)}
				\draw[draw=white,double=black,ultra thick] (2,-2) to [out=135,in=225] (0,0)--(1,1)--(2,0)  to [out=315,in=45] (0,-2);
				\draw[draw=white,double=black,ultra thick](2,0)  to [out=315,in=45] (0,-2);
				\draw [thick](1,1)--(1,1+\s);
				\fill [shift={(1,1)}] (-.1,-.1) rectangle (.1,.1);
			\end{tikzpicture}
		}
	\end{array}
	 & =
	\begin{array}{c}
		\includeTikz{AcommMid}{
			\begin{tikzpicture}[scale=.35]
				\pgfmathsetmacro{\s}{sqrt(2)}
				\draw	[thick](1,1)--(1,1+\s);
				\draw	[thick] (1,1)--(0,-2);
				\draw	[thick] (1,1)--(2,-2);
				\fill [shift={(1,1)}] (-.1,-.1) rectangle (.1,.1);
			\end{tikzpicture}
		}
	\end{array}
	=
	\begin{array}{c}
		\includeTikz{AcommRHS}{
			\begin{tikzpicture}[scale=.35,xscale=-1]
				\pgfmathsetmacro{\s}{sqrt(2)}
				\draw[draw=white,double=black,ultra thick] (2,-2) to [out=135,in=225] (0,0)--(1,1)--(2,0)  to [out=315,in=45] (0,-2);
				\draw[draw=white,double=black,ultra thick](2,0)  to [out=315,in=45] (0,-2);
				\draw [thick](1,1)--(1,1+\s);
				\fill [shift={(1,1)}] (-.1,-.1) rectangle (.1,.1);
			\end{tikzpicture}
		}
	\end{array},\label{eqn:commutativealg}
\end{align}
or in components
\begin{align}
	R_{ab}^cm_{ba}^c=m_{ab}^c=m_{ba}^c\left(R_{ba}^c\right)^*.
\end{align}

\subsection{Examples}\label{sec:examples}

To aid understanding, and illustrate our results, we will refer to the following examples throughout the remainder of the manuscript.

The unitary fusion category $\vvectwist{\ZZ{2}}{\omega}$ is the category of finite dimensional $\ZZ{2}$-graded vector spaces. The simple objects are labeled by the group elements $\ZZ{2}:=\set{1,x}{x^2=1}$. Since we neglect to draw the unit object, corresponding to the identity group element, the only nonzero trivalent vertex is
\begin{align}
	\begin{array}{c}
		\includeTikz{trivalentVecZ2_1}{
			\begin{tikzpicture}[scale=.25]
				\draw[thick] (0,0)--(1,1) (1,1)--(2,0);
			\end{tikzpicture}
		}
	\end{array}.
\end{align}
These are pointed categories, and so $d_x=1$. The total quantum dimension is $\D^2=2$.

It is straightforward to check that there are exactly two inequivalent associators compatible with this fusion rule, namely
\begin{align}
	\begin{array}{c}
		\includeTikz{FVecZ2_LHS}{
			\begin{tikzpicture}[scale=.3,yscale=-1];
				\draw[thick] (0,0)--(0,1) (-1,2)--(-2,3) (-1,2)--(0,3) (0,1)--(2,3);
			\end{tikzpicture}
		}
	\end{array}
	 & =
	\omega
	\begin{array}{c}
		\includeTikz{FVecZ2_RHS}{
			\begin{tikzpicture}[scale=.3,xscale=-1,yscale=-1];
				\draw[thick] (0,0)--(0,1) (-1,2)--(-2,3) (-1,2)--(0,3) (0,1)--(2,3);
			\end{tikzpicture}
		}
	\end{array},
\end{align}
with $\omega=\pm 1$. With the associator fixed, there are two compatible braidings
\begin{align}
	\begin{array}{c}
		\includeTikz{braidVecZ2_LHS}{
			\begin{tikzpicture}[scale=.3]
				\draw[draw=white,double=black,ultra thick] (2,-2) to [out=135,in=225] (0,0)--(1,1)--(2,0)  to [out=315,in=45] (0,-2);
				\draw[draw=white,double=black,ultra thick](2,0)  to [out=315,in=45] (0,-2);
			\end{tikzpicture}
		}
	\end{array}
	 & =\phi
	\begin{array}{c}
		\includeTikz{braidVecZ2_RHS}{
			\begin{tikzpicture}[scale=.3]
				\draw	[thick] (1,1)--(0,-2);
				\draw	[thick] (1,1)--(2,-2);
			\end{tikzpicture}
		}
	\end{array},
\end{align}
with $\phi^2 = \omega$. The twists and $\S$-matrices of these models are given by
\begin{align}
	\theta_x & =\phi &  &  & \S & =\frac{1}{\sqrt{2}}\begin{pmatrix}
		1 & 1 \\1&\omega
	\end{pmatrix},
\end{align}
so the categories are modular when $\omega=-1$, and symmetric when $\omega=+1$. We denote by $\vvectwist{\ZZ{2}}{\omega,\phi}$ the braided category, with associator $\omega$ and braiding $\phi$.

These four examples are included in the attached Mathematica file~\cite{premoddata} as $\vvectwist{\ZZ{2}}{1,1}=\FR{2}{0}{1}{0}$, $\vvectwist{\ZZ{2}}{1,-1}=\FR{2}{0}{1}{2}$, $\vvectwist{\ZZ{2}}{-1,i}=\FR{2}{0}{1}{1}$, and $\vvectwist{\ZZ{2}}{-1,-i}=\FR{2}{0}{1}{3}$.

For these examples, there are two possible algebras, namely $A_0:=1$, with trivial $m$ morphism, and $A_1:=1\oplus x$. The algebra $A_0$ is compatible as a commutative algebra. For $A_1$, compatibility as a Frobenius algebra (\cref{eqn:algCompatibility}) reduces to
\begin{align}
	m_{x,x}^{1}\omega & = m_{x,x}^{1},
\end{align}
so is only a valid algebra object when $\omega = 1$. In that case, $A_1$ is commutative when
\begin{align}
	m_{x,x}^1 \phi = m_{x,x}^{1} \iff \phi = 1.
\end{align}

% !TeX spellcheck = en_US
% This line sets the project root file.
% !TEX root = ../WW_EntanglementEntropy.tex
%
\section{Loop-gas models in (2+1) and (3+1) dimensions}\label{sec:models}

In this work, we study loop-gas models. In their most general form, these models have ground states described by superpositions of string diagrams subject to a collection of rules, for example a diagram may be declared invalid in a ground state superposition if it contains an `open string'. We focus on topological loop-gas models. In this case, the rules are a collection of local manipulations or moves under which states must be invariant. These are designed to ensure invariance under diffeomorphism. Given a triangulation of the manifold, which provides a lattice structure on which a condensed matter model can be defined~\footnote{The lattice is the dual of the triangulation.}, the local moves ensure retriangulation invariance.

Levin-Wen~\cite{wen2004quantum,levin2005string} and Walker-Wang~\cite{Walker2012TQFT,williamson2017hamiltonian,crane1993categorical,crane1997state} models are, respectively, two- and three-dimensional Hamiltonian models that give rise to topological loop-gas states as their ground states.
Hamiltonians for these models are given in \onlinecite{levin2005string} and \onlinecite{Walker2012TQFT} for Levin-Wen and Walker-Wang models respectively. Our results do not depend on the particular form of the Hamiltonian, rather on universal properties of the ground states in the associated topological phase.

Quasiparticle excitations in these models are defects in the ground state, corresponding to a local change in the rules. Far from the excitations, the excited state remains invariant under the original moves but, for example, a string may be allowed to terminate at the location of the excitation.

\subsection{Bulk}

We now briefly introduce the categorical description of the models of interest far from any physical boundary. We begin by discussing the bulk for the the (2+1)-dimensional case, namely, the Levin-Wen models before proceeding to discuss the bulk of the (3+1)-dimensional Walker Wang models.

\subsubsection*{Levin-Wen models}

Given a unitary spherical fusion category $\C$ (\cref{def:UFC}), and a given lattice embedded on a two-dimensional manifold, the ground states of Levin-Wen models are superpositions of closed diagrams from the category. Strings lie along the edges of the lattice, and closed means they cannot terminate.
To produce a lattice model, $\rk{\C}$-dimensional vector spaces are assigned to each edge. These vector spaces are equipped with an orthogonal basis consisting of the objects of $\C$.  If the category is multiplicity free, vector spaces corresponding to the fusion spaces are assigned to the vertices~\protect\footnote{The global vector space can be made a tensor product space by, for example, choosing the dimension of the vertex space to be the largest of all the fusion spaces.}. An abstract string configuration is realized by the vector consisting of the appropriate basis vector on each edge.
At the vertices, the fusion rules of the category dictate which strings can fuse.
If a given configuration occurs in a particular ground state, then any other configuration that is obtainable by local moves (i.e. $F$- or loop-moves) also occurs in that ground state. The relative coefficients are dictated by the $F$-symbols, and consistency is ensured by the pentagon equation. Since the allowed moves are all local, there may be multiple ground states on manifolds with nontrivial genus. For example, a loop enclosing a cycle of the torus cannot be removed with the local loop move.

The collection of excitations (anyons) resulting from the Levin-Wen construction is called the Drinfeld center, denoted $\drinfeld{\C}$. We refer to \onlinecite{MR3242743} for a formal definition.

\subsubsection*{Walker-Wang models}

In (3+1)-dimensions, for Walker-Wang models, the diagrams can also include crossings. This required additional data to be added to the category, in particular a braiding (\cref{def:BFC}). If we picture these diagrams embedded in 3 dimensional space, there is an ambiguity involved in these crossing. For example, if we look at a crossing from `the side', there is no crossing. This ambiguity can be resolved by widening the strings into ribbons. This is implemented by insisting that the braided category is premodular (\cref{def:Premodular}).

Given a premodular category $\C$, and a lattice embedded in a 3-dimensional manifold, a Walker-Wang model is defined in essentially the same way as a Levin-Wen model. In addition to the $F$- and loop- moves, $R$- moves and insertion of links (or knots), such as the $\S$-matrix are allowed. Again, given any closed string configuration, any other configuration that can be reached via these rules is included in the ground state superposition.

Within the collection of string types, the subset that can be `unlinked' from any other string is called the M\"uger center of $\C$, denoted $\mug{\C}$ (\cref{def:mugC}). The M\"uger center labels the particle excitations of the Walker-Wang model~\cite{ZWang}.

\subsection{Boundaries}

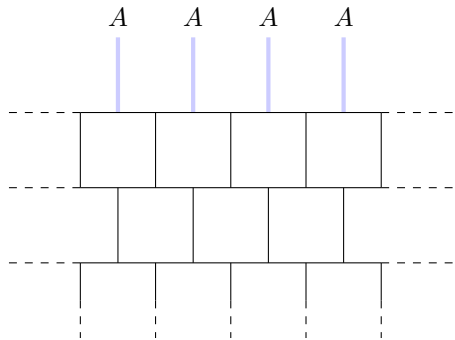
\begin{figure}
	\includeTikz{2DAlgebraLattice}{
		\begin{tikzpicture}
			\foreach \x in {-2,...,1} \draw[ultra thick,draw=blue!20] (\x+1/2,1)--(\x+1/2,0) node [pos=0,above,inner sep=.1] {\strut$A$};
			\begin{scope}[shift={(0,0)}]
				\draw (-2,0)--(2,0);\draw[dashed] (2,0)--(3,0);\draw[dashed] (-2,0)--(-3,0);
				\foreach \x in {-2,...,2} \draw (\x,0)--(\x,-1);
			\end{scope}
			\begin{scope}[shift={(0,-1)}]
				\draw (-2,0)--(2,0);\draw[dashed] (2,0)--(3,0);\draw[dashed] (-2,0)--(-3,0);
				\foreach \x in {-2,...,1} \draw (\x+1/2,0)--(\x+1/2,-1);
			\end{scope}
			\begin{scope}[shift={(0,-2)}]
				\draw (-2,0)--(2,0);\draw[dashed] (2,0)--(3,0);\draw[dashed] (-2,0)--(-3,0);
				\foreach \x in {-2,...,2} {\draw (\x,0)--(\x,-.5);\draw[dashed] (\x,-.5)--(\x,-1);};
			\end{scope}
		\end{tikzpicture}
	}
	\caption{An algebra specifies a boundary for a Levin-Wen model on a `comb lattice'. Dashed lines indicate the lattice continues. The top, thick blue lines are labeled by an algebra $A$ that defines a physical boundary to the lattice.}\label{fig:Alattice2D}
\end{figure}

To include a physical boundary in a loop-gas model, the rules must be modified. For the topological loop-gas models, these rules are again defined by local moves in the vicinity of the boundary. These must be compatible with the bulk moves, and ensure topological/retriangulation invariance at the boundary. We restrict our attention to gapped boundaries.

\subsubsection*{Levin-Wen models}

There are various equivalent classifications for the gapped boundaries of Levin-Wen models~\cite{Fuchs2002,kitaev2012models,Fuchs2013,Fuchs2015,1706.03329,1706.00650}. In this work, we use an internal classification. In this framework, gapped boundary conditions for Levin-Wen models are labeled by \define{indecomposable strongly separable, Frobenius algebra objects} (\cref{def:FrobAlg}) in $\mathcal{C}$ up to Morita equivalence~\cite{1706.03329,1706.00650}. We restrict to multiplicity free algebras for simplicity. These algebra objects are (not necessarily simple) objects in $\C$, and their simple subjects are roughly the string types that are allowed to terminate on the boundary.

On the comb lattice (\cref{fig:Alattice2D}), for example, the dangling edges only take values in the chosen algebra. Far from the boundary, the ground states look just like those with no boundary. Near the boundary, loops are no longer required to be closed, rather they can terminate on the boundary if their label occurs within the algebra. We refer to \onlinecites{1706.00650,1706.03329} for more details, including an explicit Hamiltonian.

\subsubsection*{Walker-Wang models}

Just as in the bulk, when me move to (3+1)-dimensions, the braiding must be taken into account. A general classification of gapped boundaries for Walker-Wang models has not been established, so we proceed for a class of boundaries generalizing those for Levin-Wen introduced above. As before, a boundary is labeled by an algebra object. Since the bulk is braided, an additional compatibility condition is required, namely that the algebra is commutative (\cref{eqn:commutativealg}).
Finally, in this work, an \define{indecomposable, strongly separable, commutative, Frobenius algebra object} labels a gapped boundary condition of a Walker-Wang model~\footnote{Private communications with David Aasen}.

\subsection{Examples}\label{sec:examples_phys}

Recall the examples from \cref{sec:examples}. In (2+1)-dimensions, $\vvectwist{\ZZ{2}}{1,\pm 1}$ lead to the same loop-gas model, since the Levin-Wen construction doesn't make use of the braiding. This model is the equally weighted superposition of all loop diagrams (with no branching due to the fusion rules). This is the ground state of the toric code model~\cite{kitaev1997fault}.

Likewise, $\vvectwist{-1}{\pm i}$ correspond to the same loop-gas model. Due to the nontrivial Frobenius-Schur indicator~\cite{kitaev2006anyons}, it is convenient to associate $-1$ to a loop rather than $+1$ (otherwise we can take extra care when bending lines). The ground state is therefore a superposition of loops, but weighted by $(-1)^{\text{number of loops}}$. This is commonly called the double-semion model.

There are two possible (gapped) boundaries for the toric code, the `smooth' boundary, corresponding to the algebra $A_0$, and the `rough' boundary, corresponding to $A_1$. We refer to \onlinecite{brayvi1998quantum} for more details. The double-semion model only allows for one kind of (gapped) boundary, labeled by $A_0$.

In (3+1)-dimensions, each of these models labels a distinct Walker-Wang model. The models $\vvectwist{\ZZ{2}}{1,1},\,\vvectwist{\ZZ{2}}{1,-1}$ are commonly called the \define{bosonic-} and \define{fermionic-} toric code models respectively~\cite{hamma2005string,von2015walker}. Since these categories both have $\mug{\C}=\{1,x\}$, they both have a single particle excitation, in addition to the trivial excitation, whose self-statistics lead to the names of the models. The two models $\vvectwist{-1}{\pm i}$ are both referred to as semion models.

In (3+1)-dimensions, the bosonic toric code still has two kinds of boundaries, but the remaining models are only compatible with the trivial boundary labeled by $A_0$.
% !TeX spellcheck = en_US
% This line sets the project root file.
% !TEX root = ../WW_EntanglementEntropy.tex
%

\section{Entropy diagnostics}\label{sec:entropydiagnostics}

In what follows we describe the universal correction to the area law that we expect for topological phases.
We then define two diagnostics that can be used to probe the properties of the excitations at the boundary of (3+1)-dimensional topological phases.

\subsection{The universal correction to the area law}

The ground states of topological phases of matter demonstrate robust long-range entanglement that is not present in trivial phases~\cite{hamma2005bipartite, kitaev2006topological, levin2006detecting}.
Typically, we expect the entanglement entropy shared between a subsystem of a ground state of a gapped phase with the rest of the system to respect an area law, i.e., the entanglement will scale with the size of the surface area of the subsystem.
The long-range entanglement manifests as a universal correction to the area law.
More precisely, we expect that if we partition the ground state of a system into two subsystems, $R$ and its complement $\comp{R}$, the entanglement entropy, $S_R$, will satisfy
\begin{equation}
	S_R=\alpha|\partial R|-b_R \LWTEE.
	\label{eqn:arealaw}
\end{equation}
Here $\alpha$ is a non-universal coefficient that depends on the microscopic details of the system, $|\partial R|$ is the surface area of the interface between the partitioned regions, $b_R$ is the number of disjoint components of the interface between $R$ and $\comp{R}$, and $\LWTEE$ is a universal constant commonly known as the topological entanglement entropy. We have assumed that $R$ is large compared to the correlation length of the system, and its shape has no irregular features.

\subsection{(2+1)-dimensional models}

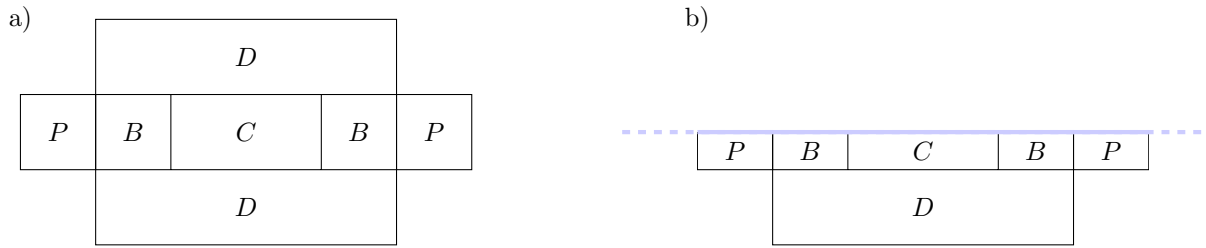
\begin{figure}
	\includeTikz{TwoDimEntropyRegions}{
		\begin{tikzpicture}
			\begin{scope}
				\node at (-3,1.5) {a)};
				\draw (-3,-.5) rectangle (-2,.5);\draw (3,-.5) rectangle (2,.5);
				\draw (-2,-1.5) rectangle (2,1.5);
				\draw (-2,-.5)--(2,-.5) (-2,.5)--(2,.5) (-1,-.5)--(-1,.5)  (1,-.5)--(1,.5);
				\node at (0,0) {$C$};\node at (-1.5,0) {$B$};\node at (1.5,0) {$B$};
				\node at (0,1) {$D$};\node at (0,-1) {$D$};
				\node at (-2.5,0) {$P$};\node at (2.5,0) {$P$};
			\end{scope}
			\begin{scope}[shift = {(9,0)}]
				\node at (-3,1.5) {b)};
				\begin{scope}
					\clip (-3,-2) rectangle (3,00);
					\draw (-3,-.5) rectangle (-2,.5);\draw (3,-.5) rectangle (2,.5);
					\draw (-2,-1.5) rectangle (2,1.5);
					\draw (-2,-.5)--(2,-.5) (-2,.5)--(2,.5) (-1,-.5)--(-1,.5)  (1,-.5)--(1,.5);
					\node at (0,-.25) {$C$};\node at (-1.5,-.25) {$B$};\node at (1.5,-.25) {$B$};
					\node at (0,1) {$D$};\node at (0,-1) {$D$};
					\node at (-2.5,-.25) {$P$};\node at (2.5,-.25) {$P$};
				\end{scope}
				\draw[blue!20,ultra thick](-3,0)--(3,0);
				\draw[blue!20,ultra thick,dashed](-4,0)--(4,0);
			\end{scope}
		\end{tikzpicture}}
	{\phantomsubcaption\label{fig:LWregionsblk}}
	{\phantomsubcaption\label{fig:LWregionsbnd}}
	\caption{Example of subsystems that can be used to find topological entropies in (2+1)-dimensions. The region $A$ is the complement of $BCD$.
		The regions \subref*{fig:LWregionsblk}) are used to find the bulk entropy $\LWTEE$, and the regions \subref*{fig:LWregionsbnd}) are used for the boundary entropy $\LWbndTEE$.
	}\label{fig:LWregions}
\end{figure}

Intimately connected to the long-range entanglement of a topological phase are the properties of its low-energy excitations.
A large class of topological models in (2+1)-dimensions are the Levin-Wen string-net models~\cite{levin2005string}. These models support topological point-like excitations that can be braided to change the state of the system.

Throughout this work we will be interested in the boundaries of topological phases. Importantly, topological particles can behave differently in the vicinity of the boundary of a phase. For instance, topological particles found in the bulk may become trivial particles close to certain boundaries. This is because topological particles can condense at the boundary such that non-trivial charges can be created locally.

As the physics of the quasi-particles of a topological phase can change close to its boundary, so to do we expect that the nature of its long-range entanglement to change. In \onlinecite{kim2015ground}, several topological entanglement entropy diagnostics were found to probe long-range entanglement of a model, both in the bulk and near to a boundary.
The first is the bulk topological entanglement entropy
\begin{align}
	\LWTEE:=S_{BC}+S_{CD}-S_B-S_D,\label{eqn:LWent}
\end{align}
where the regions are depicted in \cref{fig:LWregionsblk}, and $XY:=X\cup Y$. If $\LWTEE=0$, all point-like excitations can be created on the distinct parts of $P$ with a creation operator that has no support on $ACD$, where $A$ is the region that is complement to those shown in the figure. In this case, we declare them trivial. Conversely, if there are non-trivial topological excitations, for example created with string-like operators, $\LWTEE$ is necessarily non-zero.

In the presence of a gapped boundary, the excitations may differ. If a bulk topological excitation can be discarded or `condensed' on the boundary, it is possible to locally create such an excitation near the boundary. This is detected using the diagnostic
\begin{align}
	\LWbndTEE:=S_{BC}+S_{CD}-S_B-S_D,\label{eqn:LWbndent}
\end{align}
where the regions are depicted in \cref{fig:LWregionsbnd}. If $\LWbndTEE=0$, all point-like excitations on $P$ can be created with an operator that has no support on $ACD$, while non-trivial excitations require non-zero $\LWbndTEE$.

\subsection{\protect(3+1)-dimensional models}
Walker-Wang models give rise to both point- and line-like topological particles in the bulk, in addition to boundary excitations. Unlike Levin-Wen models, in some instances topological particles are only found at the boundary.

Since there are two kinds of topological excitations in (3+1)-dimensions, we might expect that there are two bulk diagnostics generalizing $\LWTEE$. However, as it has been shown~\cite{grover2011entanglement,bullivant2016entropic}, these coincide. We define the bulk topological entanglement entropy
\begin{align}
	\WWTEE:=S_{BC}+S_{CD}-S_B-S_D,\label{eqn:WWent}
\end{align}
where the regions are depicted in \cref{fig:WWregionsblk}. We obtained this choice of region following intuition given in Ref.~\cite{kim2015ground} where we consider the creating point excitations at the distinct parts of region $P$ using a string operator supported on $ACD$. We find $\WWTEE$ is zero only if all the excitations can be created using an operator with local support. The boundary diagnostics that we describe next are obtained by bisecting the regions shown in \cref{fig:WWregionsblk} along different planes where the boundary lies.

\begin{figure}
	\centering
	\includeTikz{bulk_point_A1}{
		\begin{tikzpicture}[scale=0.65]
			\begin{scope}
				\draw[black!20,dashed] (4.5,-.75,-.75)--(-4.5,-.75,-.75);
				\draw[black!20,dashed] (3,-.75,.75)--(-4.5,-.75,.75);
				\draw[black!20,dashed] (3,.75,-.75)--(-4.5,.75,-.75);
				\draw[black!20,dashed] (3,.75,.75)--(-4.5,.75,.75);
				\draw (4.5,-.75,.75)--(3,-.75,.75);
				\draw (4.5,.75,-.75)--(3,.75,-.75);
				\draw (4.5,.75,.75)--(3,.75,.75);
				\draw (-3.475,-.75,.75)--(-4.5,-.75,.75);
				\draw (-4.05,.75,-.75)--(-4.5,.75,-.75);
				\draw (-3.475,.75,.75)--(-4.5,.75,.75);
				\draw[black!20,dashed] (-4.5,-.75,.75)--(-4.5,.75,.75)--(-4.5,.75,-.75)--(-4.5,-.75,-.75)--cycle;
				\draw (4.5,-.75,.75)--(4.5,.75,.75)--(4.5,.75,-.75)--(4.5,-.75,-.75)--cycle;
				\draw[black!20,dashed] (1.5,-.75,.75)--(1.5,.75,.75)--(1.5,.75,-.75)--(1.5,-.75,-.75)--cycle;
				\draw[black!20,dashed] (-1.5,-.75,.75)--(-1.5,.75,.75)--(-1.5,.75,-.75)--(-1.5,-.75,-.75)--cycle;
				\draw (-4.5,-.75,.75)--(-4.5,.75,.75)--(-4.5,.75,-.75);
				\node[circle,inner sep=.01pt] at (0,0,.75) {$C$};
				\node[inner sep=0pt] at (-2.25,0,.75) {$D$};\node[inner sep=0] at (2.25,0,.75) {$D$};
				\node[circle,inner sep=.01pt] at (-4,0,.75) {$P$};\node[circle,inner sep=.01pt] at (3.75,0,.75) {$P$};
			\end{scope}
			\draw (-3,-2,2)--(-3,2,2)--(-3,2,-2);
			\draw (-3,-2,2)--(3,-2,2);
			\draw (-3,2,2)--(3,2,2);
			\draw (-3,2,-2)--(3,2,-2);
			\draw[black!20,dashed] (3,-.75,.75)--(-3,-.75,.75);
			\draw[black!20,dashed] (3,.75,.75)--(-3,.75,.75);
			\draw[black!20,dashed] (3,.75,-.75)--(-3,.75,-.75);
			\draw[black!20,dashed] (-3,-.75,-.75)--(-3,.75,-.75)--(-3,.75,.75)--(-3,-.75,.75)--cycle;
			\draw[black!20,dashed] (3,.75,-.75)--(3,-.75,-.75)--(3,-.75,.75);
			\draw (3,-.75,.75)--(3,.75,.75)--(3,.75,-.75)--(4.5,.75,-.75);
			\draw[] (3,.275,-2)--(3,2,-2)--(3,2,2)--(3,-2,2)--(3,-2,-2)--(3,-1.8,-2);
			\node[circle,inner sep=.1pt,fill=white] at (-.5,-1.75,0) {$B$};
			\node[circle,inner sep=.1pt,fill=white] at (6.4,2.5,5) {\color{white}.};
		\end{tikzpicture}
	}
	\caption{Partitioning of the lattice for detecting excitations in the bulk. $B$ encircles $CD$, and $A$ is the complement of $BCD$.
		If $\WWTEE$ is small, excitations on $P$ can be created by only acting on $PD$ and so have trivial statistics.}
	\label{fig:WWregionsblk}
\end{figure}
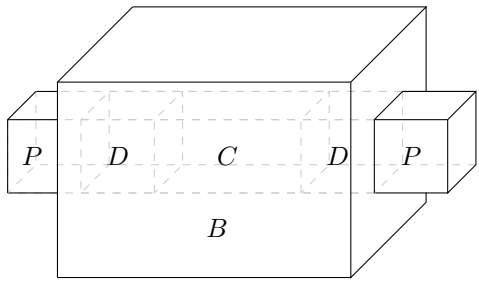

In Ref.~\cite{kim2015ground} two topological entanglement entropy diagnostics were found to probe long-range entanglement of a model near to a boundary. The first boundary diagnostic is an indicator that point-like topological particles can be created at the boundary of the system, and the second indicates that the boundary supports extended one-dimensional `loop-like' topological particles. Unlike in the bulk, these diagnostics do not necessarily coincide. The first
\begin{align}
	\WWbndTEEPoint:=S_{BC}+S_{CD}-S_B-S_D,\label{eqn:WWptdef}
\end{align}
defined using the regions in \cref{fig:WWregionsbnd_pt}, is non-zero if non-trivial point-like excitations can be created near the boundary.
If $\WWbndTEEPoint=0$, all point-like excitations on $P$ can be created with a local operator, so they are necessarily trivial. Conversely, if there are non-trivial point-like particles near the boundary, $\WWbndTEEPoint>0$.

The final diagnostic is designed to detect nontrivial loop-like excitations. Using the regions depicted in \cref{fig:WWregionsbnd_lp}, this diagnostic is
\begin{align}
	\WWbndTEELoop:=S_{BC}+S_{CD}-S_B-S_D.\label{eqn:WWloopdef}
\end{align}
Similarly to the other diagnostics, if $\WWbndTEELoop$ is zero, then line-like excitations must be trivial. Conversely, $\WWbndTEELoop$ must be nonzero if non-trivial loop excitations can be created at the boundary.

\begin{figure}
	\centering
	\includeTikz{Boundary_point_A}{
		\begin{tikzpicture}[scale=0.65]
			\draw (-3,-2,2)--(-3,0,2)--(-3,0,.75) (-3,0,-.75)--(-3,0,-2);
			\draw (-3,0,-.75)--(-3,-.6,-.75);
			\draw (3,-.75,-.75)--(-3,-.75,-.75);
			\fill[blue!20] (-3,0,2)--(3,0,2)--(3,0,.75)--(-3,0,.75)--cycle;
			\fill[blue!20] (-3,0,-2)--(3,0,-2)--(3,0,-.75)--(-3,0,-.75)--cycle;
			\filldraw[fill=white] (3,-2,2)--(3,0,2)--(3,0,.75)--(3,-.75,.75)--(3,-.75,-.75)--(3,0,-.75)--(3,0,-2)--(3,-2,-2)--cycle;
			\draw[black!20,dashed] (2.4,-.75,-.75)--(-3,-.75,-.75);
			\draw[black!20,dashed] (3,-.75,.75)--(-3,-.75,.75)--(-3,0,.75) (-3,-.75,.75)--(-3,-.75,-.75)--(-3,-.54,-.75);
			\draw (-3,-2,2)--(3,-2,2);
			\draw (-3,0,2)--(3,0,2);
			\draw (-3,0,-2)--(3,0,-2);
			\draw (-3,0,.75)--(3,0,.75);
			\draw (-3,0,-.75)--(3,0,-.75);
			\draw[] (3,-2,2)--(3,0,2)--(3,0,.75)--(3,-.75,.75)--(3,-.75,-.75)--(3,0,-.75)--(3,0,-2)--(3,-2,-2)--cycle;
			\node[circle,inner sep=.1pt,fill=white] at (-.5,-1.75,0) {$B$};
			\node[circle,inner sep=.1pt,fill=white] at (6.4,2.5,5) {\color{white}.};
		\end{tikzpicture}
	}
	\hspace{1.5cm}
	\includeTikz{Boundary_point_B}{
		\begin{tikzpicture}[scale=0.65]
			\begin{scope}[shift={(-7.5,0,0)}]
				\filldraw[fill=white] (3,-.75,-.75)--(4.5,-.75,-.75)--(4.5,0,-.75)--(3,0,-.75)--cycle;
				\filldraw[fill=white] (3,-.75,.75)--(4.5,-.75,.75)--(4.5,0,.75)--(3,0,.75)--cycle;
				\filldraw[fill=white] (4.5,-.75,-.75)--(4.5,-.75,.75)--(4.5,0,.75)--(4.5,0,-.75)--cycle;
				\filldraw[fill=blue!20] (3,0,-.75)--(4.5,0,-.75)--(4.5,0,.75)--(3,0,.75)--cycle;
			\end{scope}
			\filldraw[fill=white] (-3,0,2)--(3,0,2)--(3,0,.75)--(-3,0,.75)--cycle;
			\fill[blue!20] (-3,0,-2)--(3,0,-2)--(3,0,2)--(-3,0,2)--cycle;
			\draw (-1.5,0,-.75)--(-1.5,0,.75);\draw (1.5,0,-.75)--(1.5,0,.75);
			\draw (-3,0,-.75)--(-3,0,.75);\draw (3,0,-.75)--(3,0,.75);
			\draw (-3,-2,2)--(-3,0,2)--(-3,0,.75)--(-3,0,-.75)--(-3,0,-2);
			\fill[white] (-3,-2,2)--(3,-2,2)--(3,0,2)--(-3,0,2)--cycle;
			\draw (-3,-2,2)--(3,-2,2);\draw (-3,0,2)--(3,0,2);\draw (-3,0,-2)--(3,0,-2);
			\draw (-3,0,.75)--(3,0,.75);\draw (-3,0,-.75)--(3,0,-.75);
			\draw (3,-2,2)--(3,0,2)--(3,0,.75)--(3,0,-.75)--(3,0,-2)--(3,-2,-2)--cycle;
			\begin{scope}
				\filldraw[fill=white] (3,-.75,-.75)--(4.5,-.75,-.75)--(4.5,0,-.75)--(3,0,-.75)--cycle;
				\filldraw[fill=white] (3,-.75,.75)--(4.5,-.75,.75)--(4.5,0,.75)--(3,0,.75)--cycle;
				\filldraw[fill=white] (4.5,-.75,-.75)--(4.5,-.75,.75)--(4.5,0,.75)--(4.5,0,-.75)--cycle;
				\filldraw[fill=blue!20] (3,0,-.75)--(4.5,0,-.75)--(4.5,0,.75)--(3,0,.75)--cycle;
			\end{scope}
			\draw[] (-3,-2,2)--(3,-2,2)--(3,0,2)--(-3,0,2)--cycle;
			\node[circle,inner sep=.01pt,fill=white] at (-.5,-1.75,0) {$B$};
			\node[circle,inner sep=.01pt,fill=blue!20] at (0,0,0) {$C$};
			\node[inner sep=0pt,fill=blue!20] at (-2.25,0,0) {$D$};\node[inner sep=0,fill=blue!20] at (2.25,0,0) {$D$};
			\node[circle,inner sep=.01pt,fill=blue!20] at (-3.75,0,0) {$P$};\node[circle,inner sep=.01pt,fill=blue!20] at (3.75,0,0) {$P$};
		\end{tikzpicture}
	}
	\caption{Partitioning of the lattice for detecting point-like excitations on the boundary. The top (blue) surface is on the physical boundary of the lattice.
		If $\WWbndTEEPoint$ is small, excitations on $P$ can be created by only acting on $PD$ and so have trivial statistics.}
	\label{fig:WWregionsbnd_pt}
\end{figure}

\begin{figure}
	\centering
	\includeTikz{Boundary_line_A}{
		\begin{tikzpicture}[scale=0.65]
			\draw[fill=blue!20] (-2,0,-2)--(-2,0,2)--(2,0,2)--(2,0,-2)--cycle;
			\draw[fill=blue!20] (-1,0,-1)--(-1,0,1)--(1,0,1)--(1,0,-1)--cycle;
			\draw (2,-2,-2)--(2,0,-2) (2,-2,2)--(2,0,2) (-2,-2,2)--(-2,0,2);
			\draw(-2,-2,2)--(2,-2,2)--(2,-2,-2);
			\draw(-2,-1,2)--(2,-1,2)--(2,-1,-2);
			\node[] at (0,0,0) {$C$};
			\node[] at (0,-.5,2) {$B$};
			\node[] at (0,-1.5,2) {$D$};
		\end{tikzpicture}
	}
	\hspace{1.5cm}
	\includeTikz{Boundary_line_B}{
		\begin{tikzpicture}[scale=0.65]
			\draw[fill=blue!20] (-2,-1,-2)--(-2,-1,2)--(2,-1,2)--(2,-1,-2)--cycle;
			\draw[fill=white] (-1,-1,-1)--(-1,-1,1)--(1,-1,1)--(1,-1,-1)--cycle;
			\draw (2,-2,-2)--(2,-1,-2) (2,-2,2)--(2,-1,2) (-2,-2,2)--(-2,-1,2);
			\draw(-2,-1,-2)--(-2,-1,2)--(2,-1,2)--(2,-1,-2)--cycle;
			\draw (-2,-2,2)--(2,-2,2)--(2,-2,-2);
			\begin{scope}
				\clip (-1,-1,-1)--(-1,-1,1)--(1,-1,1)--(1,-1,-1)--cycle;
				\draw (-1,-1,-1)--(-1,-2,-1);
			\end{scope}
			\node[] at (0,-1.5,2) {$B$};
		\end{tikzpicture}
	}
	\hspace{1.5cm}
	\includeTikz{Boundary_line_C}{
		\begin{tikzpicture}[scale=0.65]
			\draw (2,-2,-2)--(2,-1,-2) (2,-2,2)--(2,-1,2) (-2,-2,2)--(-2,-1,2);
			\draw(-2,-1,-2)--(-2,-1,2)--(2,-1,2)--(2,-1,-2)--cycle;
			\draw (-2,-2,2)--(2,-2,2)--(2,-2,-2);
			\node[] at (0,-1.5,2) {$D$};
		\end{tikzpicture}
	}
	\caption{Partitioning of the lattice for detecting line-like excitations on the boundary.
		The top (blue) surface is on the boundary of the lattice.
		If $\WWbndTEELoop$ is small, excitations on $B$ can be created without acting on $C$ and so have trivial statistics.}\label{fig:WWregionsbnd_lp}
\end{figure}

The diagnostics presented in \onlinecite{kim2015ground} were found using generic arguments about the support of deformable operators that are used to create excitations. As such, it was shown rigorously that the null outcome is obtained only if a boundary does not give rise to topological particles. Conversely, a boundary that gives rise to topological excitations must give a positive reading for these diagnostics. However, due to spurious contributions~\cite{BravyiUnpublished,Cano15, Zou16, Williamson19, Kato19}, the generic arguments cannot guarantee that the diagnostics do not give false positives and, moreover, the work gives no interpretation for the magnitude of a positive reading. In our work, we restrict to loop-gas models. In that setting, for a large class of models, we obtain expressions for the topological entanglement entropy near the boundary.
% !TeX spellcheck = en_US
% This line sets the project root file.
% !TEX root = ../WW_EntanglementEntropy.tex
%

\section{Bulk entropy of topological loop-gasses}\label{sec:bulkentropy}

We now show how the entanglement entropy of ground states of Levin-Wen models is computed far from any boundary, before moving on to Walker-Wang models. To make the calculation we take the Schmidt decomposition of the ground state
\begin{align}
	\ket{\psi} & =\sum_{\lambda=1}^{r}\Phi_{\lambda} \ket{\psi_R^\lambda}\ket{\psi_{\comp{R}}^\lambda},\label{eqn:Schmidt}
\end{align}
for regions $R$, where the sets $\{\ket{\psi_R^\lambda}\}$ and  $\{\ket{\psi_{\comp{R}}^\lambda}\}$ are orthonormal, and $r$ is the Schmidt rank of the state $\ket{\psi}$. This allows us to compute the reduced state $\rho_R$ on $R$. Diagonalizing this matrix yields the entanglement entropy.

In what follows, we will need to parameterize the states $\ket{\psi_R^\lambda}$ and $\ket{\psi_{\comp{R}}^\lambda}$. Recall from \cref{sec:models} that ground states of the loop-gas models can be understood as classes of diagrams which are related by local moves. It is convenient to parameterize $\ket{\psi_R^\lambda}$ in a similar way. Far from the interface (since the correlation length is 0, far means one site), the state behaves exactly like the ground state. Unlike in the bulk, the interface defines a fixed boundary condition for the diagram in $R$. States $\ket{\psi_R^\lambda}$ will therefore be represented by some fiducial diagram $T$, and are understood to consist of a superposition of all diagrams that can be obtained from $T$ by local moves \emph{restricted to $R$} as indicated in \cref{fig:genericToTree}. Particular lattices may provide geometric complications, but the topological invariance of the ground state will mean these are of no consequence. In all cases, we will choose a particular class of fusion trees as fiducial diagrams.
For the following, we will therefore need several results concerning fusion trees. Consider fusing $n$ strings labeled $\vec{x}:=(x_1,x_2,\ldots,x_n)$ to a fixed object $a$. Using $F$-moves, we can bring the fusion tree for this process into the canonical form
\begin{align}
	\begin{array}{c}
		\includeTikz{treeA}{
			\begin{tikzpicture}
				\draw(0,0)--(1.75,1.75);
				\draw[dotted](1.75,1.75)--(2,2);
				\draw(2,2)--(3,3);
				\begin{scope}
					\clip(0,0)--(3,3)--(6,3)--(6,0)--(0,0);
					\draw(1,0)--(0,1);
					\draw(2,0)--(0,2);
					\draw(3,0)--(0,3);
					\draw(4.5,0)--(0,4.5);
					\draw(5.5,0)--(0,5.5);
				\end{scope}
				\node[below] at (0,0) {$x_1$};
				\node[below] at (1,0) {$x_2$};
				\node[below] at (2,0) {$x_3$};
				\node[below] at (3,0) {$x_4$};
				\node[below] at (4.5,0) {$x_{n-1}$};
				\node[below] at (5.5,0) {$x_{n}$};
				\node[above right] at (3,3) {$a$};
				\node[above left] at (.75,.75) {$y_1$};
				\node[above left] at (1.25,1.25) {$y_2$};
				\node[above left] at (2.5,2.5) {$y_{n-2}$};
				\node[below] at (.5,.5) {\tiny{$\mu_1$}};
				\node[below] at (1,1) {\tiny{$\mu_2$}};
				\node[below] at (1.5,1.5) {\tiny{$\mu_3$}};
				\node[right] at (2.25,2.25) {\tiny{$\mu_{n-2}$}};
			\end{tikzpicture}
		}
	\end{array},\label{eqn:treeA}
\end{align}
where $1\leq\mu\leq N_{a,b}^{c}$ parameterizes the distinct fusion channels $a\otimes b\to c$. In the following, sums over $x_i,\,y_i$ are over all simple objects in $\C$.
First, we need two results concerning summing over trees.
\begin{lemma}\label{lem:summingds}
	Let $\C$ a unitary fusion category, then for a fixed simple fusion outcome $a$,
	\begin{align}
		\sum_{\vec{x},\vec{y}}N_{x_1x_2}^{y_1}N_{y_1x_3}^{y_2}\ldots N_{y_{n-2}x_n}^{a}\prod_{j\leq n}d_{x_j} & =d_a\D^{2(n-1)},
	\end{align}
	where $\D=\sqrt{\sum_i d_i^2}$ is the total quantum dimension of $\C$.
	\begin{proof}
		Provided in Appendix~\hyperref[lem:summingds_pf]{\ref*{app:SN_results}}.
	\end{proof}
\end{lemma}

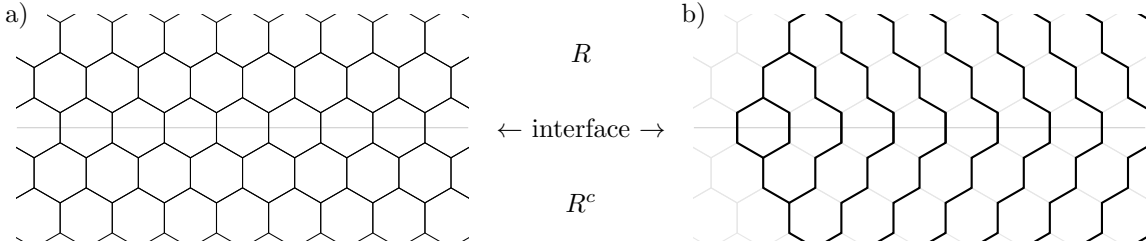
\begin{figure}
	\includeTikz{interfaceConfig}{
		\begin{tikzpicture}
			\def\r{.4};
			\pgfmathsetmacro{\dx}{\r*cos(30)};
			\begin{scope}
				\clip (-3,-1.5) rectangle (3,1.5);
				\draw[black!20] (-3,0)--(3,0);
				\begin{scope}[shift={(0,0)}]
					\draw[shift={(10*\dx,0)}] (210:\r)--(150:\r)--(90:\r)--(30:\r)--(-30:\r)--(-90:\r)--cycle;
					\draw[shift={(8*\dx,0)}] (210:\r)--(150:\r)--(90:\r)--(30:\r)--(-30:\r)--(-90:\r)--cycle;
					\draw[shift={(6*\dx,0)}] (210:\r)--(150:\r)--(90:\r)--(30:\r)--(-30:\r)--(-90:\r)--cycle;
					\draw[shift={(4*\dx,0)}] (210:\r)--(150:\r)--(90:\r)--(30:\r)--(-30:\r)--(-90:\r)--cycle;
					\draw[shift={(2*\dx,0)}] (210:\r)--(150:\r)--(90:\r)--(30:\r)--(-30:\r)--(-90:\r)--cycle;
					\draw[shift={(0,0)}] (210:\r)--(150:\r)--(90:\r)--(30:\r)--(-30:\r)--(-90:\r)--cycle;
					\draw[shift={(-2*\dx,0)}] (210:\r)--(150:\r)--(90:\r)--(30:\r)--(-30:\r)--(-90:\r)--cycle;
					\draw[shift={(-4*\dx,0)}] (210:\r)--(150:\r)--(90:\r)--(30:\r)--(-30:\r)--(-90:\r)--cycle;
					\draw[shift={(-6*\dx,0)}] (210:\r)--(150:\r)--(90:\r)--(30:\r)--(-30:\r)--(-90:\r)--cycle;
					\draw[shift={(-8*\dx,0)}] (210:\r)--(150:\r)--(90:\r)--(30:\r)--(-30:\r)--(-90:\r)--cycle;
					\draw[shift={(-10*\dx,0)}] (210:\r)--(150:\r)--(90:\r)--(30:\r)--(-30:\r)--(-90:\r)--cycle;
				\end{scope}
				\begin{scope}[shift={(\dx,3*\r/2)}]
					\draw[shift={(10*\dx,0)}] (210:\r)--(150:\r)--(90:\r)--(30:\r)--(-30:\r)--(-90:\r)--cycle;
					\draw[shift={(8*\dx,0)}] (210:\r)--(150:\r)--(90:\r)--(30:\r)--(-30:\r)--(-90:\r)--cycle;
					\draw[shift={(6*\dx,0)}] (210:\r)--(150:\r)--(90:\r)--(30:\r)--(-30:\r)--(-90:\r)--cycle;
					\draw[shift={(4*\dx,0)}] (210:\r)--(150:\r)--(90:\r)--(30:\r)--(-30:\r)--(-90:\r)--cycle;
					\draw[shift={(2*\dx,0)}] (210:\r)--(150:\r)--(90:\r)--(30:\r)--(-30:\r)--(-90:\r)--cycle;
					\draw[shift={(0,0)}] (210:\r)--(150:\r)--(90:\r)--(30:\r)--(-30:\r)--(-90:\r)--cycle;
					\draw[shift={(-2*\dx,0)}] (210:\r)--(150:\r)--(90:\r)--(30:\r)--(-30:\r)--(-90:\r)--cycle;
					\draw[shift={(-4*\dx,0)}] (210:\r)--(150:\r)--(90:\r)--(30:\r)--(-30:\r)--(-90:\r)--cycle;
					\draw[shift={(-6*\dx,0)}] (210:\r)--(150:\r)--(90:\r)--(30:\r)--(-30:\r)--(-90:\r)--cycle;
					\draw[shift={(-8*\dx,0)}] (210:\r)--(150:\r)--(90:\r)--(30:\r)--(-30:\r)--(-90:\r)--cycle;
					\draw[shift={(-10*\dx,0)}] (210:\r)--(150:\r)--(90:\r)--(30:\r)--(-30:\r)--(-90:\r)--cycle;
				\end{scope}
				\begin{scope}[shift={(\dx,-3*\r/2)}]
					\draw[shift={(10*\dx,0)}] (210:\r)--(150:\r)--(90:\r)--(30:\r)--(-30:\r)--(-90:\r)--cycle;
					\draw[shift={(8*\dx,0)}] (210:\r)--(150:\r)--(90:\r)--(30:\r)--(-30:\r)--(-90:\r)--cycle;
					\draw[shift={(6*\dx,0)}] (210:\r)--(150:\r)--(90:\r)--(30:\r)--(-30:\r)--(-90:\r)--cycle;
					\draw[shift={(4*\dx,0)}] (210:\r)--(150:\r)--(90:\r)--(30:\r)--(-30:\r)--(-90:\r)--cycle;
					\draw[shift={(2*\dx,0)}] (210:\r)--(150:\r)--(90:\r)--(30:\r)--(-30:\r)--(-90:\r)--cycle;
					\draw[shift={(0,0)}] (210:\r)--(150:\r)--(90:\r)--(30:\r)--(-30:\r)--(-90:\r)--cycle;
					\draw[shift={(-2*\dx,0)}] (210:\r)--(150:\r)--(90:\r)--(30:\r)--(-30:\r)--(-90:\r)--cycle;
					\draw[shift={(-4*\dx,0)}] (210:\r)--(150:\r)--(90:\r)--(30:\r)--(-30:\r)--(-90:\r)--cycle;
					\draw[shift={(-6*\dx,0)}] (210:\r)--(150:\r)--(90:\r)--(30:\r)--(-30:\r)--(-90:\r)--cycle;
					\draw[shift={(-8*\dx,0)}] (210:\r)--(150:\r)--(90:\r)--(30:\r)--(-30:\r)--(-90:\r)--cycle;
					\draw[shift={(-10*\dx,0)}] (210:\r)--(150:\r)--(90:\r)--(30:\r)--(-30:\r)--(-90:\r)--cycle;
				\end{scope}
				\begin{scope}[shift={(0,2*3*\r/2)}]
					\draw[shift={(10*\dx,0)}] (210:\r)--(150:\r)--(90:\r)--(30:\r)--(-30:\r)--(-90:\r)--cycle;
					\draw[shift={(8*\dx,0)}] (210:\r)--(150:\r)--(90:\r)--(30:\r)--(-30:\r)--(-90:\r)--cycle;
					\draw[shift={(6*\dx,0)}] (210:\r)--(150:\r)--(90:\r)--(30:\r)--(-30:\r)--(-90:\r)--cycle;
					\draw[shift={(4*\dx,0)}] (210:\r)--(150:\r)--(90:\r)--(30:\r)--(-30:\r)--(-90:\r)--cycle;
					\draw[shift={(2*\dx,0)}] (210:\r)--(150:\r)--(90:\r)--(30:\r)--(-30:\r)--(-90:\r)--cycle;
					\draw[shift={(0,0)}] (210:\r)--(150:\r)--(90:\r)--(30:\r)--(-30:\r)--(-90:\r)--cycle;
					\draw[shift={(-2*\dx,0)}] (210:\r)--(150:\r)--(90:\r)--(30:\r)--(-30:\r)--(-90:\r)--cycle;
					\draw[shift={(-4*\dx,0)}] (210:\r)--(150:\r)--(90:\r)--(30:\r)--(-30:\r)--(-90:\r)--cycle;
					\draw[shift={(-6*\dx,0)}] (210:\r)--(150:\r)--(90:\r)--(30:\r)--(-30:\r)--(-90:\r)--cycle;
					\draw[shift={(-8*\dx,0)}] (210:\r)--(150:\r)--(90:\r)--(30:\r)--(-30:\r)--(-90:\r)--cycle;
					\draw[shift={(-10*\dx,0)}] (210:\r)--(150:\r)--(90:\r)--(30:\r)--(-30:\r)--(-90:\r)--cycle;
				\end{scope}
				\begin{scope}[shift={(0,-2*3*\r/2)}]
					\draw[shift={(10*\dx,0)}] (210:\r)--(150:\r)--(90:\r)--(30:\r)--(-30:\r)--(-90:\r)--cycle;
					\draw[shift={(8*\dx,0)}] (210:\r)--(150:\r)--(90:\r)--(30:\r)--(-30:\r)--(-90:\r)--cycle;
					\draw[shift={(6*\dx,0)}] (210:\r)--(150:\r)--(90:\r)--(30:\r)--(-30:\r)--(-90:\r)--cycle;
					\draw[shift={(4*\dx,0)}] (210:\r)--(150:\r)--(90:\r)--(30:\r)--(-30:\r)--(-90:\r)--cycle;
					\draw[shift={(2*\dx,0)}] (210:\r)--(150:\r)--(90:\r)--(30:\r)--(-30:\r)--(-90:\r)--cycle;
					\draw[shift={(0,0)}] (210:\r)--(150:\r)--(90:\r)--(30:\r)--(-30:\r)--(-90:\r)--cycle;
					\draw[shift={(-2*\dx,0)}] (210:\r)--(150:\r)--(90:\r)--(30:\r)--(-30:\r)--(-90:\r)--cycle;
					\draw[shift={(-4*\dx,0)}] (210:\r)--(150:\r)--(90:\r)--(30:\r)--(-30:\r)--(-90:\r)--cycle;
					\draw[shift={(-6*\dx,0)}] (210:\r)--(150:\r)--(90:\r)--(30:\r)--(-30:\r)--(-90:\r)--cycle;
					\draw[shift={(-8*\dx,0)}] (210:\r)--(150:\r)--(90:\r)--(30:\r)--(-30:\r)--(-90:\r)--cycle;
					\draw[shift={(-10*\dx,0)}] (210:\r)--(150:\r)--(90:\r)--(30:\r)--(-30:\r)--(-90:\r)--cycle;
				\end{scope}
			\end{scope}
			\begin{scope}[shift = {(9,0)}]
				\clip (-3,-1.5) rectangle (3,1.5);
				\draw[black!20] (-3,0)--(3,0);
				\begin{scope}[shift={(0,0)}]
					\draw[black!10,shift={(10*\dx,0)}] (210:\r)--(150:\r)--(90:\r)--(30:\r)--(-30:\r)--(-90:\r)--cycle;
					\draw[black!10,shift={(8*\dx,0)}] (210:\r)--(150:\r)--(90:\r)--(30:\r)--(-30:\r)--(-90:\r)--cycle;
					\draw[black!10,shift={(6*\dx,0)}] (210:\r)--(150:\r)--(90:\r)--(30:\r)--(-30:\r)--(-90:\r)--cycle;
					\draw[black!10,shift={(4*\dx,0)}] (210:\r)--(150:\r)--(90:\r)--(30:\r)--(-30:\r)--(-90:\r)--cycle;
					\draw[black!10,shift={(2*\dx,0)}] (210:\r)--(150:\r)--(90:\r)--(30:\r)--(-30:\r)--(-90:\r)--cycle;
					\draw[black!10,shift={(0,0)}] (210:\r)--(150:\r)--(90:\r)--(30:\r)--(-30:\r)--(-90:\r)--cycle;
					\draw[black!10,shift={(-2*\dx,0)}] (210:\r)--(150:\r)--(90:\r)--(30:\r)--(-30:\r)--(-90:\r)--cycle;
					\draw[black!10,shift={(-4*\dx,0)}] (210:\r)--(150:\r)--(90:\r)--(30:\r)--(-30:\r)--(-90:\r)--cycle;
					\draw[black!10,shift={(-6*\dx,0)}] (210:\r)--(150:\r)--(90:\r)--(30:\r)--(-30:\r)--(-90:\r)--cycle;
					\draw[black!10,shift={(-8*\dx,0)}] (210:\r)--(150:\r)--(90:\r)--(30:\r)--(-30:\r)--(-90:\r)--cycle;
					\draw[black!10,shift={(-10*\dx,0)}] (210:\r)--(150:\r)--(90:\r)--(30:\r)--(-30:\r)--(-90:\r)--cycle;
				\end{scope}
				\begin{scope}[shift={(\dx,3*\r/2)}]
					\draw[black!10,shift={(10*\dx,0)}] (210:\r)--(150:\r)--(90:\r)--(30:\r)--(-30:\r)--(-90:\r)--cycle;
					\draw[black!10,shift={(8*\dx,0)}] (210:\r)--(150:\r)--(90:\r)--(30:\r)--(-30:\r)--(-90:\r)--cycle;
					\draw[black!10,shift={(6*\dx,0)}] (210:\r)--(150:\r)--(90:\r)--(30:\r)--(-30:\r)--(-90:\r)--cycle;
					\draw[black!10,shift={(4*\dx,0)}] (210:\r)--(150:\r)--(90:\r)--(30:\r)--(-30:\r)--(-90:\r)--cycle;
					\draw[black!10,shift={(2*\dx,0)}] (210:\r)--(150:\r)--(90:\r)--(30:\r)--(-30:\r)--(-90:\r)--cycle;
					\draw[black!10,shift={(0,0)}] (210:\r)--(150:\r)--(90:\r)--(30:\r)--(-30:\r)--(-90:\r)--cycle;
					\draw[black!10,shift={(-2*\dx,0)}] (210:\r)--(150:\r)--(90:\r)--(30:\r)--(-30:\r)--(-90:\r)--cycle;
					\draw[black!10,shift={(-4*\dx,0)}] (210:\r)--(150:\r)--(90:\r)--(30:\r)--(-30:\r)--(-90:\r)--cycle;
					\draw[black!10,shift={(-6*\dx,0)}] (210:\r)--(150:\r)--(90:\r)--(30:\r)--(-30:\r)--(-90:\r)--cycle;
					\draw[black!10,shift={(-8*\dx,0)}] (210:\r)--(150:\r)--(90:\r)--(30:\r)--(-30:\r)--(-90:\r)--cycle;
					\draw[black!10,shift={(-10*\dx,0)}] (210:\r)--(150:\r)--(90:\r)--(30:\r)--(-30:\r)--(-90:\r)--cycle;
				\end{scope}
				\begin{scope}[shift={(\dx,-3*\r/2)}]
					\draw[black!10,shift={(10*\dx,0)}] (210:\r)--(150:\r)--(90:\r)--(30:\r)--(-30:\r)--(-90:\r)--cycle;
					\draw[black!10,shift={(8*\dx,0)}] (210:\r)--(150:\r)--(90:\r)--(30:\r)--(-30:\r)--(-90:\r)--cycle;
					\draw[black!10,shift={(6*\dx,0)}] (210:\r)--(150:\r)--(90:\r)--(30:\r)--(-30:\r)--(-90:\r)--cycle;
					\draw[black!10,shift={(4*\dx,0)}] (210:\r)--(150:\r)--(90:\r)--(30:\r)--(-30:\r)--(-90:\r)--cycle;
					\draw[black!10,shift={(2*\dx,0)}] (210:\r)--(150:\r)--(90:\r)--(30:\r)--(-30:\r)--(-90:\r)--cycle;
					\draw[black!10,shift={(0,0)}] (210:\r)--(150:\r)--(90:\r)--(30:\r)--(-30:\r)--(-90:\r)--cycle;
					\draw[black!10,shift={(-2*\dx,0)}] (210:\r)--(150:\r)--(90:\r)--(30:\r)--(-30:\r)--(-90:\r)--cycle;
					\draw[black!10,shift={(-4*\dx,0)}] (210:\r)--(150:\r)--(90:\r)--(30:\r)--(-30:\r)--(-90:\r)--cycle;
					\draw[black!10,shift={(-6*\dx,0)}] (210:\r)--(150:\r)--(90:\r)--(30:\r)--(-30:\r)--(-90:\r)--cycle;
					\draw[black!10,shift={(-8*\dx,0)}] (210:\r)--(150:\r)--(90:\r)--(30:\r)--(-30:\r)--(-90:\r)--cycle;
					\draw[black!10,shift={(-10*\dx,0)}] (210:\r)--(150:\r)--(90:\r)--(30:\r)--(-30:\r)--(-90:\r)--cycle;
				\end{scope}
				\begin{scope}[shift={(0,2*3*\r/2)}]
					\draw[black!10,shift={(10*\dx,0)}] (210:\r)--(150:\r)--(90:\r)--(30:\r)--(-30:\r)--(-90:\r)--cycle;
					\draw[black!10,shift={(8*\dx,0)}] (210:\r)--(150:\r)--(90:\r)--(30:\r)--(-30:\r)--(-90:\r)--cycle;
					\draw[black!10,shift={(6*\dx,0)}] (210:\r)--(150:\r)--(90:\r)--(30:\r)--(-30:\r)--(-90:\r)--cycle;
					\draw[black!10,shift={(4*\dx,0)}] (210:\r)--(150:\r)--(90:\r)--(30:\r)--(-30:\r)--(-90:\r)--cycle;
					\draw[black!10,shift={(2*\dx,0)}] (210:\r)--(150:\r)--(90:\r)--(30:\r)--(-30:\r)--(-90:\r)--cycle;
					\draw[black!10,shift={(0,0)}] (210:\r)--(150:\r)--(90:\r)--(30:\r)--(-30:\r)--(-90:\r)--cycle;
					\draw[black!10,shift={(-2*\dx,0)}] (210:\r)--(150:\r)--(90:\r)--(30:\r)--(-30:\r)--(-90:\r)--cycle;
					\draw[black!10,shift={(-4*\dx,0)}] (210:\r)--(150:\r)--(90:\r)--(30:\r)--(-30:\r)--(-90:\r)--cycle;
					\draw[black!10,shift={(-6*\dx,0)}] (210:\r)--(150:\r)--(90:\r)--(30:\r)--(-30:\r)--(-90:\r)--cycle;
					\draw[black!10,shift={(-8*\dx,0)}] (210:\r)--(150:\r)--(90:\r)--(30:\r)--(-30:\r)--(-90:\r)--cycle;
					\draw[black!10,shift={(-10*\dx,0)}] (210:\r)--(150:\r)--(90:\r)--(30:\r)--(-30:\r)--(-90:\r)--cycle;
				\end{scope}
				\begin{scope}[shift={(0,-2*3*\r/2)}]
					\draw[black!10,shift={(10*\dx,0)}] (210:\r)--(150:\r)--(90:\r)--(30:\r)--(-30:\r)--(-90:\r)--cycle;
					\draw[black!10,shift={(8*\dx,0)}] (210:\r)--(150:\r)--(90:\r)--(30:\r)--(-30:\r)--(-90:\r)--cycle;
					\draw[black!10,shift={(6*\dx,0)}] (210:\r)--(150:\r)--(90:\r)--(30:\r)--(-30:\r)--(-90:\r)--cycle;
					\draw[black!10,shift={(4*\dx,0)}] (210:\r)--(150:\r)--(90:\r)--(30:\r)--(-30:\r)--(-90:\r)--cycle;
					\draw[black!10,shift={(2*\dx,0)}] (210:\r)--(150:\r)--(90:\r)--(30:\r)--(-30:\r)--(-90:\r)--cycle;
					\draw[black!10,shift={(0,0)}] (210:\r)--(150:\r)--(90:\r)--(30:\r)--(-30:\r)--(-90:\r)--cycle;
					\draw[black!10,shift={(-2*\dx,0)}] (210:\r)--(150:\r)--(90:\r)--(30:\r)--(-30:\r)--(-90:\r)--cycle;
					\draw[black!10,shift={(-4*\dx,0)}] (210:\r)--(150:\r)--(90:\r)--(30:\r)--(-30:\r)--(-90:\r)--cycle;
					\draw[black!10,shift={(-6*\dx,0)}] (210:\r)--(150:\r)--(90:\r)--(30:\r)--(-30:\r)--(-90:\r)--cycle;
					\draw[black!10,shift={(-8*\dx,0)}] (210:\r)--(150:\r)--(90:\r)--(30:\r)--(-30:\r)--(-90:\r)--cycle;
					\draw[black!10,shift={(-10*\dx,0)}] (210:\r)--(150:\r)--(90:\r)--(30:\r)--(-30:\r)--(-90:\r)--cycle;
				\end{scope}
				\draw[thick] ($(-6*\dx,0) + (210:\r)$)--++(0,\r)--++(\dx,\r/2)--++(0,\r)--++(\dx,\r/2)--++(0,\r)--++(\dx,\r/2);
				\draw[thick] ($(-4*\dx,0) + (210:\r)$)--++(0,\r)--++(-\dx,\r/2);
				\draw[thick] ($(-2*\dx,0) + (210:\r)$)--++(0,\r)--++(-\dx,\r/2)--++(0,\r)--++(-\dx,\r/2);
				\draw[thick] ($(0*\dx,0) + (210:\r)$)--++(0,\r)--++(-\dx,\r/2)--++(0,\r)--++(-\dx,\r/2)--++(0,\r)--++(-\dx,\r/2);
				\draw[thick] ($(2*\dx,0) + (210:\r)$)--++(0,\r)--++(-\dx,\r/2)--++(0,\r)--++(-\dx,\r/2)--++(0,\r)--++(-\dx,\r/2);
				\draw[thick] ($(4*\dx,0) + (210:\r)$)--++(0,\r)--++(-\dx,\r/2)--++(0,\r)--++(-\dx,\r/2)--++(0,\r)--++(-\dx,\r/2);
				\draw[thick] ($(6*\dx,0) + (210:\r)$)--++(0,\r)--++(-\dx,\r/2)--++(0,\r)--++(-\dx,\r/2)--++(0,\r)--++(-\dx,\r/2);
				\draw[thick] ($(8*\dx,0) + (210:\r)$)--++(0,\r)--++(-\dx,\r/2)--++(0,\r)--++(-\dx,\r/2)--++(0,\r)--++(-\dx,\r/2);
				\draw[thick] ($(10*\dx,0) + (210:\r)$)--++(0,\r)--++(-\dx,\r/2)--++(0,\r)--++(-\dx,\r/2)--++(0,\r)--++(-\dx,\r/2);
				\draw[thick] ($(12*\dx,0) + (210:\r)$)--++(0,\r)--++(-\dx,\r/2)--++(0,\r)--++(-\dx,\r/2)--++(0,\r)--++(-\dx,\r/2);
				\begin{scope}[yscale=-1]
					\draw[thick] ($(-6*\dx,0) + (210:\r)$)--++(0,\r)--++(\dx,\r/2)--++(0,\r)--++(\dx,\r/2)--++(0,\r)--++(\dx,\r/2);
					\draw[thick] ($(-4*\dx,0) + (210:\r)$)--++(0,\r)--++(-\dx,\r/2);
					\draw[thick] ($(-2*\dx,0) + (210:\r)$)--++(0,\r)--++(-\dx,\r/2)--++(0,\r)--++(-\dx,\r/2);
					\draw[thick] ($(0*\dx,0) + (210:\r)$)--++(0,\r)--++(-\dx,\r/2)--++(0,\r)--++(-\dx,\r/2)--++(0,\r)--++(-\dx,\r/2);
					\draw[thick] ($(2*\dx,0) + (210:\r)$)--++(0,\r)--++(-\dx,\r/2)--++(0,\r)--++(-\dx,\r/2)--++(0,\r)--++(-\dx,\r/2);
					\draw[thick] ($(4*\dx,0) + (210:\r)$)--++(0,\r)--++(-\dx,\r/2)--++(0,\r)--++(-\dx,\r/2)--++(0,\r)--++(-\dx,\r/2);
					\draw[thick] ($(6*\dx,0) + (210:\r)$)--++(0,\r)--++(-\dx,\r/2)--++(0,\r)--++(-\dx,\r/2)--++(0,\r)--++(-\dx,\r/2);
					\draw[thick] ($(8*\dx,0) + (210:\r)$)--++(0,\r)--++(-\dx,\r/2)--++(0,\r)--++(-\dx,\r/2)--++(0,\r)--++(-\dx,\r/2);
					\draw[thick] ($(10*\dx,0) + (210:\r)$)--++(0,\r)--++(-\dx,\r/2)--++(0,\r)--++(-\dx,\r/2)--++(0,\r)--++(-\dx,\r/2);
					\draw[thick] ($(12*\dx,0) + (210:\r)$)--++(0,\r)--++(-\dx,\r/2)--++(0,\r)--++(-\dx,\r/2)--++(0,\r)--++(-\dx,\r/2);
				\end{scope}
			\end{scope}
			\node[fill=white,circle] at (-3,1.5) {a)};
			\node[fill=white,circle] at (-3+9,1.5) {b)};
			\node[fill=white,circle] at (4.5,0) {$\leftarrow$ interface $\rightarrow$};
			\node[fill=white,circle] at (4.5,1) {$R$};
			\node[fill=white,circle] at (4.5,-1) {$\comp{R}$};
		\end{tikzpicture}
	}
	{\phantomsubcaption\label{fig:interfaceGeneric}}
	{\phantomsubcaption\label{fig:interfaceTree}}
	\caption{To compute the Schmidt decomposition, we need to parameterize states contained on one sub-region. Given a generic configuration on the lattice (\subref*{fig:interfaceGeneric}), we can utilize the `moves' outlined in \cref{sec:preliminaries}, restricted to either side of the interface, to deform into a tree (\subref*{fig:interfaceTree}). The lattice sites on the boundary provide boundary conditions for the states in $R$ and $\comp{R}$. We utilize constraints on the total fusion outcome in each region to parameterize the allowed trees.
	}\label{fig:genericToTree}
\end{figure}

\begin{lemma}\label{lem:sumlog}
	Let $\C$ a unitary fusion category, then for a fixed simple fusion outcome $a$,
	\begin{align}
		\sum_{\vec{x},\vec{y}} N_{x_1x_2}^{y_1}N_{y_1x_3}^{y_2}\ldots N_{y_{n-2}x_n}^{a}\frac{\prod_{j\leq n} d_{x_j}}{\D^{2(n-1)}}\log\prod_{k\leq n} d_{x_k} & = n d_a \sum_x \frac{d_x^2\log d_x}{\D^2}.
	\end{align}
	\begin{proof}
		Provided in Appendix~\hyperref[lem:sumlog_pf]{\ref*{app:SN_results}}.
	\end{proof}
\end{lemma}

Finally, we need the probability of a given fusion tree in a topological loop-gas model.

\begin{lemma}[Probability of trees]\label{lem:prtree}
	Let $\C$ a unitary fusion category. Given a fusion outcome $a$ on $n$ edges, the probability of the tree in \cref{eqn:treeA} is
	\begin{align}
		\Pr[\vec{x},\vec{y},\vec{\mu}|a] & =\frac{\prod_{j\leq n} d_{x_j}}{d_{a}\D^{2(n-1)}}.
	\end{align}
	\begin{proof}
		Provided in Appendix~\hyperref[lem:prtree_pf]{\ref*{app:SN_results}}.
	\end{proof}
\end{lemma}

Throughout the remainder of this section, we use the following condensed notation
\begin{align}
	\sum_{\vec{x},\vec{y},\vec{\mu}}:= & \sum_{x_1,\ldots,x_n}\sum_{y_1,\ldots,y_{n-2}}\sum_{\mu_1,\ldots,\mu_{n-2}}                               \\
	=                                  & \sum_{x_1,\ldots,x_n}\sum_{y_1,\ldots,y_{n-2}} N_{x_1x_2}^{y_1}N_{y_1x_3}^{y_2}\ldots N_{y_{n-2}x_n}^{a},
\end{align}
where we frequently leave the fusion outcome $a$ implicit.

\subsection{Levin-Wen models}

\begin{thm}[Topological entropy of (2+1)-dimensional Levin-Wen models in the bulk~\cite{levin2006detecting,kitaev2006topological,LevinThesis}]\ \\
	Consider the regions shown in \cref{fig:LWregionsblk}, then the Levin-Wen model defined by a unitary spherical fusion category $\C$, with total dimension $\D$, has topological entropy
	\begin{align}
		\LWTEE & =2\log\D^2=\log\D^2_{\mathcal{Z}(\C)},
	\end{align}
	where $\drinfeld{\C}$ is the modular category called the \define{Drinfeld center}~\cite{MR3242743} of $\C$ which describes the anyons of the theory.
	\label{thm:LWbulk}
\end{thm}

\begin{examples*}
	Recall the examples from \cref{sec:examples}. As discussed in \cref{sec:examples_phys}, these label two distinct loop-gas models in (2+1)-dimensions, the toric code and double semion models. Since all the input categories for these examples have $\D^2=2$, the topological entanglement entropy is $\LWTEE=2\log2$ for both.
\end{examples*}

\begin{lemma}[Entropy of (union of) simply connected bulk regions~\cite{levin2006detecting,LevinThesis,bullivant2016entropic}]\label{lem:bulkLW}
	On a region $R$ in the bulk consisting of the disjoint union of simply connected sub-regions, the entropy is
	\begin{align}
		S_R & =nS[\C]-b_0\log\D^2,\label{eqn:LWsimpentropy}
	\end{align}
	where $b_0$ is the number of disjoint interface components of $R$, $n$ is the number of links crossing the entanglement interface, and
	\begin{align}
		S[\C]:=\log\D^2-\sum_{x}\frac{d_x^2\log d_x}{\D^2}.
	\end{align}
	\begin{proof}[Proof of \cref{lem:bulkLW}]
		Consider a ball $R$ with $n$ sites along the interface, in the configuration $\vec{x}=x_1,x_2,\ldots, x_n$. Since any configuration must be created by inserting closed loops into the empty state, the total `charge' crossing the interface must be $1$. For a fixed $\vec{x}$, there are now many ways for this to happen, parameterized by trees depicted in \cref{eqn:treeA} with fusion outcome $a=1$.

		Trees with distinct labelings (in $\vec{x}$, $\vec{y}$ or $\vec{\mu}$) are orthogonal. This means that if the tree \cref{eqn:treeA} occurs adjacent to the interface within $R$, it must also occur on the other side of the interface
		\begin{align}
			\begin{array}{c}
				\includeTikz{tree_interface_blk}{
					\begin{tikzpicture}[scale=.9]
						\begin{scope}
							\draw(0,0)--(1.75,1.75);
							\draw[dotted](1.75,1.75)--(2,2);
							\draw(2,2)--(3,3);
							\begin{scope}
								\clip(0,0)--(3,3)--(6,3)--(6,0)--(0,0);
								\draw(1,0)--(0,1);
								\draw(2,0)--(0,2);
								\draw(3,0)--(0,3);
								\draw(4.5,0)--(0,4.5);
								\draw(5.5,0)--(0,5.5);
							\end{scope}
							\node[below left] at (0,0) {$x_1$};
							\node[below right,inner sep=.75] at (1,0) {$x_2$};
							\node[below right,inner sep=.75] at (2,0) {$x_3$};
							\node[below right,inner sep=.75] at (3,0) {$x_4$};
							\node[below right,inner sep=.75] at (4.5,0) {$x_{n-1}$};
							\node[below right,inner sep=.75] at (5.5,0) {$x_{n}$};
							\node[above right] at (3,3) {$c$};
							\node[above left] at (.75,.75) {$y_1$};
							\node[above left] at (1.25,1.25) {$y_2$};
							\node[above left] at (2.5,2.5) {$y_{n-2}$};
							\node[below] at (.5,.5) {\tiny{$\mu_1$}};
							\node[below] at (1,1) {\tiny{$\mu_2$}};
							\node[below] at (1.5,1.5) {\tiny{$\mu_3$}};
							\node[right] at (2.25,2.25) {\tiny{$\mu_{n-2}$}};
							\node[right] at (2.75,2.75) {\tiny{$\mu_{n-1}$}};
						\end{scope}
						\begin{scope}[yscale=-1]
							\draw(0,0)--(1.75,1.75);
							\draw[dotted](1.75,1.75)--(2,2);
							\draw(2,2)--(3,3);
							\begin{scope}
								\clip(0,0)--(3,3)--(6,3)--(6,0)--(0,0);
								\draw(1,0)--(0,1);
								\draw(2,0)--(0,2);
								\draw(3,0)--(0,3);
								\draw(4.5,0)--(0,4.5);
								\draw(5.5,0)--(0,5.5);
							\end{scope}
							\node[below right] at (3,3) {$d$};
							\node[below left] at (.75,.75) {$z_1$};
							\node[below left] at (1.25,1.25) {$z_2$};
							\node[below left] at (2.5,2.5) {$z_{n-2}$};
							\node[above] at (.5,.5) {\tiny{$\nu_1$}};
							\node[above] at (1,1) {\tiny{$\nu_2$}};
							\node[above] at (1.5,1.5) {\tiny{$\nu_3$}};
							\node[right] at (2.25,2.25) {\tiny{$\nu_{n-2}$}};
							\node[right] at (2.75,2.75) {\tiny{$\nu_{n-1}$}};
						\end{scope}
						\draw[black!20](-.5,0)--(6,0);
						\node at (-1,2) {$R$};
						\node at (-1,-2) {$\comp{R}$};
					\end{tikzpicture}
				}
			\end{array}\propto \indicator{\vec{z}=\vec{y}}\indicator{\vec{\nu}=\vec{\mu}}\indicator{d=c}.\label{eqn:treeB_blk}
		\end{align}
		If the trees on either side of the cut had different branching structures, we could use local moves on either side of the cut to bring them to this standard form.

		We take the Schmidt decomposition of the ground state as follows
		\begin{align}
			\ket{\psi} & =\sum_{\vec{x},\vec{y},\vec{\mu}} \Phi_{\vec{x},\vec{y},\vec{\mu}} \ket{ \psi_R^{\vec{x},\vec{y},\vec{\mu}} }\ket{ \psi_{\comp{R}}^{\vec{x},\vec{y},\vec{\mu}} },\label{eqn:simpleregionpartition}
		\end{align}
		where the notation $\vec{x},\,\vec{y},\,\vec{\mu}$ indicates the labeling of a valid tree as in \cref{eqn:treeA}. The state $\ket{ \psi_R^{\vec{x},\vec{y},\vec{\mu}} }$ includes any state that can be reached from \cref{eqn:treeA} (with $a=1$) by acting only on $R$. The reduced state on $R$ is
		\begin{align}
			\rho_R & =\sum_{\vec{x},\vec{y},\vec{\mu}} |\Phi_{\vec{x},\vec{y},\vec{\mu}}|^2 \ketbra{ \psi_R^{\vec{x},\vec{y},\vec{\mu}} } \\
			       & =\sum_{\vec{x},\vec{y},\vec{\mu}} \Pr[\vec{x},\vec{y},\vec{\mu}|1] \ketbra{ \psi_R^{\vec{x},\vec{y},\vec{\mu}} },
		\end{align}
		where $\Pr[\vec{x},\vec{y},\vec{\mu}|1]$ is the probability of the labeled tree, given that $\vec{x}$ fuses to 1. From \cref{lem:prtree}, the reduced state is
		\begin{align}
			\rho_R & =\sum_{\vec{x},\vec{y},\vec{\mu}} \frac{\prod_{j\leq n} d_{x_j}}{\D^{2(n-1)}} \ketbra{ \psi_R^{\vec{x},\vec{y},\vec{\mu}} }.
		\end{align}

		The von Neumann entropy of $\rho_R$ is therefore
		\begin{align}
			S_R := & -\tr\rho_R\log\rho_R                                                                                                                                                                                                                        \\
			=      & -\sum_{\vec{x},\vec{y},\vec{\mu}} \frac{\prod_{j\leq n} d_{x_j}}{\D^{2(n-1)}} \log \frac{\prod_{k\leq n} d_{x_k}}{\D^{2(n-1)}}                                                                                                              \\
			=      & \frac{ \log \D^{2(n-1)}}{\D^{2(n-1)}} \sum_{\vec{x},\vec{y},\vec{\mu}} \prod_{j\leq n} d_{x_j} - \sum_{\vec{x},\vec{y},\vec{\mu}}  \frac{ \prod_{j\leq n} d_{x_j} }{\D^{2(n-1)}} \log \prod_{k\leq n} d_{x_k}   \label{EqnLine:ApplyLemmas} \\
			=      & n(\log\D^2-\sum_x \frac{d_x^2\log{d_x}}{\D^2})-\log\D^2                                                                                                                                                                                     \\
			=      & n S[C]-\log\D^2
		\end{align}
		where \cref{lem:summingds,lem:sumlog} are applied to the left and right terms of line~(\ref{EqnLine:ApplyLemmas}), respectively, and
		\begin{align}
			S[C]:=\log\D^2-\sum_x \frac{d_x^2\log{d_x}}{\D^2}.
		\end{align}

		It is straightforward to check that this holds on each sub-region of $R$.
	\end{proof}
\end{lemma}

Applying \cref{lem:bulkLW} to the regions in \cref{fig:LWregionsblk} completes the proof of \cref{thm:LWbulk}.

\subsection{Walker-Wang models}\label{ss:WWbulkpf}

In this section, we prove the following result for the bulk diagnostic for Walker-Wang models. The essential arguments in this section were made in \onlinecite{bullivant2016entropic}, however we use slightly different language that allows the result to be applied more generally.
\begin{thm}
	For a Walker-Wang model defined by a unitary premodular category $\C$, the topological entanglement entropy (defined using the regions in \cref{fig:WWregionsblk}) in the bulk is given by
	\begin{align}
		\WWTEE & =\sum_{c,\lambda_c}\frac{\lambda_c}{\D^2}\log \frac{\lambda_c}{d_c},
	\end{align}
	where $\{\lambda_c\}$ are the eigenvalues of $\S_c^\dagger \S_c$, and $\S_c$ is the connected $\S$-matrix (\cref{def:commectedS}).
\end{thm}

\begin{proof}
	In simply connected regions, the arguments from \cref{lem:bulkLW} still hold. The other type of region in \cref{fig:LWregionsblk} is a torus. In this case, we cannot simply decompose the ground state as in \cref{eqn:simpleregionpartition}, with the sum over configurations on the interface. Recall that the reason we could do this for a simple region was ground states are created by inserting closed loops, and all closed loops except those crossing the interface can be added entirely within either $R$ or $\comp{R}$. This is not the case for a toroidal region. Consider, for example, the configuration depicted in \cref{fig:toroid_config}. The closed string inside $R$ (red, dashed) cannot be altered by acting entirely within $R$, so contributes additional entanglement to the ground state, which is not witnessed by the interface configuration. Additionally, the two loops may be connected by a string, such that the global charge is trivial. Therefore, unlike for simply connected regions, the net charge crossing the boundary is not necessarily trivial. With these considerations, we can decompose the ground state as
	\begin{align}
		\ket{\psi}=\sum_{
		\substack{\vec{x},\vec{y},\vec{\mu}                                                                                                                                                                                              \\
		c,a,\alpha,b,\beta}} \Phi_{\vec{x},\vec{y},\vec{\mu},c} & \frac{\left[\S_c\right]_{(b,\beta)(a,\alpha)}}{\D} \ket{ \psi_R^{\vec{x},\vec{y},\vec{\mu},c,a,\alpha} }\ket{ \psi_{\comp{R}}^{\vec{x},\vec{y},\vec{\mu},c,b,\beta} },
	\end{align}
	where $\S_c$ is the connected $\S$-matrix defined in \cref{eqn:connectedS}. The indices $\vec{x}, \vec{y}, \vec{\mu}$ are as in \cref{eqn:simpleregionpartition}, $b$ labels the loop encircling $R$, while $a$ is the loop within $R$, and $c$ is the total charge crossing the boundary (the top label in \cref{eqn:treeA}). The reduced state on $R$ is
	\begin{alignat}{2}
		\rho_R                        & =\sum_{
		\substack{\vec{x},\vec{y},\vec{\mu}                                                                                                                                                                                                                      \\
		a_1,\alpha_1,a_2,\alpha_2,c}} & \frac{\Pr[\vec{x},\vec{y},\vec{\mu}|c]}{\D^2} \left[\S_c^\dagger\S_c\right]_{(a_2,\alpha_2)(a_1,\alpha_1)} \ketbra{\psi_R^{\vec{x},\vec{y},\vec{\mu},a_1,\alpha_1,c}}{\psi_R^{\vec{x},\vec{y},\vec{\mu},a_2,\alpha_2,c}} \\
		                              & =\sum_{
		\substack{\vec{x},\vec{y},\vec{\mu}                                                                                                                                                                                                                      \\
		a_1,\alpha_1,a_2,\alpha_2,c}} & \frac{\prod_{j\leq n} d_{x_j}}{d_{c}\D^{2n}} \left[\S_c^\dagger\S_c\right]_{(a_2,\alpha_2)(a_1,\alpha_1)} \ketbra{\psi_R^{\vec{x},\vec{y},\vec{\mu},a_1,\alpha_1,c}}{\psi_R^{\vec{x},\vec{y},\vec{\mu},a_2,\alpha_2,c}}.
	\end{alignat}
	To compute the entropy of this state, it is convenient to diagonalize it. Denote the eigenvalues of $\S_c^\dagger\S_c$ by $\{\lambda_c\}$. By a unitary change of basis, we have
	\begin{align}
		U\rho_R U^\dagger & =
		\sum_{\substack{\vec{x},\vec{y},\vec{\mu}, \\
				c,\lambda_c}}  \frac{\prod_{j\leq n} d_{x_j}}{d_{c}\D^{2n}} \lambda_c\ketbra{\varphi^{\vec{x},\vec{y},\vec{\mu},c}_{\lambda_c}},
	\end{align}
	with von Neumann entropy
	\begin{align}
		S_R & =nS[\C]-\sum_{c,\lambda_c}\frac{\lambda_c}{\D^2}\log\frac{\lambda_c}{d_c},
	\end{align}
	where \cref{lem:summingds,lem:sumlog,lem:TrScSc} are used. Combining with \cref{lem:bulkLW} completes the proof.
\end{proof}

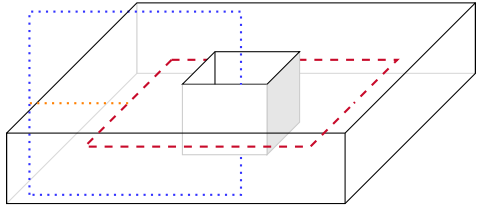
\begin{figure}
	\centering
	$\begin{array}{c}
			\includeTikz{torus}{
				\begin{tikzpicture}[scale=.75]
					\draw[gray!30] (6,-1.25,0) -- (0,-1.25,0) -- (0,-1.25,6);
					\draw[thick, dashed, darkred] (1,-0.625,1) -- (5,-0.625,1) -- (5,-0.625,3);
					\draw[fill=gray!20,draw=gray!40](3.75,-0.008,2.25) -- (3.75,-1.25,2.25) -- (3.75,-1.25,3.75) -- (3.75,-0.008,3.75);
					\draw[] (0,0,6) -- (0,-1.25,6) -- (6,-1.25,6) -- (6,-1.25,0) -- (6,0,0);
					\draw[] (6,-1.25,6) -- (6,0,6);
					\draw[] (0,-1.25,6) -- (0,0,6);
					\draw[] (6,-1.25,0) -- (6,0,0);
					\draw[gray!40] (0,-1.25,0) -- (0,0,0);
					\draw[fill=white] (2.25,0,2.25) -- (2.25,0,3.75) -- (3.75,0,3.75) -- (3.75,0,2.25) -- cycle;
					\draw[] (2.25,0,2.25) -- (2.25,-1,2.25);
					\draw[thick, dotted, blue!80] (-.75,-2.25,3) -- (-.75,1,3) -- (3,1,3)-- (3,-2.25,3) -- cycle;
					\draw[fill=white, draw=gray!40] (2.25,-0.008,3.75) -- (2.25,-1.25,3.75) -- (3.75,-1.25,3.75) -- (3.75,-0.008,3.75);
					\draw[] (0,0,0) -- (0,0,6) -- (6,0,6) -- (6,0,0) -- cycle;
					\draw[thick, dashed, darkred] (1,-0.625,1) -- (1,-0.625,5) -- (5,-0.625,5) -- (5,-0.625,3);
					\draw[thick,dotted,orange] (1,-0.625,3)--(-.75,-0.625,3);
				\end{tikzpicture}
			}
		\end{array}$
	\caption{When the region $R$ is not simply connected, the computation of entropy is more subtle. There is additional entanglement in the system due to intersecting loops that cannot be created in $R$ or $\comp{R}$ separately. This is not witnessed by the configuration of strings on the interface.}
	\label{fig:toroid_config}
\end{figure}

\begin{conjecture}
	Let $\C$ be a unitary premodular category $\C$, and define the connected $\S$-matrix via its matrix elements
	\begin{align}
		\left[\S_c\right]_{(a,\alpha),(b,\beta)} & =\frac{1}{\D}
		\begin{array}{c}
			\includeTikz{ConnectedSmatrix}{}
		\end{array}.
	\end{align}
	The connected $\S$-matrix obeys
	\begin{align}
		\sum_{c,\lambda_c}\frac{\lambda_c}{\D^2}\log \frac{\lambda_c}{d_c} & =\log\D_{\mug{\C}}^2,\label{eqn:conjecture}
	\end{align}
	where $\{\lambda_c\}$ are the eigenvalues of $\S_c^\dagger \S_c$, and $\mug{\C}$ is the M\"uger center of $\C$.
\end{conjecture}

We conjecture that \cref{eqn:conjecture} holds in general, however we are currently unable to compute the spectrum of $\S_c$ beyond the families outlined in \cref{thm:WWentropyexamples}.
\begin{thm}\label{thm:WWentropyexamples}
	For a Walker-Wang model defined by a unitary premodular category of one of the following types:
	\begin{itemize}
		\item $\C=\cat{A}\boxtimes\cat{B}$, where $\cat{A}$ is symmetric and $\cat{B}$ is modular~\cite{bullivant2016entropic},
		\item $\C$ pointed,
		\item $\rk{\C}<6$ and multiplicity free,
		\item $\rk{\C}=\rk{\mug{\C}}+1$ and $d_x=\D_{\mug{\C}}$, where $x$ is the additional object (as a special case, $\C$ is a Tambara-Yamagami category~\cite{Tambara_1998,Siehler}),
	\end{itemize}
	then \cref{eqn:conjecture} holds. As a consequence, the topological entanglement entropy (defined using the regions in \cref{fig:WWregionsblk}) in the bulk is given by
	\begin{align}
		\WWTEE & =\log\D_{\mug{\C}}^2,\label{eqn:BodyWalkerWangBulkTEE}
	\end{align}
	where $\mug{\C}$ is the M\"uger center of $\C$.
	As special cases, this includes
	\begin{align}
		\WWTEE_{\text{modular}}   & =0         \\
		\WWTEE_{\text{symmetric}} & =\log \D^2
	\end{align}
	We conjecture that \cref{eqn:BodyWalkerWangBulkTEE} holds in generality. Physically, this is seen by noting that the particle content of the bulk Walker-Wang model is given by the M\"uger center $\mug{\cdot}$~\cite{ZWang}.
	\begin{proof}
		Provided in Appendix~\hyperref[thm:WWentropyexamples_pf]{\ref*{app:SN_results}}.
	\end{proof}
\end{thm}

\begin{examples*}
	Recall the examples from \cref{sec:examples}. As discussed in \cref{sec:examples_phys}, these label four distinct loop-gas models in (3+1)-dimensions, the bosonic and fermionic toric code models, and two semion models. All four input categories are pointed, so we can apply \cref{thm:WWentropyexamples} to obtain the topological entanglement entropy. The first two models are symmetric, so $\WWTEE=\log 2$ for both. The inputs to the semion models are modular, so the bulk is trivial~\cite{von2013three}. In this case $\WWTEE=0$.
\end{examples*}
% !TeX spellcheck = en_US
% This line sets the project root file.
% !TEX root = ../WW_EntanglementEntropy.tex
%

\section{Boundary entropy of topological loop-gasses}\label{sec:boundaryentropy}

We now turn to the computation of the boundary diagnostics from \cref{sec:entropydiagnostics}. As before, we begin with Levin-Wen models.

\subsection{Levin-Wen models}

\begin{thm}[Topological entropy of (2+1){-dimensional} Levin-Wen models at a boundary]\ \\
	Consider the regions shown in \cref{fig:LWregionsbnd}. The Levin-Wen model defined by unitary spherical fusion category $\C$, with boundary specified by an indecomposable, strongly separable, special, Frobenius algebra $A\in\C$ has boundary entropy
	\begin{align}
		\LWbndTEE & =\log\D^2,
	\end{align}
	where $\D$ is the total quantum dimension of $\C$.

	\label{thm:LWbnd}
\end{thm}

\begin{examples*}
	Recall the examples from \cref{sec:examples}. As discussed in \cref{sec:examples_phys}, these label two distinct loop-gas models in (2+1)-dimensions, the toric code, and the double semion. The toric code has two possible boundary conditions, while the double semion only allows for the trivial boundary. All boundaries have $\LWbndTEE=\log 2$.
\end{examples*}

Recall that a boundary for a Levin-Wen model defined by $\C$ is specified by an algebra object $A\in\C$. The algebra encodes the strings that can terminate on the boundary. This interpretation leads us to the following lemma.

\begin{lemma}[Entropy of (union of) simply connected regions, with boundary]\label{lem:bndLW}
	On a region $R$ consisting of the disjoint union of simply connected sub-regions, the entropy is
	\begin{align}
		S_R & =nS[\C]+\frac{b_1}{2}\log d_A-b_0\log\D^2,\label{eqn:WWtrivialent}
	\end{align}
	where $b_0$ is the number of disjoint interface components of $R$, $b_1$ is the number of points where the entanglement surface intersects the physical boundary, and $n$ is the number of links crossing the entanglement interface.

	\begin{proof}
		Consider a ball $R$ with $n$ sites along the interface, which is in contact with the boundary. Recall that in the bulk, the fusion of the strings crossing the boundary was required to be $1$. In the presence of the boundary, this conservation rule is modified, since loops can terminate. All that is now required is that the fusion is in $A$
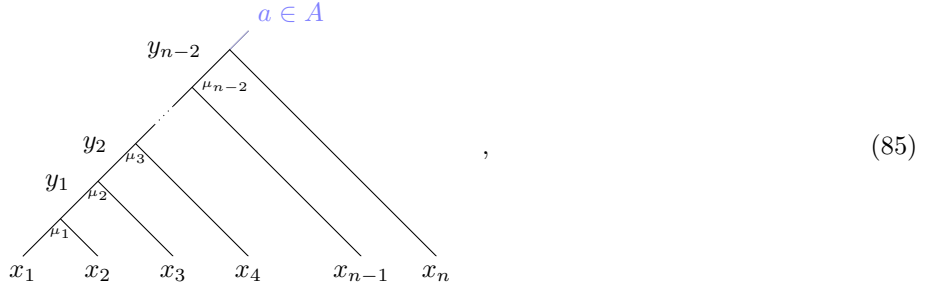
		\begin{align}
			\begin{array}{c}
				\includeTikz{tree_bnd}{
					\begin{tikzpicture}
						\draw(0,0)--(1.75,1.75);
						\draw[dotted](1.75,1.75)--(2,2);
						\draw(2,2)--(3,3);
						\draw[blue!20] (2.755,2.755)--(3,3);
						\begin{scope}
							\clip (0,0)--(3,3)--(6,3)--(6,0)--(0,0);
							\draw(1,0)--(0,1);
							\draw(2,0)--(0,2);
							\draw(3,0)--(0,3);
							\draw(4.5,0)--(0,4.5);
							\draw(5.5,0)--(0,5.5);
						\end{scope};
						\node[below] at (0,0) {$x_1$};
						\node[below] at (1,0) {$x_2$};
						\node[below] at (2,0) {$x_3$};
						\node[below] at (3,0) {$x_4$};
						\node[below] at (4.5,0) {$x_{n-1}$};
						\node[below] at (5.5,0) {$x_{n}$};
						\node[above right,blue!50] at (3,3) {$a\in A$};
						\node[above left] at (.75,.75) {$y_1$};
						\node[above left] at (1.25,1.25) {$y_2$};
						\node[above left] at (2.5,2.5) {$y_{n-2}$};
						\node[below] at (.5,.5) {\tiny{$\mu_1$}};
						\node[below] at (1,1) {\tiny{$\mu_2$}};
						\node[below] at (1.5,1.5) {\tiny{$\mu_3$}};
						\node[right] at (2.25,2.25) {\tiny{$\mu_{n-2}$}};
					\end{tikzpicture}
				}
			\end{array},\label{eqn:treeA_bnd}
		\end{align}

		The ground state can be decomposed as
		\begin{align}
			\ket{\psi} & =\sum_{\substack{\vec{x},\vec{y},\vec{\mu} \\
					a\in A}} \Phi_{\vec{x},\vec{y},\vec{\mu},a} \ket{ \psi_R^{\vec{x},\vec{y},\vec{\mu},a} }\ket{ \psi_{\comp{R}}^{\vec{x},\vec{y},\vec{\mu},a} }.\label{eqn:simpleregionpartition_bnd}
		\end{align}
		As before, the state $\ket{ \psi_R^{\vec{x},\vec{y},\vec{\mu},a} }$ includes any state that can be reached from \cref{eqn:treeA_bnd} by acting only on $R$. The reduced state on $R$ is
		\begin{align}
			\rho_R & =\sum_{\substack{\vec{x},\vec{y},\vec{\mu} \\a\in A}} |\Phi_{\vec{x},\vec{y},\vec{\mu},a}|^2 \ketbra{ \psi_R^{\vec{x},\vec{y},\vec{\mu},a} }\\
			       & =\sum_{\substack{\vec{x},\vec{y},\vec{\mu} \\a\in A}} \Pr[\vec{x},\vec{y},\vec{\mu}|a]\Pr[a] \ketbra{ \psi_R^{\vec{x},\vec{y},\vec{\mu},a} },
		\end{align}
		where $\Pr[\vec{x},\vec{y},\vec{\mu}|a]$ is the probability of the labeled tree, given that $\vec{x}$ fuses to $a$, and $\Pr[a\in A]=d_a/d_A$. Therefore,
		\begin{align}
			\rho_R & =\sum_{\substack{\vec{x},\vec{y},\vec{\mu} \\a\in A}} \frac{\prod_{j\leq n} d_{x_j}}{\D^{2(n-1)}d_A} \ketbra{ \psi_R^{\vec{x},\vec{y},\vec{\mu},a} }.
		\end{align}
		Applying \cref{lem:summingds,lem:sumlog} completes the proof for this region.
		It is straightforward to check that this holds on each sub-region of $R$, where \cref{lem:bulkLW} is used for any bulk sub-region.
	\end{proof}
\end{lemma}

Applying \cref{lem:bndLW} to the regions in \cref{fig:LWregionsbnd} completes the proof of \cref{thm:LWbnd}.

We can make sense of this halving of the entropy by considering folding the plane. Suppose we fold the model in \cref{fig:LWregionsblk} so that it resembles \cref{fig:LWregionsbnd}. This turns the bulk of a model defined by $\C$ to a boundary of a model labeled by $\C\boxtimes\C^{\rev}$. The quantum dimension of the folded theory is $\D_{\C\boxtimes\C^{\rev}}=\D_{\C}^2$, so the bulk diagnostic for $\C$ matches the boundary diagnostic computed for this folded theory.

\subsection{Walker-Wang models}

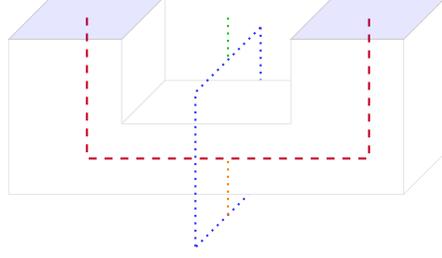
\begin{figure}
	\centering
	$\begin{array}{c}
			\includeTikz{bound_torus_A}{
				\begin{tikzpicture}[scale=.75]
					\draw[black!10] (0,0,0) -- (0,0,2) -- (2,0,2) -- (2,0,0) -- cycle;
					\draw[black!10] (2,0,2) -- (2,-1.5,2)-- (2,-1.5,0) -- (2,0,0);
					\draw[thick, dotted, blue!80] (3.5,-0.75,-0.5) -- (3.5,-3.5,-0.5) -- (3.5,-3.5,2.5);
					\draw[draw=black!10,fill=white] (2,-1.5,0) -- (5,-1.5,0) -- (5,-1.5,2) -- (2,-1.5,2) -- cycle;
					\draw[draw=black!10,fill=white] (0,0,2) -- (0,-2.75,2)-- (7,-2.75,2) -- (7,0,2) -- (5,0,2) -- (5,-1.5,2) -- (2,-1.5,2) -- (2,0,2) -- cycle;
					\draw[black!10] (7,-2.75,2) -- (7,-2.75,0) -- (7,0,0);
					\draw[black!10] (5,0,0) -- (5,0,2) -- (7,0,2) -- (7,0,0) -- cycle;
					\draw[thick, dashed, darkred] (1,0,1) -- (1,-2.5,1)-- (6,-2.5,1)-- (6,0,1);
					\draw[fill=blue,opacity=.1] (5,0,0) -- (5,0,2) -- (7,0,2) -- (7,0,0) -- cycle;
					\draw[fill=blue, opacity = 0.1] (0,0,0) -- (0,0,2) -- (2,0,2) -- (2,0,0) -- cycle;
					\draw[thick, dotted, blue!80] (3.5,-0.75,-0.5) -- (3.5,-0.75,2.5) -- (3.5,-3.5,2.5);
					\draw[thick,dotted,orange] (3.5,-3.5,1)--(3.5,-2.5,1);
					\draw[thick,dotted,nicegreen] (3.5,0,1)--(3.5,-0.75,1);
				\end{tikzpicture}
			}\end{array}$
	\caption{When the region $R$ is in non-simple contact with a boundary on which strings can terminate, the computation of entropy is more subtle. There is additional entanglement in the system due to intersecting loops that cannot be created in $R$ (red, dashed) or $\comp{R}$ (blue, dotted) separately. 	The red (internal) strings can terminate on the boundary. Also, the blue loop can emit a string which can terminate on the boundary.
	}
	\label{fig:bound_torus_A}
\end{figure}

In (3+1)-dimensions, just like in (2+1)-dimensions, strings can terminate at the boundary. In addition, loops can interlock as discussed in \cref{ss:WWbulkpf}. In the vicinity of the boundary, these two effects can occur simultaneously as depicted in \cref{fig:bound_torus_A}.

For simply connected regions in contact with a boundary, we can apply \cref{lem:bndLW}, replacing $b_1/2$ with the number of lines where the region contacts the physical boundary. By applying the results so far, it is straightforward to check that the two diagnostics \cref{eqn:WWptdef} and \cref{eqn:WWloopdef} are related by
\begin{align}
	\WWbndTEELoop & =\WWbndTEEPoint+\log d_A^2-\log\D^2,
\end{align}
so we only need to consider $\WWbndTEEPoint$. We are currently unable to compute this in general, however in this section we prove the following results:
\begin{thm}\label{thm:wwbnd}
	For a Walker-Wang model defined by a unitary premodular category $\C$, the entropy diagnostic $\WWbndTEEPoint$ for a boundary labeled by an indecomposable, strongly separable, commutative, Frobenius algebra $A$ is given by
	\begin{align}\hspace*{-2mm}
		\WWbndTEEPoint & =
		\begin{cases}
			\log\D^2           & A=1,                                                                               \\
			\log\D^2-\log d_A  & \C\text{ symmetric},                                                               \\
			\log\D^2-2\log d_A & \C\text{ pointed and }\mug{C}\cap A=\{1\}.\text{ In particular }\C\text{ modular.} \\
			\log\D^2-\log d_A  & \C\text{ pointed and }\mug{C}\cap A=A.
		\end{cases}
	\end{align}
\end{thm}

\begin{examples*}
	Recall the examples from \cref{sec:examples}. As discussed in \cref{sec:examples_phys}, these label four distinct loop-gas models in (3+1)-dimensions, the bosonic and fermionic toric code models, and two semion models. All four input categories are pointed, so we can apply \cref{thm:wwbnd} to obtain the boundary topological entanglement entropy.

	The bosonic toric code is compatible with two distinct gapped boundary conditions, labeled by $A_0$ and $A_1$ (see \cref{sec:examples}), with $d_{A_0} = d_1 = 1$, and $d_{A_1} = d_1+d_x=2$. Since the input category is symmetric, $\mug{\C}\cap A_i=A_i$, so the entropy diagnostics are $\WWbndTEEPoint(A_0) = \log 2 -  \log 1=\log2$ and $\WWbndTEEPoint(A_1) = \log 2 -  \log 2=0$.

	For the remaining examples, only the boundary labeled by $A_0$ is compatible, and $\WWbndTEEPoint = \log2$ in all cases.
\end{examples*}

\begin{proof}
	To capture configurations like the one in \cref{fig:bound_torus_A}, we need new boundary $\S$-matrices resembling
	\begin{align}
		\frac{1}{\D}
		\begin{array}{c}
			\includeTikz{linkedSmatrixalgebra_1}{
				\begin{tikzpicture}[scale=.75]
					\centerarc[draw=white,double=black,ultra thick](-.75,0)(0:180:1);
					\centerarc[draw=white,double=black,ultra thick](.75,0)(180-60:360+45:1);% Syntax: [draw options] (center) (initial angle:final angle:radius)
					\centerarc[draw=white,double=black,ultra thick](-.75,0)(180:360:1);% Syntax: [draw options] (center) (initial angle:final angle:radius)
					\node[anchor=west,inner sep=.5] at(1.75,0) {\strut$a_0$};
					\node[anchor=west,inner sep=.5] at(.25,0) {\strut$b_0$};
					\node[anchor=east,inner sep=.5] at(-.25,0) {\strut$a_1$};
					\node[anchor=east,inner sep=.5] at(-1.75,0) {\strut$b_1$};
					\draw[thick] (-.75,-1)to[out=270,in=270](.75,-1);
					\node[anchor=east,inner sep=.5] at(-.75,1.15) {\strut$d$};
					\node[anchor=east,inner sep=.5] at(-.75,-1.25) {\strut$c$};
					\node[anchor=west,inner sep=.5] at(.75,-1.25) {\strut$\dual{c}$};
					\draw[thick] (-.75,1)--(-.75,1.5);
					\fill[blue!20] (-.75,1.5) circle (.1);
					\fill[blue!20] ($(.75,0)+({cos(180-60)},{sin(180-60)})$) circle (.1);
					\fill[blue!20] ($(.75,0)+({cos(360+45)},{sin(360+45)})$) circle (.1);
				\end{tikzpicture}
			}
		\end{array},\label{eqn:bndSmatrix1}
	\end{align}
	where the dots indicate where a string meets the boundary. We use boundary retriangulation invariance, as defined in \onlinecites{1706.00650,1706.03329} to evaluate this diagram on the ground space. Using this, we define the new $\S$-matrix elements as
	\begin{align}
		\left[\S_{c,d}\right]_{(b_0,b_1),(a_0,a_1)}:= &
		\frac{m_{a_1,a_0}^{\dual{d}}}{(d_{a_0}d_{a_1}d_{d})^{1/4}d_A\D}
		\begin{array}{c}
			\includeTikz{linkedSmatrixalgebra_2}{
				\begin{tikzpicture}[scale=.75]
					\centerarc[draw=white,double=black,ultra thick](-.75,0)(0:180:1);
					\centerarc[draw=white,double=black,ultra thick](.75,0)(90:360+90:1);% Syntax: [draw options] (center) (initial angle:final angle:radius)
					\centerarc[draw=white,double=black,ultra thick](-.75,0)(180:360:1);% Syntax: [draw options] (center) (initial angle:final angle:radius)
					\node[anchor=west,inner sep=.5] at(1.75,0) {\strut$a_0$};
					\node[anchor=west,inner sep=.5] at(.25,0) {\strut$b_0$};
					\node[anchor=east,inner sep=.5] at(-.25,0) {\strut$a_1$};
					\node[anchor=east,inner sep=.5] at(-1.75,0) {\strut$b_1$};
					\draw[thick] (-.75,-1)to[out=270,in=270](.75,-1);
					\draw[thick] (-.75,1)to[out=90,in=90](.75,1);
					\node[anchor=east,inner sep=.5] at(-.75,1.25) {\strut$d$};
					\node[anchor=west,inner sep=.5] at(.75,1.25) {\strut$\dual{d}$};
					\node[anchor=east,inner sep=.5] at(-.75,-1.25) {\strut$c$};
					\node[anchor=west,inner sep=.5] at(.75,-1.25) {\strut$\dual{c}$};
				\end{tikzpicture}
			}
		\end{array},\label{eqn:bndSmatrix2}
	\end{align}
	where $a_0,a_1,d\in A$ and $b_0,b_1,c\in \C$.
	With this, the ground state can be written
	\begin{align}
		\ket{\psi}=\sum_{
		\substack{\vec{x},\vec{y},\vec{\mu}                                                                                                                                                                                                 \\
		c,b_0,b_1\in\C                                                                                                                                                                                                                      \\
		d,a_0,a_1\in A}} & \Phi_{\vec{x},\vec{y},\vec{\mu},c} \frac{\left[\S_{c,d}\right]_{(b_0,b_1),(a_0,a_1)}}{N_A}  \ket{ \psi_R^{\vec{x},\vec{y},\vec{\mu},c,a_0,a_1} }\ket{ \psi_{\comp{R}}^{\vec{x},\vec{y},\vec{\mu},c,b_0,b_1,d} },
	\end{align}
	where $N_A$ is a normalizing factor. The reduced state on $R$ is
	\begin{align}
		\rho_R=\sum_{
		\substack{\vec{x},\vec{y},\vec{\mu}                                                                                                                                                                                                         \\
		c\in\C                                                                                                                                                                                                                                      \\
		d,a_0,a_1,a_2,a_3\in A}} & \frac{\prod_{j\leq n} d_{x_j}\left[\S_{c,d}^\dagger\S_{c,d}\right]_{(a_2,a_3),(a_0,a_1)}}{d_c\D^{2(n-1)}N_A}\ketbra{\psi_R^{\vec{x},\vec{y},\vec{\mu},c,a_0,a_1}}{\psi_R^{\vec{x},\vec{y},\vec{\mu},c,a_2,a_3}}.
	\end{align}

	\subsubsection{\texorpdfstring{$A=1$}{A=1}}
	When the algebra is trivial, no strings can terminate. In that case, $\S_{1,1}^\dagger\S_{1,1}=1$, so the reduced state is
	\begin{align}
		\rho_R & =\sum_{
			\vec{x},\vec{y},\vec{\mu}} \frac{\prod_{j\leq n} d_{x_j}}{\D^{2(n-1)}} \ketbra{\psi_R^{\vec{x},\vec{y},\vec{\mu}}}{\psi_R^{\vec{x},\vec{y},\vec{\mu}}},
	\end{align}
	which is diagonal and has entropy
	\begin{align}
		S_R & =n S[C]-\log\D^2.
	\end{align}

	\subsubsection{\texorpdfstring{$\C$}{C} symmetric}
	When $\C$ is symmetric, the rings in \cref{eqn:bndSmatrix2} separate, so
	\begin{align}
		\left[\S_{c,d}\right]_{(b_0,b_1),(a_0,a_1)}=                 &
		\indicator{c=d}\sqrt{d_{b_0}d_{b_1}}
		\frac{(d_{a_0}d_{a_1})^{1/4}m_{a_1,a_0}^{\dual{d}}}{d_{d}^{1/4}d_A^2\D},                                                                                                                                 \\
		\left[\S_{c,d}^\dagger\S_{c,d}\right]_{(a_2,a_3),(a_0,a_1)}= &
		\indicator{c=d}\sum_{b_0,b_1}N_{b_0b_1}^{d}\frac{d_{b_0}d_{b_1}}{\sqrt{d_d}}
		(d_{a_0}d_{a_1}d_{a_2}d_{a_3})^{1/4}\frac{m_{a_1,a_0}^{\dual{d}}\left(m_{a_3,a_2}^{\dual{d}}\right)^*}{d_A^4\D^2}                                                                                        \\
		=                                                            & \indicator{c=d} (d_{a_0}d_{a_1}d_{a_2}d_{a_3})^{1/4}\sqrt{d_d} \frac{m_{a_1,a_0}^{\dual{d}}\left(m_{a_3,a_2}^{\dual{d}}\right)^*}{d_A^2}.
	\end{align}
	It can readily be verified that this matrix is rank 1. The eigenvalue can be found using \cref{eqn:algnorm}, giving $\lambda=d_d/d_A$. We can therefore write the state on $R$ as
	\begin{align}
		\rho_R & =\sum_{\substack{\vec{x},\vec{y},\vec{\mu}                                                                                \\c\in\C\\d\in A}}
		\indicator{c=d}
		\frac{\prod_{j\leq n} d_{x_j}}{\D^{2(n-1)}d_A} \ketbra{\psi_R^{\vec{x},\vec{y},\vec{\mu},c}}{\psi_R^{\vec{x},\vec{y},\vec{\mu},c}} \\
		       & =\sum_{
		\substack{\vec{x},\vec{y},\vec{\mu}                                                                                                \\c\in A}}
		\frac{\prod_{j\leq n} d_{x_j}}{\D^{2(n-1)}d_A} \ketbra{\psi_R^{\vec{x},\vec{y},\vec{\mu},c}}{\psi_R^{\vec{x},\vec{y},\vec{\mu},c}}.
	\end{align}
	Using \cref{lem:summingds,lem:sumlog}, we find that the entropy of this state is
	\begin{align}
		\S_R & =nS[\C]-\log\D^2+\log d_A.
	\end{align}

	\subsubsection{\texorpdfstring{$\C$}{C} pointed}
	When all quantum dimensions are equal to 1, the boundary $\S$-matrix is
	\begin{align}
		\left[\S_{c,d}\right]_{(b_0,b_1),(a_0,a_1)}= & \indicator{c=d}
		\frac{m_{a_1,a_0}^{\dual{d}}}{d_A}\left[\S\right]_{b_0,\dual{a}_1},
	\end{align}
	where $\left[\S\right]_{b_0,a_1}$ is the $\S$-matrix from \cref{eqn:Smatrix}. This gives
	\begin{align}
		\left[\S_{c,d}^\dagger\S_{c,d}\right]_{(a_2,a_3),(a_0,a_1)} & =\indicator{c=d}
		\frac{m_{a_1,a_0}^{\dual{d}}(m_{a_3,a_2}^{\dual{d}})^*}{d_A^2}\sum_{b_0}
		\left[\S\right]_{b_0,\dual{a}_3}^*\left[\S\right]_{b_0,\dual{a}_1}.
	\end{align}
	Using \cref{lem:productS,lem:sumS}, this can be simplified to
	\begin{align}
		\left[\S_{c,d}^\dagger\S_{c,d}\right]_{(a_2,a_3),(a_0,a_1)} & =\indicator{c=d}
		\frac{m_{a_1,a_0}^{\dual{d}}(m_{a_3,a_2}^{\dual{d}})^*}{d_A^2}\indicator{a_3\otimes\dual{a}_1\in\mug{\C}}.
	\end{align}

	Pointed braided categories have fusion rules given by an Abelian group $G$~\cite{Joyal_1993,MR3242743}, and algebras are twisted group algebras~\cite{0202130,bullivant2020} of subgroups of $G$. Moreover, $\mug{\C}$ also has fusion rules given by a subgroup $G$.

	Since $a_3\otimes\dual{a}_1\in A$ and $a_3\otimes\dual{a}_1\in \mug{\C}$, there must be some $h\in\mug{\C}\cap A$ so that $a_3 = a_1 h$. We can then write
	\begin{align}
		\rho_R=\sum_{
			\substack{
		\vec{x}   \\
		d,a\in A  \\
				h\in \mug{\C}\cap A
			}}
		 & \frac{
			m_{a,(ad)^{-1}}^{d^{-1}}(m_{a h,(ahd)^{-1}}^{d^{-1}})^*}{
			d_A^2\D^{2(n-1)}N_A}
		\ketbra{\psi_R^{\vec{x},d,a}}{\psi_R^{\vec{x},d,ah}}.\label{eqn:rho_ptd}
	\end{align}

	In the case that $\mug{\C}\cap A=A$, this reduces to the symmetric case, since summing over $h\in A$ is the same as summing over $A\ni h^\prime = ah$.

	When $\mug{\C}\cap A=\{1\}$, the unit, the reduced state on $R$ simplifies to
	\begin{align}
		\rho_R=\sum_{
			\substack{
		\vec{x}   \\
				d,a\in A
			}}
		 & \frac{
			m_{a,(ad)^{-1}}^{d^{-1}}(m_{a,(ad)^{-1}}^{d^{-1}})^*}{
			d_A^2\D^{2(n-1)}N_A}
		\ketbra{\psi_R^{\vec{x},d,a}}{\psi_R^{\vec{x},d,a}}.
	\end{align}
	This reduced state is diagonal, and has entropy
	\begin{align}
		S_R & =-\sum_{
		\substack{\vec{x}                                           \\
				d,a\in A}}
		\frac{|m_{a,(ad)^{-1}}^{d^{-1}}|^2}{\D^{2(n-1)}d_A^2}
		\log(\frac{|m_{a,(ad)^{-1}}^{d^{-1}}|^2}{\D^{2(n-1)}d_A^2}) \\
		    & =-\sum_{
		\substack{\vec{x}                                           \\
				d,a\in A}}
		\frac{|m_{a,(ad)^{-1}}^{d^{-1}}|^2}{\D^{2(n-1)}d_A^2}
		\log(|m_{a,(ad)^{-1}}^{d^{-1}}|^2)+\log(\D^{2(n-1)}d_A^2),
	\end{align}
	where we have made use of \cref{eqn:algnorm}. Since $A$ is a twisted group algebra, we may assume $|m_{ab}^c|\in\{0,1\}$. Finally, this gives
	\begin{align}
		S_R & =\log(\D^{2(n-1)}d_A^2)    \\
		    & =S[\C]-\log\D^2+2\log d_A,
	\end{align}
	completing the proof.
\end{proof}
% !TeX spellcheck = en_US
% This line sets the project root file.
% !TEX root = ../WW_EntanglementEntropy.tex
%

\section{Remarks}\label{sec:remarks}

To summarize, we have evaluated the long-range entanglement in the bulk, and at the boundary, of (2+1)- and (3+1)-dimensional topological phases.
In (2+1)-dimensions, we found the entropy diagnostic $\LWbndTEE=\log\D^2$ regardless of the choice of boundary algebra $A$. This is in contrast to the results for (3+1)-dimensions, where a signature of the boundary, namely its dimension as an algebra, can be seen in the diagnostics $\WWbndTEEPoint$ and $\WWbndTEELoop$.

The most natural boundary for these models is defined by the algebra $A=1$, which (uniquely) always exists. At this boundary, we found that the point-like diagnostic $\WWbndTEEPoint$ recovers the total dimension of the input category.
In particular, when $\C$ is a (2+1)-dimensional anyon model, this is consent with a boundary that supports the anyons. Conversely, the loop-like diagnostic $\WWbndTEELoop$ is zero at these boundaries, ruling out loop-like excitations in the vicinity.

We have conjectured a general property of the connected $\S$-matrix which, if proven in general, allows computation of bulk Walker-Wang topological entropy. Such a proof may also be interesting for the classification of premodular categories in general.
To the best of our knowledge, there is no complete classification of boundaries for Walker-Wang models. Such a classification is complicated by requiring, as a sub-classification, a complete understanding of (2+1)-dimensional theories. This goes beyond the scope of the current work, and we have therefore specialized to boundaries described by Frobenius algebras and to particular families of input fusion category. Extending these results may provide a more complete understanding of the possible boundary excitations and their properties.

\acknowledgments
We particularly thank David Aasen for useful discussions, and informing us of the class of Walker-Wang boundaries we consider.
We are also grateful to Taiga Adair, Daniel Barter, Isaac Kim, Sam Roberts, Thomas Smith and Dominic Williamson for helpful discussions. We also thank Stephen Bartlett and C\'{e}line B{\oe}hm.

Research at Perimeter Institute is supported in part by the Government of Canada through the Department of Innovation, Science and Industry Canada and by the Province of Ontario through the Ministry of Colleges and Universities.
This work is supported by the Australian Research Council (ARC) via the Centre of Excellence in Engineered Quantum Systems (EQuS) project number CE170100009. BJB was previously supported by the University of Sydney Fellowship Programme. SJE also received support from the EPSRC.

\bibliographystyle{apsrev_jacob}
\bibliography{refs.bib}
\appendix
\newgeometry{left=17mm,right=17mm,top=17mm,bottom=22mm,ignoreall, noheadfoot}
\clearpage
% !TeX spellcheck = en_US
% This line sets the project root file.
% !TEX root = ../WW_EntanglementEntropy.tex
%

\section{Properties of the connected \texorpdfstring{$\S$}{S}-matrix}\label{app:FC_pfs}

\begin{lemma_rep}[\ref{lem:productS}]\label{lem:productS_pf}
	Let $\C$ be a unitary premodular category. The matrix elements of $\S$ obey
	\begin{align}
		\frac{\D}{d_c}\S_{a,c}\S_{b,c} & =\S_{a\otimes b,c}=\sum_x N_{a,b}^x \S_{x,c}.
	\end{align}
	\begin{proof}
		We prove the second equality first.
		\begin{align}
			\S_{a\otimes b,c} & :=\frac{1}{\D}\begin{array}{c}
				\includeTikz{productS_1}{
					\begin{tikzpicture}[scale=.75]
						\centerarc[draw=white,double=black,ultra thick](-.75,0)(0:180:1);
						\centerarc[draw=white,double=black,ultra thick](-.75,0)(0:180:.8);
						\draw[draw=white,double=black,ultra thick] (.75,0) circle (1);
						\centerarc[draw=white,double=black,ultra thick](-.75,0)(180:360:1);
						\centerarc[draw=white,double=black,ultra thick](-.75,0)(180:360:.8);
						\node[anchor=west,inner sep=.5] at(1.75,0) {\strut$c$};
						\node[anchor=west,inner sep=.5] at(.25,0) {\strut$b$};
						\node[anchor=east,inner sep=.5] at(-1.75,0) {\strut$\dual{b}$};
						\node[anchor=west,inner sep=1] at(-1.55,0) {\strut$\dual{a}$};
					\end{tikzpicture}
				}
			\end{array}
			=\frac{1}{\D}\sum_{x}\sum_{\mu=1}^{N_{ab}^x}\sqrt{\frac{d_x}{d_ad_b}}
			\begin{array}{c}
				\includeTikz{productS_2}{
					\begin{tikzpicture}[scale=.75]
						\centerarc[draw=white,double=black,ultra thick](-.75,0)(30:180:1);
						\centerarc[draw=white,double=black,ultra thick](-.75,0)(30:180:.8);
						\draw[draw=white,double=black,ultra thick] (.75,0) circle (1);
						\centerarc[draw=white,double=black,ultra thick](-.75,0)(180:330:1);
						\centerarc[draw=white,double=black,ultra thick](-.75,0)(180:330:.8);
						\draw[semithick] ($(30:1)-(.75,0)$) to[out = 300,in=70] (.15,.15);
						\draw[semithick] ($(30:.8)-(.75,0)$) to[out = 300,in=110] (.15,.15);
						\draw[semithick] ($(330:1)-(.75,0)$) to[out = 60,in=300] (.15,-.15);
						\draw[semithick] ($(330:.8)-(.75,0)$) to[out = 60,in=240] (.15,-.15) -- (.15,.15);
						\node[anchor=west,inner sep=.5] at(1.75,0) {\strut$c$};
						\node[anchor=east,inner sep=.5] at(-1.75,0) {\strut$\dual{b}$};
						\node[anchor=west,inner sep=1] at(-1.55,0) {\strut$\dual{a}$};
						\node[anchor=west,inner sep=.5] at(.15,0) {\strut$x$};
						\node[anchor=south west,inner sep=.5] at(.15,.15) {\strut\scriptsize$\mu$};
						\node[anchor=north west,inner sep=.5] at(.15,-.15) {\strut\scriptsize$\mu$};
					\end{tikzpicture}
				}
			\end{array}                                             \\
			                  & =\frac{1}{\D}\sum_{x,\mu}\sqrt{\frac{d_x}{d_ad_b}}
			\begin{array}{c}
				\includeTikz{productS_3}{
					\begin{tikzpicture}[scale=.75]
						\centerarc[draw=white,double=black,ultra thick](-.75,0)(0:150:1);
						\draw[draw=white,double=black,ultra thick] (.75,0) circle (1);
						\centerarc[draw=white,double=black,ultra thick](-.75,0)(210:360:1);
						\draw[semithick] ($(150:1)-(.75,0)$) to[out = 240,in=120] ($(210:1)-(.75,0)$);
						\draw[semithick] ($(150:1)-(.75,0)$) to[out = 300,in=60] ($(210:1)-(.75,0)$);
						\node[anchor=west,inner sep=.5] at(1.75,0) {\strut$c$};
						\node[anchor=west,inner sep=.5] at(.25,0) {\strut$x$};
						\node[anchor=east,inner sep=.5] at(-1.75,0) {\strut$\dual{b}$};
						\node[anchor=west,inner sep=1] at(-1.5,0) {\strut$\dual{a}$};
					\end{tikzpicture}
				}
			\end{array}=\sum_{x}N_{a,b}^{x} \S_{x,c}.
		\end{align}
		And the first equality.
		\begin{align}
			\frac{1}{\D}
			\begin{array}{c}
				\includeTikz{productS_1}{}
			\end{array}
			 & =\frac{1}{\D}\sum_{x,\mu}\sqrt{\frac{d_x}{d_ad_c}}\frac{\theta_x}{\theta_a\theta_{\dual{c}}}
			\begin{array}{c}
				\includeTikz{productS_3a}{
					\begin{tikzpicture}[scale=.75]
						\centerarc[draw=white,double=black,ultra thick](-.75,0)(0:180:1);
						\draw[draw=white,double=black,ultra thick] (.75,0) circle (1);
						\centerarc[draw=white,double=black,ultra thick](-.75,0)(180:360:1);
						\draw[shift={(.75,0)}] (160:1)to[out=100,in=0] (-1.25,.75)to[out=180,in=180](-1.25,-.75)to[out=0,in=260](200:1);
						\node[anchor=west,inner sep=.5] at(1.75,0) {\strut$c$};
						\node[anchor=west,inner sep=.5] at(.25,0) {\strut$b$};
						\node[anchor=west,inner sep=.5] at(-.25,0) {\strut$x$};
						\node[anchor=east,inner sep=.5] at(-1,0) {\strut$\dual{a}$};
					\end{tikzpicture}
				}
			\end{array}
			=\frac{1}{\D}\sum_{x,\mu}\sqrt{\frac{d_x}{d_ad_c}}\frac{\theta_x}{\theta_a\theta_{\dual{c}}}
			\begin{array}{c}
				\includeTikz{productS_3b}{
					\begin{tikzpicture}[scale=.75]
						\centerarc[draw=white,double=black,ultra thick](-.75,0)(0:180:1);
						\draw[draw=white,double=black,ultra thick] (.75,0) circle (1);
						\centerarc[draw=white,double=black,ultra thick](-.75,0)(180:360:1);
						\draw[shift={(.75,0)}] (160:1)to[out=240,in=90](-1.25,0)to[out=270,in=120](200:1);
						\node[anchor=west,inner sep=.5] at(1.75,0) {\strut$c$};
						\node[anchor=west,inner sep=.5] at(.25,0) {\strut$b$};
						\node[anchor=west,inner sep=.5] at(-.25,0) {\strut$x$};
						\node[anchor=east,inner sep=.5] at(-.5,0) {\strut$\dual{a}$};
					\end{tikzpicture}
				}
			\end{array}
			\\%
			 & =\left(\frac{1}{\D d_c}\sum_{x}\frac{\theta_x}{\theta_a\theta_{\dual{c}}}d_x\right)
			\begin{array}{c}
				\includeTikz{productS_4}{
					\begin{tikzpicture}[scale=.75]
						\centerarc[draw=white,double=black,ultra thick](-.75,0)(0:180:1);
						\draw[draw=white,double=black,ultra thick] (.75,0) circle (1);
						\centerarc[draw=white,double=black,ultra thick](-.75,0)(180:360:1);
						\node[anchor=west,inner sep=.5] at(1.75,0) {\strut$c$};
						\node[anchor=west,inner sep=.5] at(.25,0) {\strut$b$};
					\end{tikzpicture}
				}
			\end{array}
			=\frac{\S_{a,c}}{d_c}
			\begin{array}{c}
				\includeTikz{productS_4}{}
			\end{array}=\frac{\D\S_{a,c}\S_{b,c}}{d_c}.
		\end{align}
	\end{proof}
\end{lemma_rep}

\begin{lemma_rep}[\ref{lem:sumS}]\label{lem:sumS_pf}
	Let $\C$ be a unitary premodular category, then
	\begin{align}
		\sum_b d_b \S_{a,b} & =\indicator{a\in\mug{\C}} d_a \D,
	\end{align}
	where $\indicator{W}=1 \iff W$ is true, and $\mug{\C}$ is the M\"uger center.
	\begin{proof}
		If $a\in \mug{\C}$, then
		\begin{align}
			\sum_b d_b\S_{a,b} & =\frac{1}{\D}\sum_{x,b}N_{a,\dual{b}}^x d_x d_b \\
			                   & =\frac{1}{\D}\sum_{b}d_a d_b^2                  \\
			                   & =\D d_a
		\end{align}
		Otherwise, we have
		\begin{align}
			\frac{\D}{d_a}\S_{a,c}\sum_b d_b \S_{a,b} & =\sum_{x,b}N_{c,b}^x \S_{a,x}d_b=\sum_x \S_{a,x}d_xd_c,
		\end{align}
		so after relabeling the RHS summation variable
		\begin{align}
			\sum_b d_b \S_{a,b}\left(\S_{a,c}-\frac{d_a d_c}{\D}\right) & =0,
		\end{align}
		for all $c$. Unless $\C=\mug{\C}$ (in which case, $a\in \mug{\C}$), this implies the result.
	\end{proof}
\end{lemma_rep}

\begin{lemma_rep}[\ref{lem:TrScSc}]\label{lem:TrScSc_pf}
	Let $\C$ be a unitary premodular category, then
	\begin{align}
		\sum_{c\in\C}\Tr\S_c^\dagger\S_c & =\D^2,
	\end{align}
	where $\S_c$ is the connected $\S$-matrix and $\D$ is the total dimension of $\C$.
	\begin{proof}
		Using \cref{eqn:Scaction}, we have
		\begin{align}
			\sum_c\Tr\S_c^\dagger\S_c & =\sum_{a,\alpha,c}\left[\S_c^\dagger\S_c\right]_{(a,\alpha),(a,\alpha)} \\
			                          & =\sum_{a,\alpha,c}\frac{\sqrt{d_c}}{\D^2}\sum_x d_x
			\begin{array}{c}
				\includeTikz{STrProof1}{
					\begin{tikzpicture}[scale=.75]
						\centerarc[draw=white,double=black,ultra thick](0,0)(0:180:1);
						\draw[draw=white,double=black,ultra thick] (1.5,0) circle (1);
						\draw[draw=white,double=black,ultra thick] (-1.5,0) circle (1);
						\centerarc[draw=white,double=black,ultra thick](0,0)(180:360:1);% Syntax: [draw options] (center) (initial angle:final angle:radius)
						\node[anchor=west,inner sep=.5] at(2.5,0) {\strut$a$};
						\node[anchor=east,inner sep=.5] at(-2.5,0) {\strut$\dual{a}$};
						\node[anchor=west,inner sep=.5] at(1,0) {\strut$x$};
						\draw[thick] (1.5,-1)--(1.5,-1.2);\draw[thick] (-1.5,1)--(-1.5,1.2);
						\draw[thick] (1.5,-1.2)to[out=270,in=270] (3,0)to[out=90, in=90] (-1.5,1.2);
						\node[anchor=south,inner sep=.5] at(1.5,-1) {\strut$\alpha$};
						\node[anchor=north,inner sep=.75] at(-1.5,1) {\strut$\alpha$};
						\node[anchor=east,inner sep=.75] at(-1.5,1.25) {\strut$c$};
						\node[anchor=west,inner sep=.75] at(3,0) {\strut$\dual{c}$};
					\end{tikzpicture}
				}
			\end{array}                                                                          \\
			                          & =\sum_{a,\alpha,c}\frac{\sqrt{d_c}}{\D^2}\sum_x d_x
			\begin{array}{c}
				\includeTikz{STrProof2}{
					\begin{tikzpicture}[scale=.75]
						\coordinate (X) at  ($(-1.5,0)+({cos(130)},{sin(130)})$);
						\coordinate (Y) at  ($(-1.5,0)+({cos(410)},{sin(410)})$);
						\centerarc[draw=white,double=black,ultra thick](0,0)(0:180:1);
						\draw[draw=white,double=black,ultra thick] (1.5,0) circle (1);
						\centerarc[draw=white,double=black,ultra thick](-1.5,0)(130:410:1);
						\centerarc[draw=white,double=black,ultra thick](0,0)(180:360:1);% Syntax: [draw options] (center) (initial angle:final angle:radius)
						\node[anchor=west,inner sep=.5] at(2.5,0) {\strut$a$};
						\node[anchor=east,inner sep=.5] at(-2.5,0) {\strut$\dual{a}$};
						\node[anchor=west,inner sep=.5] at(1,0) {\strut$x$};
						\draw[thick] (1.5,-1)--(1.5,-1.5);
						\draw[draw=white,double=black,ultra thick] (X)to[out=40,in=180] (0,2)to[out=0,in=90](3.25,0)to[out=-90,in=90](3,-1.5)to[out=-90,in=220](1.25,-1.75)--(1.5,-1.5);
						\draw[draw=white,double=black,ultra thick] (Y)to[out=130,in=180] (0,1.5)to[out=0,in=90](3,0)to[out=-90,in=305](1.75,-1.75)--(1.5,-1.5);
						\draw[thick] (1.5,-1.5)--(1.25,-1.75)  (1.5,-1.5)--(1.75,-1.75) ;
						\node[anchor=south,inner sep=.5] at(1.5,-1) {\strut$\alpha$};
						\node[anchor=north,inner sep=.75] at(1.5,-1.6) {\strut$\alpha$};
						\node[anchor=east,inner sep=.75] at(1.5,-1.25) {\strut$c$};
					\end{tikzpicture}
				}
			\end{array},
		\end{align}
		using the properties of the trace. Applying the premodular trap (\cref{cor:premodulartrap}), this gives
		\begin{align}
			\sum_c\Tr\S_c^\dagger\S_c & =\sum_{\substack{a,\alpha,c,                                             \\z\in\mug{\C},\mu}}
			\sqrt{\frac{d_cd_z}{d_a^2}}
			\begin{array}{c}
				\includeTikz{STrProof3}{
					\begin{tikzpicture}[scale=.75]
						\draw (-.5,-.75)--(0,-.5)--(0,.5)--(-.5,.75) (0,.5)--(.5,.75) (0,-.5)--(.5,-.75);
						\draw[shift={(1,-1.5)}] (-.5,-.75)--(0,-.5)--(0,.5)--(-.5,.75) (0,.5)--(.5,.75) (0,-.5)--(.5,-.75);
						\draw (.5,.75) to[out=30,in=30] (1.5,-.75);
						\draw (-.5,.75) to[out=120,in=120] (1.5,1) to[out=300,in=-30] (1.5,-2.25);
						\draw (-.5,-.75) to[out=210,in=270] (-1,-.5)to[out=90,in=90](-1.5,-.5)to[out=270,in=210](.5,-2.25);
						\node[anchor=east,inner sep=.75] at(0,0) {$z$};
						\node[anchor=north,inner sep=.75] at(0,-.5) {\strut\footnotesize$\mu$};
						\node[anchor=south,inner sep=.75] at(0,.5) {\strut\footnotesize$\mu$};
						\begin{scope}[shift={(1,-1.5)}]
							\node[anchor=east,inner sep=.75] at(0,0) {\strut$c$};
							\node[anchor=north,inner sep=.75] at(0,-.5) {\strut\footnotesize$\alpha$};
							\node[anchor=south,inner sep=.75] at(0,.5) {\strut\footnotesize$\alpha$};
						\end{scope}
						\node[anchor=south,inner sep=.75] at(.5,-.75) {\strut$\dual{a}$};
						\node[anchor=south,inner sep=.75] at(.5,-2.25) {\strut$\dual{a}$};
						\node[anchor=west,inner sep=.75] at(1.5,-.75) {\strut$a$};
						\node[anchor=south west,inner sep=.75] at(1.5,-2.25) {\strut$a$};
					\end{tikzpicture}
				}
			\end{array}                                                                           \\
			                          & =\sum_{\substack{a,z\in\mug{\C},\mu}}\sqrt{d_z}
			\begin{array}{c}
				\includeTikz{STrProof4}{
					\begin{tikzpicture}[scale=.75]
						\draw (-.5,-.75)--(0,-.5)--(0,.5)--(-.5,.75) (0,.5)--(.5,.75) (0,-.5)--(.5,-.75);
						\draw (-.5,-.75) to[out=210,in=270] (-1,-.5)to[out=90,in=90](-1.5,-.5)to[out=270,in=-45](.5,-.75);
						\draw[scale=-1] (-.5,-.75) to[out=210,in=270] (-1,-.5)to[out=90,in=90](-1.5,-.5)to[out=270,in=-45](.5,-.75);
						\node[anchor=east,inner sep=.75] at(0,0) {\strut$z$};
						\node[anchor=north,inner sep=.75] at(0,-.5) {\strut\footnotesize$\mu$};
						\node[anchor=south,inner sep=.75] at(0,.5) {\strut\footnotesize$\mu$};
					\end{tikzpicture}
				}
			\end{array}                                                                           \\
			                          & =\sum_{\substack{a,z\in\mug{\C},\mu}}|\varkappa_a|^2\indicator{z=1}d_a^2 \\
			                          & =\D^2,
		\end{align}
		where $\varkappa_a$ is the Frobenius-Schur indicator~\cite{kitaev2006anyons}.
	\end{proof}
\end{lemma_rep}

\clearpage
% !TeX spellcheck = en_US
% This line sets the project root file.
% !TEX root = ../WW_EntanglementEntropy.tex
%

\section{Loop-gas results}\label{app:SN_results}

In this section, given a fusion category $\C$, an $n$-tuple of simple objects $\vec{x}_n:=(x_1,x_2,\ldots,x_n)$, and a fixed simple object $a$, we use the notation
\begin{align}
	N_{a}(\vec{x}_n):=\sum_{\vec{y}_{n-2}}N_{x_1,x_2}^{y_1}N_{y_1,x_3}^{y_2}\ldots N_{y_{n-2},x_n}^{a},
\end{align}
where $\vec{y}_{n-2}:=(y_1,y_2,\ldots,y_{n-2})$, and the sum is over all tuples of simple objects in $\C$. $N_{a}(\vec{x}_n)$ counts the number of ways $\vec{x}_{n}$ can fuse to $a$. When it can easily be inferred, we omit the subscript on the tuple $\vec{x}$.

\begin{lemma_rep}[\ref{lem:summingds}]\label{lem:summingds_pf}
	Let $\C$ be a unitary fusion category, then
	\begin{align}
		\sum_{\vec{x}_n}N_a(\vec{x}_n)\prod_{j\leq n}d_{x_j} & =d_a\D^{2(n-1)},\label{eqn:sumd_1}
	\end{align}
	where $\D=\sqrt{\sum_a d_a^2}$ is the total quantum dimension of $\C$.
	\begin{proof}
		We proceed inductively.

		When $n=1$, $N_a(x)=\indicator{x=a}=N_{1,x}^{a}$, and \cref{eqn:sumd_1} reduces to $d_a=d_a$.

		Assume \cref{eqn:sumd_1} holds for the fusion of $n$ objects. Recall that or any fusion category, we have
		\begin{align}
			d_ad_b & =\sum_c N_{bc}^{a}d_c, \label{eqn:TheOneAbove}
		\end{align}
		and this holds for any cyclic permutation of the indices on $N_{bc}^{a}$. We now obtain
		\begin{align}
			\sum_{\vec{x}_{n+1}}N_a(\vec{x}_{n+1})\prod_{j\leq n+1}d_{x_j} & = \sum_{\vec{x}_n,y_{n-1}}N_{y_{n-1}}(\vec{x}_n)\prod_{j\leq n}d_{x_j}\sum_{x_{n+1}}N_{y_{n-1},x_{n+1}}^{a}d_{x_{n+1}} \\
			                                                               & =\D^{2(n-1)}\sum_{x_{n+1},y_{n-1}}N_{y_{n-1},x_{n+1}}^{a}d_{y_{n-1}}d_{x_{n+1}}                                        \\
			                                                               & =\D^{2(n-1)}d_a\sum_{y_{n-1}}d_{y_{n-1}}^2                                                                             \\
			                                                               & =d_a\D^{2n},
		\end{align}
		where in the second line we used the induction assumption (\cref{eqn:sumd_1}), and in the third line we used \cref{eqn:TheOneAbove}.
	\end{proof}
\end{lemma_rep}

\begin{lemma_rep}[\ref{lem:sumlog}]\label{lem:sumlog_pf}
	Let $\C$ a unitary fusion category. For the fusion of $n$ objects $\vec{x}=(x_1,x_2,\ldots,x_n)$, with $n>1$, we have
	\begin{align}
		\sum_{\vec{x}_n} N_a(\vec{x}_n)\frac{\prod_{j\leq n} d_{x_j}}{\D^{2(n-1)}}\log\prod_{k\leq n} d_{x_k} & =n d_a \sum_x \frac{d_x^2\log d_x}{\D^2}.\label{eqn:sumlog_1}
	\end{align}
	\begin{proof}
		We prove the claim inductively.
		The base case is when $n=2$.
		\begin{align}
			\sum_{x_1,x_2} N_{x_1,x_2}^{a}\frac{d_{x_1}d_{x_2}}{\D^{2}}(\log d_{x_1}+\log d_{x_2}) & =\sum_{x_1,x_2} N_{x_1,x_2}^{a}\frac{d_{x_1}d_{x_2}}{\D^{2}}\log d_{x_1}+\sum_{x_1,x_2} N_{x_1,x_2}^{a}\frac{d_{x_1}d_{x_2}}{\D^{2}}\log d_{x_2} \\
			                                                                                       & =d_a\sum_{x_1} \frac{d_{x_1}^2}{\D^{2}}\log d_{x_1}+d_a\sum_{x_2} \frac{d_{x_2}^2}{\D^{2}}\log d_{x_2}                                           \\
			                                                                                       & =2 d_a \sum_x \frac{d_x^2\log d_x}{\D^2}.
		\end{align}
		Assume \cref{eqn:sumlog_1} holds for $n$-tuples, then
		\begin{align}
			\sum_{\vec{x}_n,x_{n+1}} & N_a(\vec{x}_{n+1})\frac{\prod_{j\leq n+1} d_{x_j}}{\D^{2n}}\log\prod_{k\leq n+1} d_{x_k}  =
			\!\!\sum_{\substack{\vec{x}_n                                                                                                                   \\y_{n-1},x_{n+1}}} \!\!N_{y_{n-1}}(\vec{x}_{n})N_{y_{n-1},x_{n+1}}^{a}\frac{\prod_{j\leq n} d_{x_j}}{\D^{2n}}d_{x_{n+1}}\log(\prod_{k\leq n} d_{x_k}d_{n+1})\\
			                         & =\sum_{
			\substack{\vec{x}_n                                                                                                                             \\y_{n-1},x_{n+1}}} N_{y_{n-1}}(\vec{x}_{n})\frac{\prod_{j\leq n}d_{x_j}}{\D^{2n}}N_{y_{n-1},x_{n+1}}^{a}d_{x_{n+1}}\left( \log\prod_{k\leq n} d_{x_k}+\log d_{x_{n+1}}\right)
			\\
			                         & =n\sum_{x}\frac{d_x^2\log d_x}{\D^2}\sum_{y_{n-1},x_{n+1}}\frac{N_{y_{n-1},x_{n+1}}^{a}d_{y_{n-1}}d_{x_{n+1}}}{\D^2}
			+\sum_{y_{n-1},x_{n+1}} N_{y_{n-1},x_{n+1}}^{a} \frac{d_{y_{n-1}}d_{x_{n+1}}\log d_{x_{n+1}}}{\D^{2}}                                           \\
			                         & =nd_a\sum_{x}\frac{d_x^2\log d_x}{\D^2}\sum_{x_{n+1}}\frac{d_{x_{n+1}}^2}{\D^2}
			+d_a\sum_{x_{n+1}} \frac{d_{x_{n+1}}^2\log d_{x_{n+1}}}{\D^{2}}
			\\
			                         & =(n+1) d_a \sum_x \frac{d_x^2\log d_x}{\D^2}.
		\end{align}
	\end{proof}
\end{lemma_rep}

\begin{lemma_rep}[\ref{lem:prtree_pf}]\label{lem:prtree_pf}
	Let $\C$ a unitary fusion category. Given a fixed fusion outcome $a$ on $n$ simple objects, the probability of the tree
	\begin{align}
		\begin{array}{c}
			\includeTikz{treeextraA}{
				\begin{tikzpicture}
					\draw(0,0)--(1.75,1.75);
					\draw[dotted](1.75,1.75)--(2,2);
					\draw(2,2)--(3,3);
					\begin{scope}
						\clip(0,0)--(3,3)--(6,3)--(6,0)--(0,0);
						\draw(1,0)--(0,1);
						\draw(2,0)--(0,2);
						\draw(3,0)--(0,3);
						\draw(4.5,0)--(0,4.5);
						\draw(5.5,0)--(0,5.5);
					\end{scope}
					\node[below] at (0,0) {$x_1$};
					\node[below] at (1,0) {$x_2$};
					\node[below] at (2,0) {$x_3$};
					\node[below] at (3,0) {$x_4$};
					\node[below] at (4.5,0) {$x_{n-1}$};
					\node[below] at (5.5,0) {$x_{n}$};
					\node[above right] at (3,3) {$a$};
					\node[above left] at (.75,.75) {$y_1$};
					\node[above left] at (1.25,1.25) {$y_2$};
					\node[above left] at (2.5,2.5) {$y_{n-2}$};
					\node[below] at (.5,.5) {\tiny{$\mu_1$}};
					\node[below] at (1,1) {\tiny{$\mu_2$}};
					\node[below] at (1.5,1.5) {\tiny{$\mu_3$}};
					\node[right] at (2.25,2.25) {\tiny{$\mu_{n-2}$}};
				\end{tikzpicture}
			}
		\end{array},\label{eqn:treeextraA}
	\end{align}
	in the ground state of a topological loop-gas (Levin-Wen or Walker-Wang) model is
	\begin{align}
		\Pr[\vec{x},\vec{y},\vec{\mu}|a] & =\frac{\prod_{j\leq n}\Pr[x_j]}{\Pr[a]\prod_{k\leq n} d_{x_k}}d_a \\
		                                 & =\frac{\prod_{j\leq n} d_{x_j}}{d_{a}\D^{2(n-1)}}.
	\end{align}
	\begin{proof}
		Given a pair of objects $a,b$, the probability that they fuse to $c$ is given by~\cite{preskillnotes,bullivant2016entropic}
		\begin{align}
			\Pr[a\otimes b\to c]=\frac{N_{ab}^c d_c}{d_ad_b },
		\end{align}
		so the probability that $x_1\otimes x_2\otimes\ldots\otimes x_n\to a$ is
		\begin{align}
			\Pr[x_1\otimes x_2\otimes x_3\otimes \ldots \otimes x_n\to a]= & \sum_{\vec{y}}\Pr[x_1\otimes x_2\to y_1]\Pr[y_1\otimes x_3\to y_2]\cdots \Pr[y_{n-2}\otimes x_n\to a] \\
			=                                                              & \frac{N_{a}(\vec{x})}{\prod_{j\leq n} d_{x_j}}d_{a},
		\end{align}
		where
		\begin{align}
			N_{a}(\vec{x}) := & \sum_{\vec{y}}N_{x_1x_2}^{y_1}N_{y_1x_3}^{y_2}\ldots N_{y_{n-2}x_n}^{a} \\
			=                 & \sum_{\vec{y}}N_{a}(\vec{x},\vec{y}).
		\end{align}

		The probability of a configuration is
		\begin{align}
			\Pr[\vec{x},\vec{y}|a] & =\Pr[x_1\otimes x_2\otimes x_3\otimes \ldots \otimes x_n\to a]\frac{\prod_{j\leq n}\Pr[x_j]}{\Pr[a]} \\
			                       & =\frac{N_{a}(\vec{x},\vec{y})d_a}{\prod_{k\leq n} d_{x_k}}\frac{\prod_{j\leq n}\Pr[x_j]}{\Pr[a]},
		\end{align}
		where $\Pr[x_i]=d_{x_i}^2/\D^2$. For a fixed $\vec{x}$ and $\vec{y}$, all (allowed) $\vec{\mu}$ are equally likely, and there are $N_a(\vec{x},\vec{y})$ such configurations, so
		\begin{align}
			\Pr[\vec{x},\vec{y},\vec{\mu}|a] & =\frac{\prod_{j\leq n}\Pr[x_j]}{\Pr[a]\prod_{k\leq n} d_{x_k}}d_a \\
			                                 & =\frac{\prod_{j\leq n} d_{x_j}}{d_{a}\D^{2(n-1)}}.
		\end{align}
		Lemma~\ref{lem:summingds} can be used to show these are properly normalized.
	\end{proof}
\end{lemma_rep}

\begin{thm_rep}[\ref{thm:WWentropyexamples}]\label{thm:WWentropyexamples_pf}
	For a Walker-Wang model defined by a unitary premodular category of one of the following types:
	\begin{enumerate}
		\item $\C=\cat{A}\boxtimes\cat{B}$, where $\cat{A}$ is symmetric and $\cat{B}$ is modular~\cite{bullivant2016entropic},
		\item $\C$ pointed,
		\item $\rk{\C}<6$ and multiplicity free,
		\item $\rk{\C}=\rk{\mug{\C}}+1$ and $d_x=\D_{\mug{\C}}$, where $x$ is the additional object,
	\end{enumerate}
	then \cref{eqn:conjecture} holds. As a consequence, the topological entanglement entropy (defined using the regions in \cref{fig:WWregionsblk}) in the bulk is given by
	\begin{align}
		\WWTEE & =\log\D_{\mug{\C}}^2.\label{eqn:WalkerWangBulkTEE_app}
	\end{align}
	As special cases, this includes
	\begin{align}
		\WWTEE_{\text{modular}}   & =0         \\
		\WWTEE_{\text{symmetric}} & =\log \D^2
	\end{align}
	\begin{proof}~
		\subsection{Case 1}
		Using the premodular trap (\cref{cor:premodulartrap}), we have the matrix elements of $\S_c^\dagger\S_c$
		\begin{align}
			\left[\S_c^\dagger \S_c\right]_{(a,\alpha),(b,\beta)} & =\sqrt{\frac{d_c}{d_ad_b}}\sum_{x\in\mug{\C},\mu}\sqrt{d_x}
			\begin{array}{c}
				\includeTikz{SdotS}{
					\begin{tikzpicture}[scale=.75]
						\pgfmathsetmacro{\s}{sqrt(2)/2}
						\draw (0,-.5)--(0,.5)--(-\s,\s+.5) (0,.5)--(\s,\s+.5);
						\draw (0,-.5)--(-\s,-\s-.5) (0,-.5)--(\s,-\s-.5);
						\draw (-\s,\s+.5) to [out=180+45,in=180-45] (-\s,-\s-.5);
						\draw (\s,\s+.5) to [out=-45,in=45] (\s,-\s-.5);
						\draw[] (\s,-\s-.5)to[out=270,in=270] (2.5,0)to[out=90, in=90](-\s,\s+.5);
						\node[anchor=east,inner sep=.2] at (0,.5){\strut$\mu$};
						\node[anchor=east,inner sep=.2] at (0,-.5){\strut$\mu$};
						\node[anchor=west,inner sep=.5] at (0,0){\strut$x$};
						\node[inner sep=.5] at (1,0){\strut$b$};
						\node[inner sep=.5] at (-1,0){\strut$\dual{a}$};
						\node[inner sep=.5,anchor=south east] at (-\s,\s+.5){\strut$\alpha$};
						\node[inner sep=.5,anchor=north east] at (\s,-\s-.5){\strut$\beta$};
						\node[anchor=east,inner sep=.5] at (2.5,0){\strut$\dual{c}$};
						\node[anchor=south west,inner sep=.2] at (-\s/2,.5+\s/2) {\strut$a$};
						\node[anchor=north east,inner sep=.2] at (\s/2,-.5-\s/2) {\strut$\dual{b}$};
					\end{tikzpicture}
				}
			\end{array}\label{eqn:SdotS}.
		\end{align}
		If $\C$ is symmetric, $\mug{\C}=\C$, and
		\begin{align}
			\left[\S_c^\dagger \S_c\right]_{(a,\alpha),(b,\beta)} & =\indicator{c=1}d_ad_b.
		\end{align}
		This matrix is rank 1, with eigenvalue $\D^2$.
		If $\C$ is modular, $\mug{\C}=\vvec{}$, and
		\begin{align}
			\left[\S_c^\dagger \S_c\right]_{(a,\alpha),(b,\beta)} & =\indicator{a=b}\indicator{\alpha=\beta}\indicator{\dual{a}\otimes a=c}d_c.
		\end{align}
		For fixed $c$, this matrix is rank $\sum_c N_{\dual{a},a}^c$, with all eigenvalues equal to $d_c$.

		\subsection{Case 2}
		If $\C$ is pointed (every simple object has dimension 1), then $\S_c=0$ unless $c=1$. In this case, the fusion rules are given by a finite Abelian group $A$~\cite{Joyal_1993,MR3242743}, and $\mug{\C}=A^\prime$ has fusion rules given by a subgroup.
		From \cref{lem:productS,lem:sumS}, along with symmetries of the $\S_1$ matrix proven in \onlinecite{kitaev2006anyons} we know that
		\begin{align}
			\left[\S_1^\dagger\S_1\right]_{ab}=\sum_{c\in\mug{\C}}N_{a\dual{b}}^cd_c,\label{eqn:S1S1}
		\end{align}
		so
		\begin{align}
			\left[\S_1^\dagger\S_1\right]_{ab}=1\iff a\in bA^\prime.
		\end{align}
		Therefore, $[\S_1^\dagger \S_1]$ is a block matrix, with $[A:A^\prime]=|A|/|A^\prime|$ blocks, labeled by the cosets of $A^\prime$, each full of ones. Therefore, there are $[A:A^\prime]$ eigenvalues, identically $\D^2_{\mug{\C}}$. The entropy is given by
		\begin{align}
			\WWTEE & =\log\D_{\mug{\C}}^2.
		\end{align}

		\subsection{Case 3}

		Case 3 is proven explicitly in the attached Mathematica file~\cite{premoddata}. Classification of the fusion rings for ranks 2-5 can be found in \onlinecite{rk2,rk3,Rowell2009On,rk4,rk5}, along with \onlinecite{Yu_2020}. Additionally, all multiplicity free fusion rings for ranks 1-6 can be found at \onlinecite{anyonwiki}. From this, explicit $F$ and $R$ data can be found. The list of categories, along with their properties, is included beginning on \cpageref{sec:case3}.

		\subsection{Case 4}

		It is straightforward to check that if $a$ or $b$ are in $\mug{\C}$, then
		\begin{align}
			\left[\S_c^\dagger \S_c\right]_{(a,\alpha),(b,\beta)} & =\indicator{c=1}d_ad_b\indicator{a\in\mug{\C}}\indicator{b\in\mug{\C}},
		\end{align}
		so $\S_c^\dagger \S_c$ has the form
		\begin{align}
			[\S_c^\dagger \S_c]=
			\kbordermatrix{ & \mug{\C}                   &             \\
			\mug{\C}        & d_ad_b\indicator{c=1}      & 0           \\
			                & 0                          & X_c
			}
			=U
			\kbordermatrix{ & \mug{\C}                   &             \\
			\mug{\C}        & \begin{matrix}
					\D_{\mug{\C}}^2\indicator{c=1} & 0 & \cdots \\0&0&\cdots\\\vdots&\vdots&\ddots
				\end{matrix} & 0           \\
			                & 0                          & \tilde{X}_c
			}U^\dagger.
		\end{align}

		From the top left block, we have an eigenvector of $\S_1^\dagger\S_1$ with entries $v_a=\indicator{a\in\mug{\C}}d_a$ with eigenvalue $\D^2_{\mug{\C}}$. From \cref{lem:productS,lem:sumS}, along with symmetries of the $\S_1$ matrix proven in \onlinecite{kitaev2006anyons} we know that
		\begin{align}
			\left[\S_1^\dagger\S_1\right]_{ab}=\sum_{c\in\mug{\C}}N_{a\dual{b}}^cd_c.
		\end{align}

		The vector with entries $w_a=d_a$ is also an eigenvector with the same eigenvalue:
		\begin{align}
			\sum_{b\in \C}\sum_{c\in\mug{\C}}N_{a\dual{b}}^c d_cd_b & =\sum_{c\in\mug{\C}}d_ad_c^2 \\
			                                                        & =\D^2_{\mug{\C}}d_a,
		\end{align}
		so we have an orthogonal vector $w-v$ with eigenvalue $\D^2_{\mug{\C}}$.

		If $\rk{\C}=\rk{\mug{\C}}+1$ and the additional object has $d_x=\D_{\mug{\C}}$, then all other eigenvalues must be $0$ since $\Tr \S_1^\dagger S_1=2\D^2_{\mug{\C}}$ and $\D^2=\D^2_{\mug{\C}}+d_x^2$. The entropy of the Walker-Wang model in the bulk is
		\begin{align}
			\WWTEE & =\log\D_{\mug{\C}}^2.
		\end{align}
	\end{proof}
\end{thm_rep}

\subsection{Small category data}\label{sec:case3}

Data for small categories. ``Valid" indicates that the pentagon, hexagon, and ribbon equations, along with unitarity, are true. ``TY" indicates that the category has the property defined in Case 4 of \cref{thm:WWentropyexamples}.

Full data, including explicit $F$ and $R$ symbols is provided in the attached Mathematica files, also available at \onlinecite{premoddata}. Note that these may take \emph{a very long time} to check. This is due to the complicated algebraic integers occurring, and Mathematica needing to simplify using the functions ``Simplify'' and ``RootReduce''.

Categories are named $\FR{a}{b}{c}{x}$ according to their fusion ring $\mathrm{FR}^{a,b}_{c}$ from \onlinecite{anyonwiki}, along with their categorification ID $x$. Highlighted categories do not fall within any of the other cases in \cref{thm:WWentropyexamples}.

\vfill

\centering\begin{tabular}{
	!{\vrule width 1pt}>{\columncolor[gray]{.9}[\tabcolsep]}c!{\vrule width 1pt}c!{\vrule width 1pt}c!{\vrule width 1pt}c!{\vrule width 1pt}c!{\vrule width 1pt}c!{\vrule width 1pt}c!{\vrule width 1pt}c!{\vrule width 1pt}c!{\vrule width 1pt}c!{\vrule width 1pt}c!{\vrule width 1pt}c!{\vrule width 1pt}}
	\toprule[1pt]
	\rowcolor[gray]{.9}[\tabcolsep] Cat. ID               & Rank & $\D^2$                                & Valid     & $\rk{\mug{\C}}$ & Premodular? & Pointed?  & TY?       & TEE      & $\log\D_{\mug{\C}}^2$ & Conjecture true? \\
	\toprule[1pt]
	$\FR{2}{0}{1}{0}$                                     & 2    & 2                                     & \ding{51} & 2               & Symm.       & \ding{51} &           & $\log 2$ & $\log 2$              & \ding{51}        \\
	$\FR{2}{0}{1}{1}$                                     & 2    & 2                                     & \ding{51} & 1               & Mod.        & \ding{51} & \ding{51} & 0        & 0                     & \ding{51}        \\
	$\FR{2}{0}{1}{2}$                                     & 2    & 2                                     & \ding{51} & 2               & Symm.       & \ding{51} &           & $\log 2$ & $\log 2$              & \ding{51}        \\
	$\FR{2}{0}{1}{3}$                                     & 2    & 2                                     & \ding{51} & 1               & Mod.        & \ding{51} & \ding{51} & 0        & 0                     & \ding{51}        \\
	$\FR{2}{0}{2}{0}$                                     & 2    & $\frac{1}{2} \left(\sqrt{5}+5\right)$ & \ding{51} & 1               & Mod.        &           &           & 0        & 0                     & \ding{51}        \\
	$\FR{2}{0}{2}{1}$                                     & 2    & $\frac{1}{2} \left(\sqrt{5}+5\right)$ & \ding{51} & 1               & Mod.        &           &           & 0        & 0                     & \ding{51}        \\
	\toprule[1pt]
	$\FR{3}{0}{1}{0}$                                     & 3    & 4                                     & \ding{51} & 1               & Mod.        &           &           & 0        & 0                     & \ding{51}        \\
	$\FR{3}{0}{1}{1}$                                     & 3    & 4                                     & \ding{51} & 1               & Mod.        &           &           & 0        & 0                     & \ding{51}        \\
	$\FR{3}{0}{1}{2}$                                     & 3    & 4                                     & \ding{51} & 1               & Mod.        &           &           & 0        & 0                     & \ding{51}        \\
	$\FR{3}{0}{1}{3}$                                     & 3    & 4                                     & \ding{51} & 1               & Mod.        &           &           & 0        & 0                     & \ding{51}        \\
	$\FR{3}{0}{1}{4}$                                     & 3    & 4                                     & \ding{51} & 1               & Mod.        &           &           & 0        & 0                     & \ding{51}        \\
	$\FR{3}{0}{1}{5}$                                     & 3    & 4                                     & \ding{51} & 1               & Mod.        &           &           & 0        & 0                     & \ding{51}        \\
	$\FR{3}{0}{1}{6}$                                     & 3    & 4                                     & \ding{51} & 1               & Mod.        &           &           & 0        & 0                     & \ding{51}        \\
	$\FR{3}{0}{1}{7}$                                     & 3    & 4                                     & \ding{51} & 1               & Mod.        &           &           & 0        & 0                     & \ding{51}        \\
	$\FR{3}{0}{2}{0}$                                     & 3    & 6                                     & \ding{51} & 3               & Symm.       &           &           & $\log 6$ & $\log 6$              & \ding{51}        \\
	\rowcolor{nicegreen!10}[\tabcolsep] $\FR{3}{0}{2}{1}$ & 3    & 6                                     & \ding{51} & 2               & \ding{51}   &           &           & $\log 2$ & $\log 2$              & \ding{51}        \\
	\rowcolor{nicegreen!10}[\tabcolsep] $\FR{3}{0}{2}{2}$ & 3    & 6                                     & \ding{51} & 2               & \ding{51}   &           &           & $\log 2$ & $\log 2$              & \ding{51}        \\
	$\FR{3}{0}{3}{0}$                                     & 3    & $\sim9.30$                            & \ding{51} & 1               & Mod.        &           &           & 0        & 0                     & \ding{51}        \\
	$\FR{3}{0}{3}{1}$                                     & 3    & $\sim9.30$                            & \ding{51} & 1               & Mod.        &           &           & 0        & 0                     & \ding{51}        \\
	$\FR{3}{2}{1}{0}$                                     & 3    & 3                                     & \ding{51} & 3               & Symm.       & \ding{51} &           & $\log 3$ & $\log 3$              & \ding{51}        \\
	$\FR{3}{2}{1}{1}$                                     & 3    & 3                                     & \ding{51} & 1               & Mod.        & \ding{51} &           & 0        & 0                     & \ding{51}        \\
	$\FR{3}{2}{1}{2}$                                     & 3    & 3                                     & \ding{51} & 1               & Mod.        & \ding{51} &           & 0        & 0                     & \ding{51}        \\
	\toprule[1pt]
\end{tabular}

\clearpage

\newgeometry{left=17mm,right=17mm,top=8mm,bottom=8mm,ignoreall, noheadfoot}
\centering\resizebox*{!}{\textheight}{
	\begin{tabular}{
		!{\vrule width 1pt}>{\columncolor[gray]{.9}[\tabcolsep]}c!{\vrule width 1pt}c!{\vrule width 1pt}c!{\vrule width 1pt}c!{\vrule width 1pt}c!{\vrule width 1pt}c!{\vrule width 1pt}c!{\vrule width 1pt}c!{\vrule width 1pt}c!{\vrule width 1pt}c!{\vrule width 1pt}c!{\vrule width 1pt}c!{\vrule width 1pt}}
		\toprule[1pt]
		\rowcolor[gray]{.9}[\tabcolsep] Cat. ID               & Rank & $\D^2$                                   & Valid     & $\rk{\mug{\C}}$ & Premodular? & Pointed?  & TY? & TEE       & $\log\D_{\mug{\C}}^2$ & Conjecture true? \\
		\toprule[1pt]
		$\FR{4}{0}{1}{0}$                                     & 4    & 4                                        & \ding{51} & 4               & Symm.       & \ding{51} &     & $\log 4$  & $\log 4$              & \ding{51}        \\
		$\FR{4}{0}{1}{1}$                                     & 4    & 4                                        & \ding{51} & 1               & Mod.        & \ding{51} &     & 0         & 0                     & \ding{51}        \\
		$\FR{4}{0}{1}{2}$                                     & 4    & 4                                        & \ding{51} & 4               & Symm.       & \ding{51} &     & $\log 4$  & $\log 4$              & \ding{51}        \\
		$\FR{4}{0}{1}{3}$                                     & 4    & 4                                        & \ding{51} & 1               & Mod.        & \ding{51} &     & 0         & 0                     & \ding{51}        \\
		$\FR{4}{0}{1}{4}$                                     & 4    & 4                                        & \ding{51} & 2               & \ding{51}   & \ding{51} &     & $\log 2$  & $\log 2$              & \ding{51}        \\
		$\FR{4}{0}{1}{5}$                                     & 4    & 4                                        & \ding{51} & 2               & \ding{51}   & \ding{51} &     & $\log 2$  & $\log 2$              & \ding{51}        \\
		$\FR{4}{0}{1}{6}$                                     & 4    & 4                                        & \ding{51} & 2               & \ding{51}   & \ding{51} &     & $\log 2$  & $\log 2$              & \ding{51}        \\
		$\FR{4}{0}{1}{7}$                                     & 4    & 4                                        & \ding{51} & 2               & \ding{51}   & \ding{51} &     & $\log 2$  & $\log 2$              & \ding{51}        \\
		$\FR{4}{0}{1}{8}$                                     & 4    & 4                                        & \ding{51} & 4               & Symm.       & \ding{51} &     & $\log 4$  & $\log 4$              & \ding{51}        \\
		$\FR{4}{0}{1}{9}$                                     & 4    & 4                                        & \ding{51} & 1               & Mod.        & \ding{51} &     & 0         & 0                     & \ding{51}        \\
		$\FR{4}{0}{1}{10}$                                    & 4    & 4                                        & \ding{51} & 4               & Symm.       & \ding{51} &     & $\log 4$  & $\log 4$              & \ding{51}        \\
		$\FR{4}{0}{1}{11}$                                    & 4    & 4                                        & \ding{51} & 1               & Mod.        & \ding{51} &     & 0         & 0                     & \ding{51}        \\
		$\FR{4}{0}{1}{12}$                                    & 4    & 4                                        & \ding{51} & 2               & \ding{51}   & \ding{51} &     & $\log 2$  & $\log 2$              & \ding{51}        \\
		$\FR{4}{0}{1}{13}$                                    & 4    & 4                                        & \ding{51} & 2               & \ding{51}   & \ding{51} &     & $\log 2$  & $\log 2$              & \ding{51}        \\
		$\FR{4}{0}{1}{14}$                                    & 4    & 4                                        & \ding{51} & 2               & \ding{51}   & \ding{51} &     & $\log 2$  & $\log 2$              & \ding{51}        \\
		$\FR{4}{0}{1}{15}$                                    & 4    & 4                                        & \ding{51} & 2               & \ding{51}   & \ding{51} &     & $\log 2$  & $\log 2$              & \ding{51}        \\
		$\FR{4}{0}{1}{16}$                                    & 4    & 4                                        & \ding{51} & 1               & Mod.        & \ding{51} &     & 0         & 0                     & \ding{51}        \\
		$\FR{4}{0}{1}{17}$                                    & 4    & 4                                        & \ding{51} & 2               & \ding{51}   & \ding{51} &     & $\log 2$  & $\log 2$              & \ding{51}        \\
		$\FR{4}{0}{1}{18}$                                    & 4    & 4                                        & \ding{51} & 1               & Mod.        & \ding{51} &     & 0         & 0                     & \ding{51}        \\
		$\FR{4}{0}{1}{19}$                                    & 4    & 4                                        & \ding{51} & 2               & \ding{51}   & \ding{51} &     & $\log 2$  & $\log 2$              & \ding{51}        \\
		$\FR{4}{0}{1}{20}$                                    & 4    & 4                                        & \ding{51} & 1               & Mod.        & \ding{51} &     & 0         & 0                     & \ding{51}        \\
		$\FR{4}{0}{1}{21}$                                    & 4    & 4                                        & \ding{51} & 1               & Mod.        & \ding{51} &     & 0         & 0                     & \ding{51}        \\
		$\FR{4}{0}{1}{22}$                                    & 4    & 4                                        & \ding{51} & 1               & Mod.        & \ding{51} &     & 0         & 0                     & \ding{51}        \\
		$\FR{4}{0}{1}{23}$                                    & 4    & 4                                        & \ding{51} & 1               & Mod.        & \ding{51} &     & 0         & 0                     & \ding{51}        \\
		$\FR{4}{0}{1}{24}$                                    & 4    & 4                                        & \ding{51} & 1               & Mod.        & \ding{51} &     & 0         & 0                     & \ding{51}        \\
		$\FR{4}{0}{1}{25}$                                    & 4    & 4                                        & \ding{51} & 2               & \ding{51}   & \ding{51} &     & $\log 2$  & $\log 2$              & \ding{51}        \\
		$\FR{4}{0}{1}{26}$                                    & 4    & 4                                        & \ding{51} & 1               & Mod.        & \ding{51} &     & 0         & 0                     & \ding{51}        \\
		$\FR{4}{0}{1}{27}$                                    & 4    & 4                                        & \ding{51} & 2               & \ding{51}   & \ding{51} &     & $\log 2$  & $\log 2$              & \ding{51}        \\
		$\FR{4}{0}{1}{28}$                                    & 4    & 4                                        & \ding{51} & 1               & Mod.        & \ding{51} &     & 0         & 0                     & \ding{51}        \\
		$\FR{4}{0}{1}{29}$                                    & 4    & 4                                        & \ding{51} & 1               & Mod.        & \ding{51} &     & 0         & 0                     & \ding{51}        \\
		$\FR{4}{0}{1}{30}$                                    & 4    & 4                                        & \ding{51} & 1               & Mod.        & \ding{51} &     & 0         & 0                     & \ding{51}        \\
		$\FR{4}{0}{1}{31}$                                    & 4    & 4                                        & \ding{51} & 1               & Mod.        & \ding{51} &     & 0         & 0                     & \ding{51}        \\
		$\FR{4}{0}{2}{0}$                                     & 4    & $\sqrt{5}+5$                             & \ding{51} & 2               & \ding{51}   &           &     & $\log 2$  & $\log 2$              & \ding{51}        \\
		$\FR{4}{0}{2}{1}$                                     & 4    & $\sqrt{5}+5$                             & \ding{51} & 1               & Mod.        &           &     & 0         & 0                     & \ding{51}        \\
		$\FR{4}{0}{2}{2}$                                     & 4    & $\sqrt{5}+5$                             & \ding{51} & 2               & \ding{51}   &           &     & $\log 2$  & $\log 2$              & \ding{51}        \\
		$\FR{4}{0}{2}{3}$                                     & 4    & $\sqrt{5}+5$                             & \ding{51} & 1               & Mod.        &           &     & 0         & 0                     & \ding{51}        \\
		$\FR{4}{0}{2}{4}$                                     & 4    & $\sqrt{5}+5$                             & \ding{51} & 2               & \ding{51}   &           &     & $\log 2$  & $\log 2$              & \ding{51}        \\
		$\FR{4}{0}{2}{5}$                                     & 4    & $\sqrt{5}+5$                             & \ding{51} & 1               & Mod.        &           &     & 0         & 0                     & \ding{51}        \\
		$\FR{4}{0}{2}{6}$                                     & 4    & $\sqrt{5}+5$                             & \ding{51} & 2               & \ding{51}   &           &     & $\log 2$  & $\log 2$              & \ding{51}        \\
		$\FR{4}{0}{2}{7}$                                     & 4    & $\sqrt{5}+5$                             & \ding{51} & 1               & Mod.        &           &     & 0         & 0                     & \ding{51}        \\
		$\FR{4}{0}{3}{0}$                                     & 4    & 10                                       & \ding{51} & 4               & Symm.       &           &     & $\log 10$ & $\log 10$             & \ding{51}        \\
		\rowcolor{nicegreen!10}[\tabcolsep] $\FR{4}{0}{3}{1}$ & 4    & 10                                       & \ding{51} & 2               & \ding{51}   &           &     & $\log 2$  & $\log 2$              & \ding{51}        \\
		\rowcolor{nicegreen!10}[\tabcolsep] $\FR{4}{0}{3}{2}$ & 4    & 10                                       & \ding{51} & 2               & \ding{51}   &           &     & $\log 2$  & $\log 2$              & \ding{51}        \\
		\rowcolor{nicegreen!10}[\tabcolsep] $\FR{4}{0}{3}{3}$ & 4    & 10                                       & \ding{51} & 2               & \ding{51}   &           &     & $\log 2$  & $\log 2$              & \ding{51}        \\
		\rowcolor{nicegreen!10}[\tabcolsep] $\FR{4}{0}{3}{4}$ & 4    & 10                                       & \ding{51} & 2               & \ding{51}   &           &     & $\log 2$  & $\log 2$              & \ding{51}        \\
		$\FR{4}{0}{4}{0}$                                     & 4    & 4 $\sqrt{2}+8$                           & \ding{51} & 2               & \ding{51}   &           &     & $\log 2$  & $\log 2$              & \ding{51}        \\
		$\FR{4}{0}{4}{1}$                                     & 4    & 4 $\sqrt{2}+8$                           & \ding{51} & 2               & \ding{51}   &           &     & $\log 2$  & $\log 2$              & \ding{51}        \\
		$\FR{4}{0}{5}{0}$                                     & 4    & $\frac{1}{2} \left(5 \sqrt{5}+15\right)$ & \ding{51} & 1               & Mod.        &           &     & 0         & 0                     & \ding{51}        \\
		$\FR{4}{0}{5}{1}$                                     & 4    & $\frac{1}{2} \left(5 \sqrt{5}+15\right)$ & \ding{51} & 1               & Mod.        &           &     & 0         & 0                     & \ding{51}        \\
		$\FR{4}{0}{5}{2}$                                     & 4    & $\frac{1}{2} \left(5 \sqrt{5}+15\right)$ & \ding{51} & 1               & Mod.        &           &     & 0         & 0                     & \ding{51}        \\
		$\FR{4}{0}{5}{3}$                                     & 4    & $\frac{1}{2} \left(5 \sqrt{5}+15\right)$ & \ding{51} & 1               & Mod.        &           &     & 0         & 0                     & \ding{51}        \\
		$\FR{4}{0}{6}{0}$                                     & 4    & $\sim19.24$                              & \ding{51} & 1               & Mod.        &           &     & 0         & 0                     & \ding{51}        \\
		$\FR{4}{0}{6}{1}$                                     & 4    & $\sim19.24$                              & \ding{51} & 1               & Mod.        &           &     & 0         & 0                     & \ding{51}        \\
		$\FR{4}{2}{1}{0}$                                     & 4    & 4                                        & \ding{51} & 4               & Symm.       & \ding{51} &     & $\log 4$  & $\log 4$              & \ding{51}        \\
		$\FR{4}{2}{1}{1}$                                     & 4    & 4                                        & \ding{51} & 1               & Mod.        & \ding{51} &     & 0         & 0                     & \ding{51}        \\
		$\FR{4}{2}{1}{2}$                                     & 4    & 4                                        & \ding{51} & 2               & \ding{51}   & \ding{51} &     & $\log 2$  & $\log 2$              & \ding{51}        \\
		$\FR{4}{2}{1}{3}$                                     & 4    & 4                                        & \ding{51} & 1               & Mod.        & \ding{51} &     & 0         & 0                     & \ding{51}        \\
		$\FR{4}{2}{1}{4}$                                     & 4    & 4                                        & \ding{51} & 4               & Symm.       & \ding{51} &     & $\log 4$  & $\log 4$              & \ding{51}        \\
		$\FR{4}{2}{1}{5}$                                     & 4    & 4                                        & \ding{51} & 1               & Mod.        & \ding{51} &     & 0         & 0                     & \ding{51}        \\
		$\FR{4}{2}{1}{6}$                                     & 4    & 4                                        & \ding{51} & 2               & \ding{51}   & \ding{51} &     & $\log 2$  & $\log 2$              & \ding{51}        \\
		$\FR{4}{2}{1}{7}$                                     & 4    & 4                                        & \ding{51} & 1               & Mod.        & \ding{51} &     & 0         & 0                     & \ding{51}        \\
		\toprule[1pt]
	\end{tabular}
}

\centering\resizebox*{!}{\textheight}{
	\begin{tabular}{
		!{\vrule width 1pt}>{\columncolor[gray]{.9}[\tabcolsep]}c!{\vrule width 1pt}c!{\vrule width 1pt}c!{\vrule width 1pt}c!{\vrule width 1pt}c!{\vrule width 1pt}c!{\vrule width 1pt}c!{\vrule width 1pt}c!{\vrule width 1pt}c!{\vrule width 1pt}c!{\vrule width 1pt}c!{\vrule width 1pt}c!{\vrule width 1pt}}
		\toprule[1pt]
		\rowcolor[gray]{.9}[\tabcolsep] Cat. ID                & Rank & $\D^2$          & Valid     & $\rk{\mug{\C}}$ & Premodular? & Pointed?  & TY?       & TEE       & $\log\D_{\mug{\C}}^2$ & Conjecture true? \\
		\toprule[1pt]
		\rowcolor{nicegreen!10}[\tabcolsep] $\FR{5}{0}{1}{0}$  & 5    & 8               & \ding{51} & 2               & \ding{51}   &           &           & $\log 2$  & $\log 2$              & \ding{51}        \\
		$\FR{5}{0}{1}{1}$                                      & 5    & 8               & \ding{51} & 5               & Symm.       &           &           & $\log 8$  & $\log 8$              & \ding{51}        \\
		\rowcolor{nicegreen!10}[\tabcolsep] $\FR{5}{0}{1}{2}$  & 5    & 8               & \ding{51} & 2               & \ding{51}   &           &           & $\log 2$  & $\log 2$              & \ding{51}        \\
		$\FR{5}{0}{1}{3}$                                      & 5    & 8               & \ding{51} & 5               & Symm.       &           &           & $\log 8$  & $\log 8$              & \ding{51}        \\
		\rowcolor{nicegreen!10}[\tabcolsep] $\FR{5}{0}{1}{4}$  & 5    & 8               & \ding{51} & 2               & \ding{51}   &           &           & $\log 2$  & $\log 2$              & \ding{51}        \\
		$\FR{5}{0}{1}{5}$                                      & 5    & 8               & \ding{51} & 4               & \ding{51}   &           & \ding{51} & $\log 4$  & $\log 4$              & \ding{51}        \\
		\rowcolor{nicegreen!10}[\tabcolsep] $\FR{5}{0}{1}{6}$  & 5    & 8               & \ding{51} & 2               & \ding{51}   &           &           & $\log 2$  & $\log 2$              & \ding{51}        \\
		$\FR{5}{0}{1}{7}$                                      & 5    & 8               & \ding{51} & 4               & \ding{51}   &           & \ding{51} & $\log 4$  & $\log 4$              & \ding{51}        \\
		\rowcolor{nicegreen!10}[\tabcolsep] $\FR{5}{0}{1}{8}$  & 5    & 8               & \ding{51} & 2               & \ding{51}   &           &           & $\log 2$  & $\log 2$              & \ding{51}        \\
		$\FR{5}{0}{1}{9}$                                      & 5    & 8               & \ding{51} & 4               & \ding{51}   &           & \ding{51} & $\log 4$  & $\log 4$              & \ding{51}        \\
		\rowcolor{nicegreen!10}[\tabcolsep] $\FR{5}{0}{1}{10}$ & 5    & 8               & \ding{51} & 2               & \ding{51}   &           &           & $\log 2$  & $\log 2$              & \ding{51}        \\
		$\FR{5}{0}{1}{11}$                                     & 5    & 8               & \ding{51} & 4               & \ding{51}   &           & \ding{51} & $\log 4$  & $\log 4$              & \ding{51}        \\
		\rowcolor{nicegreen!10}[\tabcolsep] $\FR{5}{0}{1}{12}$ & 5    & 8               & \ding{51} & 2               & \ding{51}   &           &           & $\log 2$  & $\log 2$              & \ding{51}        \\
		$\FR{5}{0}{1}{13}$                                     & 5    & 8               & \ding{51} & 5               & Symm.       &           &           & $\log 8$  & $\log 8$              & \ding{51}        \\
		\rowcolor{nicegreen!10}[\tabcolsep] $\FR{5}{0}{1}{14}$ & 5    & 8               & \ding{51} & 2               & \ding{51}   &           &           & $\log 2$  & $\log 2$              & \ding{51}        \\
		$\FR{5}{0}{1}{15}$                                     & 5    & 8               & \ding{51} & 5               & Symm.       &           &           & $\log 8$  & $\log 8$              & \ding{51}        \\
		\rowcolor{nicegreen!10}[\tabcolsep] $\FR{5}{0}{1}{16}$ & 5    & 8               & \ding{51} & 2               & \ding{51}   &           &           & $\log 2$  & $\log 2$              & \ding{51}        \\
		$\FR{5}{0}{1}{17}$                                     & 5    & 8               & \ding{51} & 5               & Symm.       &           &           & $\log 8$  & $\log 8$              & \ding{51}        \\
		\rowcolor{nicegreen!10}[\tabcolsep] $\FR{5}{0}{1}{18}$ & 5    & 8               & \ding{51} & 2               & \ding{51}   &           &           & $\log 2$  & $\log 2$              & \ding{51}        \\
		\rowcolor{nicegreen!10}[\tabcolsep] $\FR{5}{0}{1}{19}$ & 5    & 8               & \ding{51} & 2               & \ding{51}   &           &           & $\log 2$  & $\log 2$              & \ding{51}        \\
		$\FR{5}{0}{1}{20}$                                     & 5    & 8               & \ding{51} & 4               & \ding{51}   &           & \ding{51} & $\log 4$  & $\log 4$              & \ding{51}        \\
		\rowcolor{nicegreen!10}[\tabcolsep] $\FR{5}{0}{1}{21}$ & 5    & 8               & \ding{51} & 2               & \ding{51}   &           &           & $\log 2$  & $\log 2$              & \ding{51}        \\
		\rowcolor{nicegreen!10}[\tabcolsep] $\FR{5}{0}{1}{22}$ & 5    & 8               & \ding{51} & 2               & \ding{51}   &           &           & $\log 2$  & $\log 2$              & \ding{51}        \\
		$\FR{5}{0}{1}{23}$                                     & 5    & 8               & \ding{51} & 4               & \ding{51}   &           & \ding{51} & $\log 4$  & $\log 4$              & \ding{51}        \\
		\rowcolor{nicegreen!10}[\tabcolsep] $\FR{5}{0}{1}{24}$ & 5    & 8               & \ding{51} & 2               & \ding{51}   &           &           & $\log 2$  & $\log 2$              & \ding{51}        \\
		\rowcolor{nicegreen!10}[\tabcolsep] $\FR{5}{0}{1}{25}$ & 5    & 8               & \ding{51} & 2               & \ding{51}   &           &           & $\log 2$  & $\log 2$              & \ding{51}        \\
		$\FR{5}{0}{1}{26}$                                     & 5    & 8               & \ding{51} & 5               & Symm.       &           &           & $\log 8$  & $\log 8$              & \ding{51}        \\
		\rowcolor{nicegreen!10}[\tabcolsep] $\FR{5}{0}{1}{27}$ & 5    & 8               & \ding{51} & 2               & \ding{51}   &           &           & $\log 2$  & $\log 2$              & \ding{51}        \\
		\rowcolor{nicegreen!10}[\tabcolsep] $\FR{5}{0}{1}{28}$ & 5    & 8               & \ding{51} & 2               & \ding{51}   &           &           & $\log 2$  & $\log 2$              & \ding{51}        \\
		\rowcolor{nicegreen!10}[\tabcolsep] $\FR{5}{0}{1}{29}$ & 5    & 8               & \ding{51} & 2               & \ding{51}   &           &           & $\log 2$  & $\log 2$              & \ding{51}        \\
		\rowcolor{nicegreen!10}[\tabcolsep] $\FR{5}{0}{1}{30}$ & 5    & 8               & \ding{51} & 2               & \ding{51}   &           &           & $\log 2$  & $\log 2$              & \ding{51}        \\
		\rowcolor{nicegreen!10}[\tabcolsep] $\FR{5}{0}{1}{31}$ & 5    & 8               & \ding{51} & 2               & \ding{51}   &           &           & $\log 2$  & $\log 2$              & \ding{51}        \\
		\rowcolor{nicegreen!10}[\tabcolsep] $\FR{5}{0}{1}{32}$ & 5    & 8               & \ding{51} & 2               & \ding{51}   &           &           & $\log 2$  & $\log 2$              & \ding{51}        \\
		\rowcolor{nicegreen!10}[\tabcolsep] $\FR{5}{0}{1}{33}$ & 5    & 8               & \ding{51} & 2               & \ding{51}   &           &           & $\log 2$  & $\log 2$              & \ding{51}        \\
		\rowcolor{nicegreen!10}[\tabcolsep] $\FR{5}{0}{1}{34}$ & 5    & 8               & \ding{51} & 2               & \ding{51}   &           &           & $\log 2$  & $\log 2$              & \ding{51}        \\
		\rowcolor{nicegreen!10}[\tabcolsep] $\FR{5}{0}{1}{35}$ & 5    & 8               & \ding{51} & 2               & \ding{51}   &           &           & $\log 2$  & $\log 2$              & \ding{51}        \\
		\rowcolor{nicegreen!10}[\tabcolsep] $\FR{5}{0}{1}{36}$ & 5    & 8               & \ding{51} & 2               & \ding{51}   &           &           & $\log 2$  & $\log 2$              & \ding{51}        \\
		\rowcolor{nicegreen!10}[\tabcolsep] $\FR{5}{0}{1}{37}$ & 5    & 8               & \ding{51} & 2               & \ding{51}   &           &           & $\log 2$  & $\log 2$              & \ding{51}        \\
		\rowcolor{nicegreen!10}[\tabcolsep] $\FR{5}{0}{1}{38}$ & 5    & 8               & \ding{51} & 2               & \ding{51}   &           &           & $\log 2$  & $\log 2$              & \ding{51}        \\
		\rowcolor{nicegreen!10}[\tabcolsep] $\FR{5}{0}{1}{39}$ & 5    & 8               & \ding{51} & 2               & \ding{51}   &           &           & $\log 2$  & $\log 2$              & \ding{51}        \\
		\rowcolor{nicegreen!10}[\tabcolsep] $\FR{5}{0}{1}{40}$ & 5    & 8               & \ding{51} & 2               & \ding{51}   &           &           & $\log 2$  & $\log 2$              & \ding{51}        \\
		\rowcolor{nicegreen!10}[\tabcolsep] $\FR{5}{0}{1}{41}$ & 5    & 8               & \ding{51} & 2               & \ding{51}   &           &           & $\log 2$  & $\log 2$              & \ding{51}        \\
		\rowcolor{nicegreen!10}[\tabcolsep] $\FR{5}{0}{1}{42}$ & 5    & 8               & \ding{51} & 2               & \ding{51}   &           &           & $\log 2$  & $\log 2$              & \ding{51}        \\
		\rowcolor{nicegreen!10}[\tabcolsep] $\FR{5}{0}{1}{43}$ & 5    & 8               & \ding{51} & 2               & \ding{51}   &           &           & $\log 2$  & $\log 2$              & \ding{51}        \\
		\rowcolor{nicegreen!10}[\tabcolsep] $\FR{5}{0}{1}{44}$ & 5    & 8               & \ding{51} & 2               & \ding{51}   &           &           & $\log 2$  & $\log 2$              & \ding{51}        \\
		\rowcolor{nicegreen!10}[\tabcolsep] $\FR{5}{0}{1}{45}$ & 5    & 8               & \ding{51} & 2               & \ding{51}   &           &           & $\log 2$  & $\log 2$              & \ding{51}        \\
		\rowcolor{nicegreen!10}[\tabcolsep] $\FR{5}{0}{1}{46}$ & 5    & 8               & \ding{51} & 2               & \ding{51}   &           &           & $\log 2$  & $\log 2$              & \ding{51}        \\
		\rowcolor{nicegreen!10}[\tabcolsep] $\FR{5}{0}{1}{47}$ & 5    & 8               & \ding{51} & 2               & \ding{51}   &           &           & $\log 2$  & $\log 2$              & \ding{51}        \\
		\rowcolor{nicegreen!10}[\tabcolsep] $\FR{5}{0}{1}{48}$ & 5    & 8               & \ding{51} & 2               & \ding{51}   &           &           & $\log 2$  & $\log 2$              & \ding{51}        \\
		\rowcolor{nicegreen!10}[\tabcolsep] $\FR{5}{0}{1}{49}$ & 5    & 8               & \ding{51} & 2               & \ding{51}   &           &           & $\log 2$  & $\log 2$              & \ding{51}        \\
		\rowcolor{nicegreen!10}[\tabcolsep] $\FR{5}{0}{1}{50}$ & 5    & 8               & \ding{51} & 2               & \ding{51}   &           &           & $\log 2$  & $\log 2$              & \ding{51}        \\
		\rowcolor{nicegreen!10}[\tabcolsep] $\FR{5}{0}{1}{51}$ & 5    & 8               & \ding{51} & 2               & \ding{51}   &           &           & $\log 2$  & $\log 2$              & \ding{51}        \\
		\rowcolor{nicegreen!10}[\tabcolsep] $\FR{5}{0}{1}{52}$ & 5    & 8               & \ding{51} & 2               & \ding{51}   &           &           & $\log 2$  & $\log 2$              & \ding{51}        \\
		$\FR{5}{0}{1}{53}$                                     & 5    & 8               & \ding{51} & 5               & Symm.       &           &           & $\log 8$  & $\log 8$              & \ding{51}        \\
		$\FR{5}{0}{1}{54}$                                     & 5    & 8               & \ding{51} & 4               & \ding{51}   &           & \ding{51} & $\log 4$  & $\log 4$              & \ding{51}        \\
		$\FR{5}{0}{1}{55}$                                     & 5    & 8               & \ding{51} & 4               & \ding{51}   &           & \ding{51} & $\log 4$  & $\log 4$              & \ding{51}        \\
		$\FR{5}{0}{1}{56}$                                     & 5    & 8               & \ding{51} & 5               & Symm.       &           &           & $\log 8$  & $\log 8$              & \ding{51}        \\
		\rowcolor{nicegreen!10}[\tabcolsep] $\FR{5}{0}{1}{57}$ & 5    & 8               & \ding{51} & 2               & \ding{51}   &           &           & $\log 2$  & $\log 2$              & \ding{51}        \\
		\rowcolor{nicegreen!10}[\tabcolsep] $\FR{5}{0}{1}{58}$ & 5    & 8               & \ding{51} & 2               & \ding{51}   &           &           & $\log 2$  & $\log 2$              & \ding{51}        \\
		\rowcolor{nicegreen!10}[\tabcolsep] $\FR{5}{0}{1}{59}$ & 5    & 8               & \ding{51} & 2               & \ding{51}   &           &           & $\log 2$  & $\log 2$              & \ding{51}        \\
		\rowcolor{nicegreen!10}[\tabcolsep] $\FR{5}{0}{1}{60}$ & 5    & 8               & \ding{51} & 2               & \ding{51}   &           &           & $\log 2$  & $\log 2$              & \ding{51}        \\
		\rowcolor{nicegreen!10}[\tabcolsep] $\FR{5}{0}{1}{61}$ & 5    & 8               & \ding{51} & 2               & \ding{51}   &           &           & $\log 2$  & $\log 2$              & \ding{51}        \\
		\rowcolor{nicegreen!10}[\tabcolsep] $\FR{5}{0}{1}{62}$ & 5    & 8               & \ding{51} & 2               & \ding{51}   &           &           & $\log 2$  & $\log 2$              & \ding{51}        \\
		\rowcolor{nicegreen!10}[\tabcolsep] $\FR{5}{0}{1}{63}$ & 5    & 8               & \ding{51} & 2               & \ding{51}   &           &           & $\log 2$  & $\log 2$              & \ding{51}        \\
		$\FR{5}{0}{3}{0}$                                      & 5    & 12              & \ding{51} & 1               & Mod.        &           &           & 0         & 0                     & \ding{51}        \\
		$\FR{5}{0}{3}{1}$                                      & 5    & 12              & \ding{51} & 1               & Mod.        &           &           & 0         & 0                     & \ding{51}        \\
		$\FR{5}{0}{3}{2}$                                      & 5    & 12              & \ding{51} & 1               & Mod.        &           &           & 0         & 0                     & \ding{51}        \\
		$\FR{5}{0}{3}{3}$                                      & 5    & 12              & \ding{51} & 1               & Mod.        &           &           & 0         & 0                     & \ding{51}        \\
		$\FR{5}{0}{3}{4}$                                      & 5    & 12              & \ding{51} & 1               & Mod.        &           &           & 0         & 0                     & \ding{51}        \\
		$\FR{5}{0}{3}{5}$                                      & 5    & 12              & \ding{51} & 1               & Mod.        &           &           & 0         & 0                     & \ding{51}        \\
		$\FR{5}{0}{3}{6}$                                      & 5    & 12              & \ding{51} & 1               & Mod.        &           &           & 0         & 0                     & \ding{51}        \\
		$\FR{5}{0}{3}{7}$                                      & 5    & 12              & \ding{51} & 1               & Mod.        &           &           & 0         & 0                     & \ding{51}        \\
		$\FR{5}{0}{4}{0}$                                      & 5    & 14              & \ding{51} & 5               & Symm.       &           &           & $\log 14$ & $\log 14$             & \ding{51}        \\
		\rowcolor{nicegreen!10}[\tabcolsep] $\FR{5}{0}{4}{1}$  & 5    & 14              & \ding{51} & 2               & \ding{51}   &           &           & $\log 2$  & $\log 2$              & \ding{51}        \\
		\rowcolor{nicegreen!10}[\tabcolsep] $\FR{5}{0}{4}{2}$  & 5    & 14              & \ding{51} & 2               & \ding{51}   &           &           & $\log 2$  & $\log 2$              & \ding{51}        \\
		\rowcolor{nicegreen!10}[\tabcolsep] $\FR{5}{0}{4}{3}$  & 5    & 14              & \ding{51} & 2               & \ding{51}   &           &           & $\log 2$  & $\log 2$              & \ding{51}        \\
		\rowcolor{nicegreen!10}[\tabcolsep] $\FR{5}{0}{4}{4}$  & 5    & 14              & \ding{51} & 2               & \ding{51}   &           &           & $\log 2$  & $\log 2$              & \ding{51}        \\
		\rowcolor{nicegreen!10}[\tabcolsep] $\FR{5}{0}{4}{5}$  & 5    & 14              & \ding{51} & 2               & \ding{51}   &           &           & $\log 2$  & $\log 2$              & \ding{51}        \\
		\rowcolor{nicegreen!10}[\tabcolsep] $\FR{5}{0}{4}{6}$  & 5    & 14              & \ding{51} & 2               & \ding{51}   &           &           & $\log 2$  & $\log 2$              & \ding{51}        \\
		$\FR{5}{0}{6}{0}$                                      & 5    & 24              & \ding{51} & 5               & Symm.       &           &           & $\log 24$ & $\log 24$             & \ding{51}        \\
		\rowcolor{nicegreen!10}[\tabcolsep] $\FR{5}{0}{6}{1}$  & 5    & 24              & \ding{51} & 3               & \ding{51}   &           &           & $\log 6$  & $\log 6$              & \ding{51}        \\
		$\FR{5}{0}{6}{2}$                                      & 5    & 24              & \ding{51} & 5               & Symm.       &           &           & $\log 24$ & $\log 24$             & \ding{51}        \\
		\rowcolor{nicegreen!10}[\tabcolsep] $\FR{5}{0}{6}{3}$  & 5    & 24              & \ding{51} & 3               & \ding{51}   &           &           & $\log 6$  & $\log 6$              & \ding{51}        \\
		\rowcolor{nicegreen!10}[\tabcolsep] $\FR{5}{0}{7}{0}$  & 5    & 5 $\sqrt{5}+15$ & \ding{51} & 2               & \ding{51}   &           &           & $\log 2$  & $\log 2$              & \ding{51}        \\
		\rowcolor{nicegreen!10}[\tabcolsep] $\FR{5}{0}{7}{1}$  & 5    & 5 $\sqrt{5}+15$ & \ding{51} & 2               & \ding{51}   &           &           & $\log 2$  & $\log 2$              & \ding{51}        \\
		$\FR{5}{0}{10}{0}$                                     & 5    & $\sim34.65$     & \ding{51} & 1               & Mod.        &           &           & 0         & 0                     & \ding{51}        \\
		$\FR{5}{0}{10}{1}$                                     & 5    & $\sim34.65$     & \ding{51} & 1               & Mod.        &           &           & 0         & 0                     & \ding{51}        \\
		$\FR{5}{4}{1}{0}$                                      & 5    & 5               & \ding{51} & 5               & Symm.       & \ding{51} &           & $\log 5$  & $\log 5$              & \ding{51}        \\
		$\FR{5}{4}{1}{1}$                                      & 5    & 5               & \ding{51} & 1               & Mod.        & \ding{51} &           & 0         & 0                     & \ding{51}        \\
		$\FR{5}{4}{1}{2}$                                      & 5    & 5               & \ding{51} & 1               & Mod.        & \ding{51} &           & 0         & 0                     & \ding{51}        \\
		$\FR{5}{4}{1}{3}$                                      & 5    & 5               & \ding{51} & 1               & Mod.        & \ding{51} &           & 0         & 0                     & \ding{51}        \\
		$\FR{5}{4}{1}{4}$                                      & 5    & 5               & \ding{51} & 1               & Mod.        & \ding{51} &           & 0         & 0                     & \ding{51}        \\
		\toprule[1pt]
	\end{tabular}
}

\end{document}